\documentclass[bj,authoryear,noshowframe]{imsart}
\usepackage{amsmath, amssymb, amsthm}
\usepackage{bbm}
\usepackage[dvipsnames]{xcolor}
\usepackage{dsfont}
\usepackage{tikz}
\usepackage{graphicx} 
\usepackage{caption}
\usepackage{mathbbol}
\usepackage{algorithm}
\usepackage{algpseudocode}
\usepackage{amsmath}
\usepackage{amssymb}
\usepackage{natbib}
\usepackage{accents}
\newcommand{\ubar}[1]{\underaccent{\bar}{#1}}

\usepackage{tikz}
\usetikzlibrary{positioning}

\newcommand{\AMMD}[2]{\mathbb{A}\mathbb{D}\left(#1,#2\right)}
\newcommand{\AMMDn}[2]{\mathbb{A}\mathbb{D}_n\left(#1,#2\right)}
\newcommand{\AMMDtwo}[2]{\mathbb{A}\mathbb{D}^2\left(#1,#2\right)}
\newcommand{\AMMDtwon}[2]{\mathbb{A}\mathbb{D}_n^2\left(#1,#2\right)}

\startlocaldefs
\theoremstyle{plain}
\newtheorem{thm}{Theorem}[section]
\newtheorem{prop}{Proposition}[section]
\newtheorem{lemma}[thm]{Lemma}
\newtheorem{cor}[thm]{Corollary}

\theoremstyle{definition}
\newtheorem{dfn}{Definition}[section]
\newtheorem{asm}{Assumption}[section]

\newtheorem{exm}{Example}[section]

\newcommand{\INPUT}[1]{\State \textbf{Input:} #1}
\newcommand{\R}{\ensuremath{{\mathbb R}}}

\newcommand{\Pjoint}{\ensuremath{{\mathbb{P}_{X,M}}}}
\newcommand{\PjointMCAR}{\ensuremath{{\mathbb{P}_{X} \times \mathbb{P}_{M}}}}
\newcommand{\PX}{\ensuremath{{\mathbb{P}_X}}}
\newcommand{\PmargmX}{\ensuremath{{\mathbb{P}^{(m)}_X}}}
\newcommand{\PmargMX}{\ensuremath{{\mathbb{P}^{(M)}_X}}}
\newcommand{\PXM}{\ensuremath{{\mathbb{P}_{X\mid M}}}}
\newcommand{\PXMm}{\ensuremath{{\mathbb{P}_{X\mid M=m}}}}
\newcommand{\PmargMXM}{\ensuremath{{\mathbb{P}_{X\mid M}^{(M)}}}}
\newcommand{\PmargmXMm}{\ensuremath{{\mathbb{P}_{X\mid M=m}^{(m)}}}}
\newcommand{\PM}{\ensuremath{{\mathbb{P}_{M}}}}
\newcommand{\pjoint}{\ensuremath{{p_{X,M}}}}
\newcommand{\pX}{\ensuremath{{p_X}}}
\newcommand{\pmargmX}{\ensuremath{{p^{(m)}_{X}}}}

\newcommand{\pXMm}{\ensuremath{{p_{X\mid M=m}}}}
\newcommand{\pmargmXMm}{\ensuremath{{p^{(m)}_{X\mid M=m}}}}

\newcommand{\PMm}{\ensuremath{{\mathbb{P}[M=m]}}}
\newcommand{\PMmX}{\ensuremath{{\mathbb{P}[M=m|X]}}}
\newcommand{\PMmx}{\ensuremath{{\mathbb{P}[M=m|X=x]}}}

\newcommand{\PMdX}{\ensuremath{{\mathbb{P}_{M\mid X}}}}

\newcommand{\Ptheta}{\ensuremath{{P_{\theta}}}}
\newcommand{\Pmargmtheta}{\ensuremath{{P^{(m)}_{\theta}}}}
\newcommand{\PmargMtheta}{\ensuremath{{P^{(M)}_{\theta}}}}
\newcommand{\Pthetastar}{\ensuremath{{P_{\theta^*}}}}

\newcommand{\PmargMthetastar}{\ensuremath{{P^{(M)}_{\theta^*}}}}
\newcommand{\ptheta}{\ensuremath{{p_{\theta}}}}
\newcommand{\pmargmtheta}{\ensuremath{{p^{(m)}_{\theta}}}}

\newcommand{\PmargMitheta}{\ensuremath{{P^{(M_i)}_{\theta}}}}

\newcommand{\QX}{\ensuremath{{\mathbb{Q}_X}}}

\newcommand{\QMmx}{\ensuremath{{\mathbb{Q}[M=m|X=x]}}}
\newcommand{\E}{\ensuremath{{\mathbb E}}}
\newcommand{\Prob}{\ensuremath{{\mathbb P}}}

\newcommand{\Var}{\ensuremath{{\mathbb V}}}

\newcommand*{\MLEn}{\ensuremath{\theta_n^{\textnormal{ML}}}}
\newcommand*{\MMDn}{\ensuremath{\theta_n^{\textnormal{MMD}}}}
\newcommand*{\MLEinf}{\ensuremath{\theta_{\infty}^{\textnormal{ML}}}}
\newcommand*{\MMDinf}{\ensuremath{\theta_{\infty}^{\textnormal{MMD}}}}

\newcommand*{\dataeps}{\ensuremath{\delta}}
\newcommand*{\missingeps}{\ensuremath{\varepsilon}}

\newcommand{\KL}[2]{\textnormal{KL}\left(#1\|#2\right)}
\newcommand{\MMD}[2]{\mathbb{D}\left(#1,#2\right)}
\newcommand{\MMDtwo}[2]{\mathbb{D}^2\left(#1,#2\right)}

\newcommand{\MMDm}[3]{\mathbb{D}^{#1}\left(#2, #3 \right)}

\newcommand{\Ind}[1]{\mathbb{1}\left\{#1\right\}}

\renewcommand{\H}{\ensuremath{{\mathcal H}}}

\DeclareMathOperator*{\argmax}{arg\,max}
\DeclareMathOperator*{\argmin}{arg\,min}

\makeatletter
\newcommand{\ostar}{\mathbin{\mathpalette\make@circled\star}}
\newcommand{\make@circled}[2]{%
  \ooalign{$\m@th#1\smallbigcirc{#1}$\cr\hidewidth$\m@th#1#2$\hidewidth\cr}%
}
\newcommand{\smallbigcirc}[1]{%
  \vcenter{\hbox{\scalebox{0.77778}{$\m@th#1\bigcirc$}}}%
}
\makeatother

\endlocaldefs

\begin{document}

\begin{frontmatter}
\title{Parametric MMD Estimation with Missing Values: Robustness to Missingness and Data Model Misspecification}
\runtitle{MMD Estimation with Missing Values}

\begin{aug}
\author[A]{\inits{B-E. C-A.}\fnms{Badr-Eddine}~\snm{Chérief-Abdellatif}\ead[label=e1]{badr-eddine.cherief-abdellatif@cnrs.fr}}
\author[B]{\inits{J.N}\fnms{Jeffrey}~\snm{Näf}\ead[label=e2]{jeffrey.naf@unige.ch}}
\address[A]{CNRS, LPSM, Sorbonne Université, Université Paris Cité, 4 Place Jussieu, 75005 Paris, France\printead[presep={,\ }]{e1}}

\address[B]{Research Institute for Statistics and Information Science, University of Geneva, Bd du Pont-d'Arve 40, 1205 Genève, Switzerland\printead[presep={,\ }]{e2}}
\end{aug}

\begin{abstract}
In the missing data literature, the Maximum Likelihood Estimator (MLE) is celebrated for its ignorability property under missing at random (MAR) data. However, its sensitivity to misspecification of the (complete) data model, even under MAR, remains a significant limitation.  This issue is further exacerbated by the fact that the MAR assumption may not always be realistic, introducing an additional source of potential misspecification through the missingness mechanism. To address this, we propose a novel M-estimation procedure based on the Maximum Mean Discrepancy (MMD), which is provably robust to both model misspecification and deviations from the assumed missingness mechanism. Our approach offers strong theoretical guarantees and improved reliability in complex settings. We establish the consistency and asymptotic normality of the estimator under missing completely at random (MCAR), provide an efficient stochastic gradient descent algorithm, and derive error bounds that explicitly separate the contributions of model misspecification and missingness bias. Furthermore, we analyze missing not at random (MNAR) scenarios where our estimator maintains controlled error, including a Huber setting where both the missingness mechanism and the data model are contaminated. Our contributions refine the understanding of the limitations of the MLE and provide a robust and principled alternative for handling missing data.
\end{abstract}


\end{frontmatter}


\section{Introduction}

Missing values present a persistent challenge across various fields. When data contain missing values, some dimensions of the observations are masked, requiring consideration not only of the data model but also of the missingness mechanism. Due to the pervasiveness of this issue, numerous approaches have been developed to obtain reliable estimates despite missing data. These include imputation methods (\cite{FCS_Van_Buuren2007, VANBUUREN2018, näf2024goodimputationmarmissingness}), weighted estimators (\cite{inverseweightingoverview, Shpitser_2016, MARinverseweighting, CANTONI2020, Malinsky2022}), and likelihood-based procedures (\cite{Rubin_Inferenceandmissing, Rubin_Imputationafter18, MLEConsistencyunderMissing, MLEMisspecificationunderMissing, Sportisse2024}, among others). 
Among these, likelihood-based methods - particularly Maximum Likelihood Estimation (MLE) - are widely used, as they provide consistent estimates under the missing at random (MAR) condition. Introduced by \cite{Rubin_Inferenceandmissing}, MAR roughly states that the probability of missingness depends only on observed data. Rubin further showed that, under an additional parameter distinctness assumption, maximizing the likelihood while ignoring the missingness mechanism remains valid.
This result established MAR as a standard assumption in missing data analysis, as evidenced by the extensive literature discussing this condition (\cite{whatismeant1, whatismeant2, whatismeant3, näf2024goodimputationmarmissingness}). Beyond MAR, a stronger assumption is missing completely at random (MCAR), where the probability of missingness is entirely independent of the data. Conversely, when the missingness mechanism depends on unobserved information, it falls under missing not at random (MNAR), which is generally more challenging to handle.

\vspace{0.2cm}
Despite its popularity, the MAR condition has been questioned, both due to its potential lack of realism (see e.g., \cite{Robins1997_ingorable, Malinsky2022, directcompetitor0} and references therein) and because it is challenging to handle outside of MLE. For instance, \cite{MNARcontamination} show that MAR can lead to parameter non-identifiability even in simple cases, while \cite{näf2024goodimputationmarmissingness} highlight the complex distribution shifts that can arise under MAR. These issues cast doubt on MAR as a default assumption for missing data. As an alternative, \cite{MNARcontamination} propose a framework in which the missingness mechanism is assumed to be MCAR, but with Huber-style MNAR deviations affecting at most a fraction $\missingeps$ of the data. This approach is particularly relevant given that the missingness mechanism is typically unknown. Indeed, since MAR is fundamentally untestable (\cite{Robins1997_ingorable, ourresult}), identifying the true missing data mechanism without strong domain knowledge is often infeasible. More generally, the framework of \cite{MNARcontamination} allows for contamination in both the assumed data model and the missingness mechanism. This flexibility is crucial, as MLE lacks robustness against model misspecification and contamination. These issues become even more problematic in the presence of missing data, where misspecifications can stem not only from the data model but also from an incorrectly specified missingness mechanism.

\vspace{0.2cm}
A promising alternative to MLE is Maximum Mean Discrepancy (MMD) estimation, which selects parameters by minimizing the MMD between the postulated and observed distributions, rather than relying on the Kullback-Leibler divergence as in MLE. This approach is known to exhibit desirable robustness properties (\cite{briol2019statistical, BadrAlquierMMD}).
Building on the idea of safeguarding against deviations in both the assumed data distribution and the missingness mechanism, we extend MMD estimation to the case of missing values. Our first step is to reformulate the MMD estimator as an M-estimator, which allows us to adapt it to the more challenging setting of missing data. We then establish that the resulting estimator is consistent and asymptotically normal under MCAR, while remaining robust to certain deviations from both the assumed data model and the missingness mechanism. Specifically, we show that under regularity conditions, the MMD estimator converges to the parameter associated with the closest distribution (in MMD) to the true conditional within the model class. We derive explicit bounds on this MMD distance, which split into two components: one corresponding to errors from model misspecification and the other accounting for deviations from the MCAR assumption in the missingness mechanism. 
To facilitate practical implementation, we extend an existing stochastic gradient descent (SGD) algorithm for MMD estimation to handle missing values. This allows for a straightforward adaptation of the algorithm to a wide range of parametric models, maintaining its generality, flexibility, and computational efficiency while providing rigorous robustness guarantees. \textcolor{black}{In the special case of the Gaussian distribution with fixed covariance matrix, this leads to a robust estimator of the mean under missing values. While estimation of a Gaussian mean under contamination has been the subject of active research in the last decade with complete data (see e.g., \citet{diakonikolas2019recent,loh2024theoretical,anderson2025robust} and the references therein for an overview), much less attention has been paid to the problem with missing data. Two recent exceptions are \citet{Meanestimationmissing1} who study the problem theoretically in great detail and \citet{MNARcontamination} who improve on their results.}


\vspace{0.2cm}
The paper is organized as follows. After discussing related work and notation, we first give some background on MMD and maximum likelihood estimation and motivate our approach in Section \ref{Sec_MLEdiscussion}. Section \ref{Sec_EstimatorIntro} then presents our estimator and its asymptotic properties. Section \ref{Sec_robust} discusses the general robustness guarantees of the new estimator. Finally, Sections \ref{Sec_exm} and \ref{Sec_Simulation} provide several examples. Code to replicate the examples can be found on \url{https://github.com/JeffNaef/MissingnessMMD}. 

\subsection{Related Work}

This paper considers an estimation approach that accounts for deviations from both the data model and the assumed MCAR missingness mechanism. To analyze robustness under such deviations, we derive general bounds on the discrepancy between the estimated and true distributions in terms of MMD. A related contamination framework is studied in \cite{MNARcontamination}, where the entire joint distribution is affected: a fraction $\missingeps$ of points follows a different data distribution and is subject to a different missingness mechanism. In contrast, we allow for separate contamination of the data model and the missingness mechanism, with potentially different levels of contamination. Furthermore, while \cite{MNARcontamination} develops robust estimators for means and regression parameters, our approach provides a more general framework applicable to any parametric model, with provable robustness guarantees.

\vspace{0.2cm}
A key component of our methodology is the reformulation of the MMD estimator as an M-estimator, linking our work to classical M-estimation under missing data. In particular, \cite{Mestimatormissingvalues} demonstrate that M-estimators generally fail to be consistent when the MCAR assumption is violated, leading to the development of weighted estimators that ensure consistency under more general missingness mechanisms (\cite{inverseweightingoverview, Shpitser_2016, MARinverseweighting, CANTONI2020, Malinsky2022} among others). Notably, \cite{Malinsky2022} propose a weighted M-estimation procedure that remains consistent and asymptotically normal under a no-self-censoring assumption. However, these approaches rely on explicit modeling of the missingness mechanism, typically through parametric models, and do not explicitly address robustness to misspecifications in this mechanism. In contrast, our approach ensures consistency under MCAR while providing robustness guarantees against deviations from the assumed missingness mechanism, thereby avoiding the need to model this mechanism explicitly. Nonetheless, our method could be combined with such reweighting strategies to construct estimators that are both robust and consistent beyond MCAR.

\vspace{0.2cm}
The central statistical tool in our framework is the Maximum Mean Discrepancy (MMD) \citep{gretton2007kernel, gretton2012kernel}, a distance measure that has gained increasing attention in parametric and nonparametric estimation. Initially introduced for hypothesis testing, MMD has since been established as a general estimation principle by \cite{briol2019statistical} and \cite{BadrAlquierMMD}, who showed that MMD-based estimators are consistent, robust to model misspecification, and computationally tractable. These properties make MMD particularly well-suited for scenarios involving complex data distributions and deviations from assumed models. The application of MMD estimation has been extended to various settings, including generalized Bayesian inference \citep{cherief2020mmd, pacchiardi2024generalized,frazier2024impact,shen2024prediction}, Approximate Bayesian Computation (ABC) \citep{park2016k2, mitrovic2016dr, kajihara2018kernel, bharti2021general, legramanti2025concentration}, Bayesian nonparametric learning \citep{dellaporta2022robust}, generative adversarial networks \citep{dziugaite2015training, li2015generative, sutherland2016generative, li2017mmd, binkowski2018demystifying}, quantization \citep{teymur2021optimal}, copula estimation \citep{AlquierCopulas}, label shift quantification \citep{zhang2013domain, iyer2014maximum, dussap2023label}, and regression \citep{AlquierGerberRegression}, among others.

\vspace{0.2cm}
Our work extends this line of research by applying MMD-based estimation to the context of missing data, demonstrating that it provides a natural and robust alternative to traditional methods while avoiding explicit specification of the missingness mechanism. \textcolor{black}{Although MMD-based hypothesis testing has recently been extended to missing data scenarios \citep{TwosampleTestingMMD, zeng2024mmdtwosampletestingpresence}, to the best of our knowledge, this is the first work to address this problem within the framework of robust estimation.}

\subsection{Notation}




In missing data analysis, the underlying random vector $X$ with values in $\R^d$ gets masked by another random vector $M$ taking values in $\{0,1\}^d$, where $M_j=0$ indicates that $X_j$ is observed while $M_j=1$ means that $X_j$ is missing. Under i.i.d. data sampling of $(X_1, M_1), \ldots ,(X_n, M_n)$, this induces at most $2^d$ possible missingness patterns, with each observation being associated with a specific pattern. Each observation from a pattern $m$ can be seen as a masked sample from the distribution $X$ conditional on $M=m$, leading to potentially different distributions in each pattern, exacerbating the information loss through the non-observed values. We consider the following notation:

\begin{itemize}
    \item $\PX$ is the marginal distribution of the complete data variable $X$, which is assumed to be absolutely continuous with respect to Lebesgue's measure with density $\pX$.
    \item $\PM$ is the marginal distribution of the mask variable $M$. Given its discrete nature, we introduce the probability mass function $\PMm$ for a given pattern $m\in\{0,1\}^d$ so that for every set $A\subset \{0,1\}^d$, $\PM[A]=\sum_{m \in A} \PMm$.
    \item $\Pjoint$ is the joint distribution of $(X,M)$. Since $\PX$ is assumed to be absolutely continuous, $\Pjoint$ is absolutely continuous with respect to the product of Lebesgue's and counting measures, with joint density $\pjoint$. 
    \item We denote $\PXMm$ the conditional distribution of $X$ given $M=m$, which is defined by its density w.r.t.\ Lebesgue's measure (provided that $\PMm\ne0$):
    $$
    \pXMm(x) = \frac{\pjoint(x,m)}{\PMm} .
    $$
    The corresponding conditional Markov kernel of $X$ given $M$ is then defined as $\PXM$. 
    \item Similarly, we define the conditional distribution of $M$ given $X=x$ by its probability mass function:
    $$
    \PMmx = \frac{\pjoint(x,m)}{\pX(x)} .
    $$
    The corresponding conditional Markov kernel of $M$ given $X$ is then defined as $\PMdX$. This is what we  refer to as the \emph{missingness mechanism}.
    \item For two distributions $P_1, P_2$ with densities $p_1, p_2$, we define the Kullback–Leibler divergence: 
    \begin{align}
    \KL{P_1}{P_2} = \begin{cases}
        \int \log \left( \frac{p_1(x)}{p_2(x)} \right) p_1(x) \mathrm{d}x , \textnormal{ if }  \int \Ind{p_2(x)=0} p_1(x) \mathrm{d}x=0\\
        \infty, \textnormal{ else.}
    \end{cases}
\end{align}
    \item We finally introduce a complete data model $\{\Ptheta\}_\theta$ with density $\ptheta$ w.r.t.\ Lebesgue's measure. The model is well-specified when there is a true parameter $\theta^*$ such that $\PX=\Pthetastar$.
\end{itemize}

For $m \in \{0,1\}^d$, $X^{(m)}$ is the subvector of $X$ corresponding to the variables such that $m_j=0$, and $\PmargmX/ \pmargmX$, $\PmargmXMm/\pmargmXMm$, $\Pmargmtheta/\pmargmtheta$ are the corresponding distributions/densities. We call a missingness mechanism missing completely at random (MCAR) if $\Pjoint=\PjointMCAR$. A missingness mechanism is called missing at random (MAR) if for all $m$ in the support of $\PM$, $\PMmx=\Prob[M=m \mid X^{(m)}=x^{(m)}]$ for $\PX$-almost every $x$. \textcolor{black}{Note that for all the expectations in $M$, the "empty pattern" is (implicitly) excluded, i.e.
$$
\mathbb{E}_{M}\left[f(M)\right] = \mathbb{E}_{M}\left[\mathbbm{1}(M\ne\{1\}^d)f(M)\right] = \sum_{m\ne\{1\}^d} \pi_m f(m).
$$
The indicator will be omitted for ease of presentation.}

\vspace{0.2cm}

Finally, in Section \ref{Sec_robust}, we will use the shortened notation $\pi_m(X)=\PMmX$ for the missingness mechanism and $\pi_m=\PMm=\E_{X\sim\PX}\left[\pi_m(X)\right]$ for the marginal probability. 


\subsection{Contributions}


Our contributions may be summarized as follows:

\begin{enumerate}
    \item We adapt the MMD estimator to general missing value patterns by recasting it as an M-estimator and show that the new estimator is consistent and asymptotically normal with limit $\MMDinf$ corresponding to the best approximation to the data distribution in terms of MMD. This result holds without any assumption on the missingness mechanism.
    \item We show that, surprisingly, the MLE exhibits some robustness against deviations from M(C)AR in one dimension. However, it is not clear how to extend this to higher dimensions, and this robustness does not hold under contamination of the data generating distribution. Instead, we derive robustness guarantees for the MMD estimator, by deriving bounds on the MMD distance between $\PX$ and the best approximation from the model class. These bounds factor into a component coming from the deviations from MCAR and a component arising from the deviations from the proposed model class (Section \ref{Sec_general_general}).
    \item We refine these bounds by applying them to several specific settings, utilizing an adaptation of the Huber-style contamination in \cite{MNARcontamination} to the missingness mechanism:
    \begin{itemize}
    \item Huber contamination of the data distribution under a general MNAR
    mechanism (Section~\ref{Sec_Huber_general}).
    \item Huber contamination of both the data distribution and the missingness
    mechanism (Section~\ref{Sec_Huber_Huber}).
    \item A truncation missingness mechanism under a well-specified model (Section~\ref{Sec_wellspecified_Truncation}).
    \item Huber contamination of the data distribution combined with an adversarial contamination of the missingness mechanism (Section~\ref{Sec_AdversarialCont}).
\end{itemize}
\end{enumerate}

Figure \ref{fig:resultssummary} summarizes the results. Moreover, we complement the results with finite sample bounds in Appendix \ref{app_finite_sample}, including a finite sample bound on the problem of Gaussian mean estimation under MNAR. In addition to these technical contributions, we also derive a simple SGD estimation method of the new estimator and showcase its performance in simulations. 

\begin{figure}
    \centering
    \scalebox{0.9}{%
\begin{tikzpicture}[
  node distance = 1.2cm and 0.6cm,
  box/.style = {
    draw,
    rounded corners=4pt,
    align=center,
    font=\small,
    inner sep=8pt,
    minimum width=4.8cm
  },
  edge/.style = {
    ->,
    >=stealth,
    thin
  }
]
 
\node[box] (root) {%
  \textbf{General results}\\
  {\footnotesize Section \ref{Sec_general_general}}%
};
 
\node[box, below left=0.6cm and -1cm of root] (n42) {%
  \textbf{(D) Huber contamination}\\
  \textbf{(M) General MNAR}\\
  {\footnotesize Section \ref{Sec_Huber_general}}%
};
 
\node[box, below right=0.6cm and -1cm of root] (n52) {%
  \textbf{(D) Well-specified}\\
  \textbf{(M) Truncation MNAR}\\
  {\footnotesize Section \ref{Sec_wellspecified_Truncation}}%
};
 
\node[box, below left=0.6cm and -1cm of n42] (n51) {%
  \textbf{(D) Huber contamination}\\
  \textbf{(M) Huber contamination}\\
  {\footnotesize Section \ref{Sec_Huber_Huber}}%
};
 
\node[box, below right=0.6cm and -1cm of n42] (n53) {%
  \textbf{(D) Huber contamination}\\
  \textbf{(M) Adversarial contamination}\\
  {\footnotesize Section \ref{Sec_AdversarialCont}}%
};
 
\draw[edge] (root) -- (n42);
\draw[edge] (root) -- (n52);
\draw[edge] (n42)  -- (n51);
\draw[edge] (n42)  -- (n53);
 
\end{tikzpicture}%
}

    \caption{Overview of the results in the paper. ``D'' describes the data generating mechanism ($\PX$), while ``M'' corresponds to the missingness mechanism ($\Prob(M=m\mid X=x)$). In the latter case, ``Huber contamination'' refers to the contamination framework introduced in \cite{MNARcontamination}.}
    \label{fig:resultssummary}
\end{figure}

\section{Background}\label{Sec_MLEdiscussion}

We present in this section all the required background to understand the rest of the paper. We first remind the definition of the MMD and introduce related concepts. Then, we properly define the minimum distance estimator based on the MMD in the complete-data setting, that we shortly compare to the MLE. Finally, we provide a brief survey of the literature on the MLE in the presence of missing data.

\subsection{Maximum Mean Discrepancy}

The Maximum Mean Discrepancy (MMD) relies on the theory of reproducing kernels. Consider a positive definite kernel \( k: \mathcal{X} \times \mathcal{X} \to \mathbb{R} \), which is a symmetric function satisfying the condition that, for any integer \( N \geq 1 \), any set of points \( x_1, \dots, x_N \in \mathcal{X} \), and any real coefficients \( c_1, \dots, c_N \), we have  
\[
\sum_{i=1}^{N} \sum_{j=1}^{N} c_i c_j k(x_i, x_j) \geq 0.
\]
We assume in this work that the kernel is bounded (say by $1$ without loss of generality), that is $|k(x,y)|\leq 1$ for any $x,y\in\mathcal{X}$. Associated with such a kernel is a Reproducing Kernel Hilbert Space (RKHS), denoted \( (\mathcal{H}, \langle \cdot, \cdot \rangle_{\mathcal{H}}) \), which satisfies the fundamental reproducing property: for any function \( f \in \mathcal{H} \) and any \( x \in \mathcal{X} \),  
\[
f(x) = \langle f, k(x, \cdot) \rangle_{\mathcal{H}}.
\]
We refer the interested reader to \cite[Chapter 2.7]{hilbertspacebook} for a more detailed introduction to RKHS theory. 

\vspace{0.2cm}
A key concept in kernel methods is the kernel mean embedding, which provides a way to represent probability measures as elements of the RKHS. Given a probability distribution \( P \) over $\mathcal{X}$, its mean embedding is defined as
\[
\Phi(P) := \mathbb{E}_{X \sim P} [k(X, \cdot)] \in \mathcal{H}.
\]
This embedding generalizes the feature map $\Phi(X) := k(X, \cdot)$ used in kernel-based learning algorithms such as Support Vector Machines. A major advantage of this approach is that expectations with respect to \( P \) can be expressed as inner products in \( \mathcal{H} \), i.e.,  
\[
\mathbb{E}_{X \sim P} [f(X)] = \langle f, \Phi(P) \rangle_{\mathcal{H}}, \quad \forall f \in \mathcal{H}.
\]
Using these embeddings, one can define a discrepancy measure between probability distributions known as the Maximum Mean Discrepancy (MMD). For two distributions \( P_1 \) and \( P_2 \), the MMD is given by  
\[
\MMD{P_1}{P_2}=\| \Phi(P_1) - \Phi(P_2) \|_{\H}.
\]
Since $k$ is bounded and thus integrable, $\MMD{P_1}{P_2}$ can be rewritten explicitly as  
\begin{equation}
\label{alt_MMD}
\MMDtwo{P_1}{P_2} = \mathbb{E}_{X, X' \sim P_1} [k(X, X')] + \mathbb{E}_{Y, Y' \sim P_2} [k(Y, Y')] - 2 \mathbb{E}_{X \sim P_1, Y \sim P_2} [k(X, Y)].
\end{equation}
A kernel \( k \) is called characteristic if the mapping \( P \mapsto \Phi(P) \) is injective, ensuring that \( \MMD{P_1}{P_2} = 0 \) if and only if \( P_1 = P_2 \), then providing a proper metric. Many conditions ensuring that a kernel is characteristic are discussed in Section 3.3.1 of \cite{kernelmeanembeddingreview}, including examples such as the widely used Gaussian kernel:
\begin{align}\label{Gausskern}
k(x, y) = \exp\left(-\frac{\|x - y\|_2^2}{\gamma^2} \right) .
\end{align}
For the remainder of this work, we assume that \( k \) is characteristic and dimension-independent, meaning the same function can be applied to $x^{(m)}$ and $y^{(m)}$ on $\R^{|m|}$ for any $m\in\{0,1\}^d \setminus \{1\}^d$:
\begin{dfn}
A selection of kernels $(k^{(m)})_{m \in \{0,1\}^d \setminus \{1\}^d}$ is dimension-independent, if for all $m \in \{0,1\}^d \setminus \{1\}^d$,
\[
k^{(m)}(x^{(m)},y^{(m)} )=k(x^{(m)},y^{(m)} ).
\]
In this case, we identify $(k^{(m)})_{m \in \{0,1\}^d \setminus \{1\}^d}$ with the kernel $k$ and call $k$ dimension-independent.
\end{dfn}


The Gaussian kernel is dimension-independent. On the other hand, the collection $(k^{(m)})_{m \in \{0,1\}^d \setminus \{1\}^d}$,
\begin{align*}
k^{(m)}(x^{(m)}, y^{(m)}) = \exp\left(-(x^{(m)} - y^{(m)})^{\top} \Gamma_{m,m}^{-1} (x^{(m)} - y^{(m)}) \right),
\end{align*}
where $\Gamma$ is a diagonal matrix with non-equal diagonal elements and
\[
\Gamma_{m,m}=(\Gamma_{i,j})_{i: m_i=0, j: m_j=0},
\]
is not dimension-independent.

We note that our results also hold if we use a set of kernels $(k^{(m)})_{m \in \{0,1\}^d \setminus \{1\}^d}$, but we impose a dimension-independent kernel for simplicity. Moreover, $\mathcal{H}$, $\Phi$ and $\MMD{\cdot}{\cdot}$ all implicitly depend on the kernel $k$.



\subsection{MMD estimation with no missing data}


We start by revisiting the case of fully observed data. For a complete-data model $\{\Ptheta\}_\theta$, we recall the general definition of the MMD estimator as proposed by \cite{briol2019statistical,BadrAlquierMMD}. Let $X_1,\cdots,X_n$ be a collection of i.i.d.\ random variables following an unknown data generating process $\PX$, and $P_n = \frac{1}{n} \sum_{i=1}^n \delta_{\{X_i\}}$ be the associated empirical measure. We then define the MMD estimator as the following minimum distance estimator:
\begin{align}\label{MMD_Estimator_full_data}
    \MMDn = \arg\min_{\theta \in \Theta} \MMD{\Ptheta}{P_n} ,
\end{align}
assuming this minimum exists. An approximate minimizer can be used instead when the exact minimizer does not exist. In practice $\MMDn$ can be estimated from data using methods such as stochastic gradient descent (see e.g., \cite{BadrAlquierMMD} or Algorithm \ref{PSGAAlgorithm} below). The MMD estimator has several favorable properties. First, it can be shown that $\MMDn$ is consistent and asymptotically normal under appropriate conditions:
\begin{thm}[Informal variant of Proposition 1 and Theorem 2 in \cite{briol2019statistical}]
\label{MMD_missing_asymp}
    Under appropriate regularity conditions (including the existence of the involved quantities), $\MMDn$ is strongly consistent and asymptotically normal:
    $$
    \MMDn \xrightarrow[\ n \to +\infty]{\PX-\textnormal{a.s.}} \MMDinf ,
    $$
    $$
    \sqrt{n} (\MMDn- \MMDinf) 
    \xrightarrow[\ n \to +\infty]{\mathcal{L}} \mathcal{N}\left(0,B^{-1}\Sigma B^{-1}\right) ,
    $$
    where 
    $$
    \MMDinf=\arg\min_{\theta\in\Theta} \MMD{\Ptheta}{\PX} ,
    $$
    $$
    B = \nabla_{\theta,\theta}^2 \MMDtwo{P_{\MMDinf}}{\PX} ,
    $$
    and
    \vspace{-0.3cm}
    $$
    \Sigma = \Var_{X\sim \PX}\left[\nabla_{\theta} \MMDtwo{P_{\MMDinf}}{\delta_{\{X\}}}\right] .
    $$
\end{thm}

Hence, when the model is correctly specified, that is $\PX=\Pthetastar \in\{\Ptheta\}_\theta$, then $\MMDn$ converges to the true parameter $\theta^*$, i.e.\ $\MMDinf=\theta^*$, as soon as the model is identifiable, that is $\theta\neq\theta' \implies P_{\theta} \neq P_{\theta'}$. In case of misspecification, i.e.\ when $\PX \notin\{\Ptheta\}_\theta$, $\MMDn$ converges towards the parameter corresponding to the best (in the MMD sense) approximation of the true distribution $\PX$ in the model. A formal statement of Theorem \ref{MMD_missing_asymp} can be found in \cite{briol2019statistical} for generative models. The extension to the general case is straightforward.

\vspace{0.2cm}
What is particularly interesting with the MMD is that it is a robustness-inducing distance. For instance, when estimating a univariate Gaussian mean $\theta^*\in\mathbb{R}$ with a proportion $\dataeps$ of data points that are adversarially contaminated, the limit $\MMDinf$ of $\MMDn$ will be no further (in Euclidean distance) from $\theta^*$ than $\dataeps$, while the limit of the MLE can be made arbitrarily far away from $\theta^*$ (see the results of Section 4.1 in \cite{BadrAlquierMMD} for a formal statement). This is the case because the MLE
\begin{align}\label{MLE_Estimator_full_data}
   \MLEn=  \argmax_{\theta} \sum_{i=1}^{n} \log \ptheta(X_i),
\end{align}
where $\ptheta$ is the density of $\Ptheta$ w.r.t.\ Lebesgue's measure, converges a.s. to $\MLEinf$ whereby
\[
\MLEinf=  \argmin_{\theta} \KL{\PX}{\Ptheta} ,
\]
and the KL divergence is not bounded and does not induce the desirable robustness properties inherent to the MMD. Consider for example the case where the data are realizations of $n$ i.i.d.\ random variables $X_1, \dots, X_n$ actually drawn from the mixture:
$
\PX = \left(1 - \dataeps\right)\cdot\mathcal{N}(\theta^*,1) + \dataeps \cdot \mathcal{N}(\xi, 1),
$  
where the distribution of interest $\mathcal{N}(\theta^*,1)$ is contaminated by another univariate normal $\mathcal{N}(\xi, 1)$. Then $\MLEinf = (1-\dataeps)\theta^*+\dataeps\xi$, which can be quite far from $\theta^*$ when $|\dataeps(\xi-\theta^*)|$ is large when compared to $\theta^*$. This illustrates the non-robustness of the MLE to small deviations from the parametric assumption. An empirical example illustrating this is given in Section \ref{Sec_Simulation}.

\subsection{Maximum Likelihood Estimation under missing data}

Maximum likelihood estimation can be applied in the presence of missing values. The key idea is to treat the pair \( (X, M) \) as random and model their joint distribution using a complete-data model \( \{P_\theta\}_\theta \) over $X$ along with a missingness mechanism parameterized by some \( \phi \). Given a collection of observed i.i.d.\ random variables \( \{(X_i^{(M_i)}, M_i)\}_{i=1}^n \) with \( (X, M) \sim \Pjoint \), the MLE for \( \theta \) is obtained by maximizing the likelihood of \( \{(X_i^{(M_i)}, M_i)\}_{i=1}^n \) over \( (\theta, \phi) \). 

\vspace{0.2cm}
A fundamental result in the missing data literature, known as the ignorability property of the MLE, was established by \cite{Rubin_Inferenceandmissing}. This result states that inference based on the following estimator:  
\[
\MLEn = \arg\max_{\theta} \sum_{i=1}^{n} \log p_{\theta}^{(M_i)}\left(X_i^{(M_i)}\right)
\]
remains valid, i.e.\ that \( \MLEn \) is equivalent to the maximum likelihood estimator computed using the full available dataset \( \{(X_i^{(M_i)}, M_i)\}_{i=1}^n \), provided that:  
\begin{itemize}
    \item the assumed missingness mechanism is Missing At Random,
    \item \( \theta \) and \( \phi \) are distinct, meaning they do not share any components.  
\end{itemize}
It is important to note that \( \MLEn \) differs from the standard MLE in general, and different terminology is often used to emphasize this distinction. Specifically, $\MLEn$ is sometimes referred to as the ignoring MLE, while the regular MLE is known as the full-information MLE. However, under the above conditions, both optimization procedures yield the same estimator. A key advantage of \( \MLEn \) is that it eliminates the need to specify the missingness mechanism. By implicitly treating each \( X_i^{(M_i)} \) as a sample from a marginal of \( \PX \) (while it is in fact a marginal of \( \PXM \)), it leads to a simpler optimization problem while still adhering to the principles of maximum likelihood estimation. 

\vspace{0.2cm}
The asymptotics of $\MLEn$ under MAR have been formally studied by \cite{MLEConsistencyunderMissing}. Perhaps surprisingly, the often cited parameter distinctness condition is not necessary for consistency and asymptotic normality, although the asymptotic variance of the ignoring MLE $\MLEn$ is suboptimal compared to the full-information MLE variance when the parameter distinctness does not hold. The analysis has been extended to the MNAR setting in a recent paper \cite{MLEMisspecificationunderMissing}, as summarized below.



\begin{thm}[Informal variant of Theorems 1 and 2 in \cite{MLEMisspecificationunderMissing}]
\label{MLE_missing_asymp}
    Under appropriate regularity conditions (including the existence of the involved quantities), $\MLEn$ is strongly consistent and asymptotically normal:
    $$
    \MLEn \xrightarrow[\ n \to +\infty]{\Pjoint-\textnormal{a.s.}} \MLEinf ,
    $$
    $$
    \sqrt{n} (\MLEn- \MLEinf) 
    \xrightarrow[\ n \to +\infty]{\mathcal{L}} \mathcal{N}\left(0,A^{-1}V A^{-1}\right) ,
    $$
    where 
    $$
    \MLEinf=\arg\min_{\theta\in\Theta} \E_{M \sim \PM}\left[\KL{\PmargMXM}{\PmargMtheta}\right] ,
    $$
    $$
    A = \E_{(X,M)\sim \Pjoint}\left[\nabla_{\theta,\theta}^2 \log p_{\MLEinf}^{(M)}\big(X^{(M)}\big) \right] ,
    $$
    and
    \vspace{-0.3cm}
    $$
    V = \mathbb{E}_{(X,M)\sim \Pjoint}\left[\nabla_{\theta}\log p_{\MLEinf}^{(M)}\big(X^{(M)}\big) \cdot \nabla_{\theta}\log p_{\MLEinf}^{(M)}\big(X^{(M)}\big)^T \right] .
    $$
\end{thm}

While $\MLEinf$ is not formulated using the KL in \cite{MLEMisspecificationunderMissing} but using the expected density (see their Assumption 5), it is straightforward to see the correspondence between the two quantities. We refer the reader to \cite{MLEMisspecificationunderMissing} for a precise statement of the regularity conditions. Crucially, we note that this result holds without any assumption on the missingness mechanism.

As without missing values, we remark that $\MLEinf$ is a KL minimizer, but trying now to match the observed marginal of the conditional $\PXM$ on average over the pattern $M$. 
Remarkably, when the model is well-specified, that is $\PX=\Pthetastar$, and the missingness mechanism is MAR, then we recover $\MLEinf=\theta^*$. 
Unfortunately, the KL also implies that when the complete-data model is misspecified, the ignoring MLE under missingness suffers from similar shortcomings as in the complete-data case.
Inspired by the form the MLE estimator takes, we introduce in the next section an MMD estimator for the case of missing values.

\section{MMD M-estimator for missing data}\label{Sec_EstimatorIntro} 
\label{sec_asymp}

In this section, we define our MMD estimator in the presence of missing data. We shortly discuss its computation using stochastic gradient-based algorithms, and provide an asymptotic analysis of the estimator. 

\subsection{Definition of the estimator}

We begin by presenting an alternative interpretation of the MMD estimator. Notably,
$$
\MMDn = \arg\min_{\theta\in\Theta} \MMDtwo{\Ptheta}{P_n} = \arg\min_{\theta\in\Theta} \frac{1}{n} \sum_{i=1}^n \MMDtwo{\Ptheta}{\delta_{\{X_i\}}} ,
$$ 
which implies that our MMD estimator can be viewed not only as a minimum distance estimator but also as an M-estimator. While this connection was previously noted in \cite{oates2022minimum,AlquierCopulas} for technical proofs, we leverage it to construct an estimator that is both computationally feasible and meaningful in the presence of missing values. Similarly, \cite{AlquierGerberRegression} introduced an MMD-based M-estimator for regression, primarily for computational efficiency. This perspective reinterprets the estimated distribution within the model as the best approximation (on average) of a delta mass $\delta_{\{X_i\}}$ using a quadratic MMD criterion, rather than as the best approximation of the empirical measure $P_n$. This distinction allows for reasoning at the individual level via each $X_i$, whereas the traditional minimum distance approach operates at the population level through $P_n$.

\vspace{0.2cm}
Inspired by the MLE framework ignoring the missingness mechanism, we then propose to compare marginal quantities at the individual level, aligning the marginal distribution $\PmargMitheta$ corresponding to the pattern $M_i$ with the observed values $X_i^{(M_i)}$ using the MMD distance. This finally leads to the following estimator:
\begin{align*}
\MMDn = \arg\min_{\theta\in\Theta}
\frac{1}{|\{i:M_i\neq\{1\}^d\}|}
\sum_{i:M_i\neq\{1\}^d}
\MMDtwo{\PmargMitheta}{\delta_{\{X_i^{(M_i)}\}}} \, 
\end{align*}
where individuals for which all components are missing are excluded from the criterion.

$\MMDn$ is no longer a minimum distance estimator, but simply an M-estimator. Notice that each MMD distance considered in the sum is defined over a different space $\mathbb{R}^{|M_i|}\times\mathbb{R}^{|M_i|}$, which poses no issues due to the dimension-independence of the kernel $k$. The same remark actually applies to the KL divergence used in the definition of $\MLEinf$. We also define the population counterpart of $\MMDn$ as 
\begin{align}\label{eq_MMDinfdef}
    \MMDinf= \arg\min_{\theta\in\Theta} \E_{M\sim\PM}\left[\MMDtwo{\PmargMtheta}{\PmargMXM} \right].    
\end{align}



\subsection{Computation using SGD}

The computation of the MMD estimator in the complete-data setting has already been discussed in \cite{briol2019statistical,BadrAlquierMMD} using stochastic gradient-based algorithms. We explain in this subsection how to adapt the analysis to the missing values setting, assuming that $\Theta \subset \mathbb{R}^p$.

\vspace{0.2cm}

We note that by definition $\mathbb{E}_{Y^{(m)} \sim \Pmargmtheta}[f(Y^{(m)})]=\mathbb{E}_{Y \sim \Ptheta}[f(Y^{(m)})]$ for any pattern $m$ and any measurable function $f:\mathbb{R}^{|m|}\rightarrow\mathbb{R}$, which allows to rewrite the definition of $\MMDn$ using the alternative formulation \eqref{alt_MMD} of the MMD: We first define,
\begin{align*}
    \operatorname{Crit}_n(\theta)=\frac{1}{|\{i:M_i\neq\{1\}^d\}|} \sum_{i:M_i\neq\{1\}^d} \Big\{ \mathbb{E}_{Y,\widetilde{Y} \sim \Ptheta}\left[k\left(Y^{(M_i)}, \widetilde{Y}^{(M_i)}\right)\right] - 2 \mathbb{E}_{Y \sim \Ptheta}\left[k\left(X_i^{(M_i)}, Y^{(M_i)}\right)\right] \Big\},
\end{align*}
so that,
\begin{align}
\label{criterion}
\MMDn =\arg\min_{\theta\in\Theta} \operatorname{Crit}_n(\theta) . 
\end{align}
To apply a gradient-based method, the first step is to compute the gradient of this criterion with respect to $\theta$. The formula is provided in the following proposition, which is a straightforward adaptation of Proposition 5.1 in \cite{BadrAlquierMMD}: 
\begin{prop}
    Assume that each $P_\theta$ has a density $p_\theta$ w.r.t.\ Lebesgue's measure. Assume that for any $x$, the map $\theta \mapsto p_\theta(x)$ is differentiable w.r.t.\ $\theta$, and there exists a nonnegative function $g(x, x')$ such that for any $\theta \in \Theta$, $\left| k(x, x') \nabla_\theta [p_\theta(x) p_\theta(x')] \right| \leq g(x, x')$ and $\int \int g(x, x') dx dx' < \infty$.
    Then the gradient of Objective \eqref{criterion} is given by:
    \[
   \nabla_{\theta} \operatorname{Crit}_n(\theta)=\mathbb{E}_{Y,\widetilde{Y} \sim \Ptheta}\bigg[ \frac{2}{|\{i:M_i\neq\{1\}^d\}|} \sum_{i:M_i\neq\{1\}^d} \bigg\{ \left[ k\left(Y^{(M_i)}, \widetilde{Y}^{(M_i)}\right) - k\left(X_i^{(M_i)}, Y^{(M_i)}\right) \right] \cdot \nabla_\theta[\log p_\theta(Y)]  \bigg\} \bigg]
    \]
\end{prop}

While this gradient may be challenging to compute explicitly and not be available in closed-form, we can remark as \cite{briol2019statistical,BadrAlquierMMD} in the complete-data setting that the gradient is an expectation w.r.t.\ the model. Hence, if we can evaluate $\nabla_\theta [\log p_\theta(x)]$ and if $\{P_\theta\}_\theta$ is a generative model, then we can approximate the gradient using a simple unbiased Monte Carlo estimator: first simulate $(Y_1, \dots, Y_S) \stackrel{i.i.d.}{\sim} P_\theta$, and then use the following estimate for the gradient at $\theta$
\begin{align*}
&\widehat{\nabla_{\theta} \operatorname{Crit}}_n(\theta)=\\
&\frac{2}{|\{i:M_i\neq\{1\}^d\}|\,S} \sum_{i:M_i\neq\{1\}^d} \sum_{j=1}^S \left( \frac{1}{S - 1} \sum_{j' \neq j} k\left(Y_j^{(M_i)}, Y_{j'}^{(M_i)}\right) -  k\left(X_i^{(M_i)}, Y_j^{(M_i)}\right) \right) \nabla_\theta \log [p_\theta(Y_j)] .
\end{align*}
We note that since the expectation is formed by simulation from the proposed model, this expression is unbiased, conditional on the observed data, under any missingness mechanism. This unbiased estimate can then be used to perform a stochastic gradient descent (SGD), leading to Algorithm \ref{PSGAAlgorithm}. We incorporated a projection step to consider the case where $\Theta$ is strictly included in $\mathbb{R}^p$, and implicitly assumed that $\Theta$ is closed and convex. $\Pi_\Theta$ denotes the orthogonal projection onto $\Theta$.

\begin{algorithm}
\small
\caption{Projected Stochastic Gradient Algorithm for MMD under missing data}
\label{PSGAAlgorithm}
\begin{algorithmic}
\INPUT{a dataset $(X_1^{(M_1)},\cdots,X_n^{(M_n)})$, a model $\{P_\theta:\theta \in \Theta \subset \mathbb{R}^p\}$, a kernel $k$, a sequence of step sizes $\{\eta_t\}_{t\geq1}$, an integer $S$, a stopping time $T$.}
\State Initialize $\theta^{(0)} \in \Theta$, $t = 0$.
\For{$t = 1,\ldots,T$}
    \State draw $(Y_1,\ldots,Y_S)$ i.i.d from $P_{\theta^{(t-1)}}$ 
\State 
$
\theta^{(t)} =\Pi_\Theta \left\{\theta^{(t-1)} - \eta_t \widehat{\nabla_{\theta} \operatorname{Crit}}_n(\theta^{(t-1)})\right\}
$
\EndFor
\end{algorithmic}
\end{algorithm}

This algorithm provides a basic approach to computing our estimator, primarily relying on standard Monte Carlo estimation combined with stochastic gradient descent in the context of generative models. As detailed in Section \ref{Sec_conc}, this algorithm might be improved. In particular, a fully generative approach as in \cite{briol2019statistical} is possible.

\subsection{Asymptotic analysis}

In this section, we analyze the properties of $\MMDn$, without making any assumption on the nature of the missingness mechanism. The conditions we formulate here are broad and may differ slightly from those typically used in the MMD literature. Existing approaches generally fall into two categories: assumptions on the density \citep{key2021composite,AlquierGerberRegression} or on the generator \citep{briol2019statistical}. Some works \cite{oates2022minimum, AlquierCopulas} have proposed more general conditions that encompass both perspectives, though these can be challenging to establish due to the diverse nature of the underlying random objects.

\vspace{0.2cm}
In this work, we adopt a similarly general approach by working directly with the model distribution. This allows us to capture subtle differences between complete and missing data settings. Nevertheless, standard assumptions ensuring the consistency and asymptotic normality of the MMD estimator in the complete data case continue to guarantee the consistency of our estimator, as further detailed below.


    \begin{asm}
    \label{cond_consistency1}
    $\Theta$ is compact and $\MMDinf$ in \eqref{eq_MMDinfdef} is unique, if it exists. 
    \end{asm}

    

    \begin{asm}
    \label{cond_consistency2} 
    The map $\theta \mapsto \Phi(\Pmargmtheta)$ is continuous on $\Theta$ for each $m\in\{0,1\}^d$ in the support of $\PM$. 
    \end{asm}

We note that with a bounded kernel, existence of $\MMDinf$ is guaranteed by the compactness of $\Theta$ and Assumption \ref{cond_consistency2}. As such, only uniqueness of $\MMDinf$ needs to be assumed. These conditions ensuring the consistency of $\MMDn$ towards $\MMDinf$ are almost the same conditions as for the standard MMD estimator, with a few notable differences:
\begin{itemize}
\item The limit is no longer the minimizer of the MMD distance to the true whole distribution $\PX$ as in the complete data setting. It now corresponds to the best average (over $M\sim\PM$) approximation of the observed marginal of the conditional $\PXM$ in terms of MMD distance, as for the ignoring MLE, and is equal to $\theta^*$ as soon as the model is well-specified and the missingness mechanism is MCAR. 
This setting of the mechanism is crucial, as under M(N)AR, distribution shifts are possible when moving from one missingness pattern to the next. In fact, as discussed in \cite{ourresult, näf2024goodimputationmarmissingness} and others, MAR can be defined as a restriction on the arbitrary distribution shifts that are allowed under MNAR: MAR is equivalent to the condition: $\Prob_{X \mid M=m}=\Prob_{X^{(-m)} \mid X^{(m)}} \ltimes  \Prob_{X^{(m)} \mid M=m}$, showing roughly that only the observed variables may change distributions but the conditional distribution cannot depend on $m$ arbitrarily.\footnote{Here $\ltimes$ denotes the semidirect product of the conditional law of
$X^{(m)} \mid M=m$ with the Markov kernel
$x^{(m)} \mapsto \mathbb{P}_{X^{(-m)} \mid X^{(m)} = x^{(m)}}$,
i.e.\ the unique measure on $\mathbb{R}^d$ satisfying
\[
\mathbb{P}_{X \mid M=m}(A_1 \times A_2)
= \int_{A_1} \mathbb{P}_{X^{(-m)} \mid X^{(m)} = x^{(m)}}(A_2)\,
\mathbb{P}_{X^{(m)} \mid M=m}(dx^{(m)}).
\]} For a more detailed discussion we refer to \cite{näf2024goodimputationmarmissingness}.
Intuitively, if these distribution shifts from one pattern to the next are not too extreme $\MMDinf$ will be ``reasonably close'' to the truth.
\item The continuity in $\theta$ of the embedding of each marginal $\Phi( \Pmargmtheta)$ is required here as opposed to the continuity of the embedding of the whole distribution $\Phi(\Ptheta)$ that would be required in the complete data case. Our condition is then slightly stronger. However, both conditions are guaranteed to hold under standard regularity assumptions generally formulated over the density and the kernel, see e.g.\ \cite{key2021composite}.
\item Uniqueness of $\MMDinf$ is not trivially checked even with complete data. An MNAR missingness mechanism might make this even more difficult. In Appendix \ref{app_proofs} we discuss how under MCAR and in the well-specified setting, uniqueness of $\MMDinf$ follows under relatively simple conditions. We leave further investigation under misspecification and/or MNAR for future work.
\end{itemize}

\begin{thm}\label{thm_consistency}
If Conditions \ref{cond_consistency1} and \ref{cond_consistency2} are fulfilled, then $\MMDn$ is strongly consistent, i.e.
$$
\MMDn \xrightarrow[\ n \to +\infty]{\Pjoint-\textnormal{a.s.}} \MMDinf.
$$
\end{thm}

We need a few more assumptions to establish the asymptotic normality of our estimator:

    \begin{asm}
    \label{cond_normality_3} 
    $\MMDinf$ is an interior point of $\Theta$.
    \end{asm}

    \begin{asm}
    \label{cond_normality_4} 
    There exists an open neighborhood $O \subset \Theta$ of $\MMDinf$ such that the maps $\theta \mapsto \Phi( \Pmargmtheta)$ are twice (Fr\'echet) continuously differentiable on $O$ for each $m\in\{0,1\}^d$ in the support of $\PM$, with interchangeability of integration and derivation.
    \end{asm}
   
    \begin{asm}
    \label{cond_normality_5} 
    There exists a compact set $K \subset O$ whose interior contains $\MMDinf$ and such that for any pair $(j,k)$, and any $m \in \{0,1\}^d$ in the support of $\PM$,
    \[
    \sup_{\theta \in K} \left\| \frac{\partial^2 \Phi( \Pmargmtheta)}{\partial \theta_j \partial \theta_k} \right\|_{\mathcal{H}} < +\infty .
    \]
    \end{asm}

    \begin{asm}
    \label{cond_normality_6} 
    The matrix $H = \E_{M\sim\PM}\left[\nabla_{\theta,\theta}^2 \MMDm{2}{P_{\MMDinf}^{(M)}}{\PmargMXM}\right]$ is nonsingular.
    \end{asm}

    \begin{asm}
    \label{cond_normality_7} 
    The two matrices $\Sigma_1=\E_{M\sim\PM}\left[\Var_{X\sim\PXM}\left[\nabla_{\theta} \MMDm{2}{P_{\MMDinf}^{(M)}}{\delta_{\{X^{(M)}\}}}\right]\right]$ and $\Sigma_2=\Var_{M\sim\PM}\left[ \nabla_{\theta} \MMDm{2}{P_{\MMDinf}^{(M)}}{\PmargMXM} \right]$ are nonsingular. 
    \end{asm}

Once again, we recover almost exactly the main conditions as for the standard MMD estimators with four notable differences:
\begin{itemize}
\item The continuous differentiability in $\theta$ of the embedding of the whole distribution $\Phi(\Ptheta)$ is replaced by the embedding of each marginal $\Phi( \Pmargmtheta)$, as for the continuity required in Assumption \ref{cond_consistency2}. Note that interchangeability is usually implicitly ensured through domination conditions imposed on the densities.
\item The boundedness assumption of the second order derivative holds once again for each marginal rather than for the whole distribution.
\item The Hessian matrix $\nabla_{\theta,\theta}^2 \mathbb{D}^2\left(P_{\theta^*}, \PX \right)$ involving the distance to $\PX$ that was used in the definition of the asymptotic variance of $\theta_n$ without missing data is replaced by the expected (over $M\sim\PM$) value of the Hessian matrix $\nabla_{\theta,\theta}^2 \MMDm{2}{P_{\MMDinf}^{(M)}}{\PmargMXM}$ involving the distance to the observed marginal of the conditional $\PXM$.
\item The new version of the variance matrix $\Sigma=\Var_{X\sim\PX}\left[\nabla_{\theta} \MMDtwo{P_{\theta^*}}{\delta_{\{X\}}}\right]$ that was used for $\theta_n$ without missing data now involves two terms: $\Sigma_1$ is the expected (over $M\sim\PM$) value of the conditional variance $\Var_{X\sim\PXM}\left[\nabla_{\theta} \MMDm{2}{P_{\MMDinf}^{(M)}}{\delta_{\{X^{(M)}\}}}\right]$, while $\Sigma_2$ is a term that solely takes into account the variability in the missingness mechanism. 
\end{itemize}
We now state the weak convergence of $\sqrt{n}(\MMDn-\MMDinf)$, which is a result that holds for any missingness mechanism (even for MNAR data) and that provides a $n^{-1/2}$ rate of convergence:

\begin{thm}\label{thm_normality}
If Conditions  \ref{cond_consistency1} to \ref{cond_normality_7} are fulfilled, then $\sqrt{n}(\MMDn - \MMDinf)$ is asymptotically normal, with limit:
\[
\sqrt{n} (\MMDn - \MMDinf) 
\xrightarrow[\ n \to +\infty]{\mathcal{L}} \mathcal{N}\left(0,H^{-1}\Sigma_1 H^{-1} + H^{-1}\Sigma_2 H^{-1}\right) .
\]
\end{thm}

Interestingly, the decomposition of the central asymptotic variance term \(\Sigma_1 + \Sigma_2\) for our estimator \(\MMDn\) is not explicit in the case of the MLE under missing data (see Theorem \ref{MLE_missing_asymp}). This arises because, in the absence of missing data, the central term \(V\) in the asymptotic variance of the MLE can be alternatively formulated as the variance of the score:
\[
\mathbb{E}_{X\sim \PX}\left[\nabla_{\theta}\log p_{\MLEinf}(X) \cdot \nabla_{\theta}\log p_{\MLEinf}(X)^T \right] = \Var_{X\sim \PX}\left[\nabla_{\theta}\log p_{\MLEinf}(X)\right].
\]
However, this identity no longer holds in the presence of missing data, since the expectation \(\mathbb{E}_{X\sim \PX}[\nabla_{\theta}\log p_{\MLEinf}(X)]\) becomes \(\mathbb{E}_{X\sim \PXM}[\nabla_{\theta}\log p_{\MLEinf}^{(M)}(X^{(M)})]\) which is not zero. As a result, expressing \(V\) as an expectation rather than a variance under missing data obscures the additional term. Instead, \(V\) can be decomposed as the average (over \(M\sim\PM\)) conditional variance of the score plus an explicit additional term:
\begin{align*}
    V &= \mathbb{E}_{(X,M)\sim \Pjoint}\left[\nabla_{\theta}\log p_{\MLEinf}^{(M)}(X^{(M)}) \cdot \nabla_{\theta}\log p_{\MLEinf}^{(M)}(X^{(M)})^T \right] \\
    &= \mathbb{E}_{M\sim \PM}\left[\Var_{X\sim \PXM}\left[\nabla_{\theta}\log p_{\MLEinf}^{(M)}(X^{(M)})\right]\right] \\
    &\quad + \mathbb{E}_{M\sim \PM}\left[\mathbb{E}_{X\sim \PXM}\left[\nabla_{\theta}\log p_{\MLEinf}^{(M)}(X^{(M)}) \right] \cdot \mathbb{E}_{X\sim \PXM}\left[\nabla_{\theta}\log p_{\MLEinf}^{(M)}(X^{(M)}) \right]^T\right].
\end{align*}

Having established these results we analyze the robustness of our new estimator.

\section{Robustness}
\label{Sec_robust}


In the previous section, we analyzed the asymptotic behavior of our estimator $\MMDn$, and determined its limit value. While it can be shown that under a well-specified model $\PX=\Pthetastar$ and a MCAR mechanism, $\MMDn$ consistently converges to the true parameter $\theta^*$, the situation is less clear for an M(N)AR mechanism. Specifically, it remains uncertain whether the limit $\MMDinf$ remains close to $\theta^*$ in the presence of a non-ignorable missingness mechanism, that is, an M(N)AR mechanism. In this section, we examine the implications of an M(N)AR mechanism on the behavior of $\MMDn$ and even briefly consider its impact on the maximum likelihood estimator $\MLEn$. Additionally, we address the scenario of a misspecified complete-data model, where $\PX\notin\{P_\theta:\theta\in\Theta\}$, with a specific focus on contamination models accounting for the presence of outliers in the data. We will focus on bounds on $ \E_{M \sim \PM}\left[\MMDtwo{P_{\MMDinf}^{(M)}}{\PmargMX}\right]$, bounding the average (squared) MMD distance to the true distribution over all patterns. This average MMD distance provides a measure of divergence between $P_{\MMDinf}$ and $\PX$, which is particularly relevant in settings where the model is identifiable, that is $\theta\neq\theta' \implies P_{\theta} \neq P_{\theta'}$. Under this condition, if $\PX=\Pthetastar$, a zero averaged squared MMD discrepancy implies perfect recovery, i.e., $\MMDinf = \theta^*$, as soon as the kernel is characteristic and the probability of complete data is not zero (i.e. $\Prob(M=\{0\}^d) > 0$), since
\[
\MMDtwo{P_{\MMDinf}}{\PX} \leq \frac{1}{\Prob(M=\{0\}^d)} \E_{M \sim \PM}\left[\MMDtwo{P_{\MMDinf}^{(M)}}{\PmargMX}\right].
\]
In the special case of one-dimensional data, this averaged squared MMD distance reduces to the standard MMD distance $\MMDtwo{P_{\MMDinf}}{\PX}$ weighted by the probability of observing the datapoint. It is worth noting that, under standard regularity conditions (see, e.g.\ Theorem 4 in \cite{AlquierGerberRegression}), the squared MMD distance between distributions can be related to the Euclidean distance between their parameters. For instance, the MMD distance $\MMDtwo{P_{\theta}}{P_{\theta'}}$ between two normal distributions $\mathcal{N}(\theta,I_d)$ and $\mathcal{N}(\theta',I_d)$ is equivalent (up to a factor) to the Euclidean distance $\|\theta-\theta'\|^2_2$ when $\theta$ and $\theta'$ are close enough. We explore this in more detail in Section \ref{Sec_wellspecified_Truncation} and Appendix \ref{app_finite_sample}.


\subsection{Robustness of the MLE to MNAR}

The first question to consider is whether the MLE is inherently sensitive to contamination introduced by the missingness mechanism. Perhaps unexpectedly, the following result demonstrates that, for one-dimensional data and well-specified complete-data models, the MLE exhibits robustness under an MNAR mechanism.

\begin{thm}[Robustness of the MLE limit to an MNAR missingness mechanism when $d=1$]
\label{MLE_Robust}
    Assume that $d=1$ and that the model is well-specified $\PX=\Pthetastar$. Then:
    $$    \textnormal{TV}^2\left(P_{\MLEinf},\Pthetastar\right) \leq 2 \cdot \frac{\Var_{X\sim\PX}[\pi(X)]}{\pi^2} 
    $$
    where $\Var_{X}[\pi(X)]$ is the variance (w.r.t.\ $X\sim\PX$) of the missingness mechanism $\pi(X)=\mathbb{P}[M=0|X]$, and $\pi=\mathbb{P}[M=0]=\E_{X\sim\PX}[\pi(X)]$.
\end{thm}
This theorem establishes that the discrepancy between $\MLEinf$ and $\theta^*$, quantified by the total variation (TV) distance between $\Pthetastar$ and $P_{\MLEinf}$, is upper bounded by the relative variance of the missingness mechanism with respect to $X\sim\PX$. This relative variance serves as a measure of the level of misspecification to the M(C)AR assumption: a large variance indicates that the missingness mechanism deviates significantly from being M(C)AR, whereas a small variance implies it is nearly M(C)AR. In the extreme case where the missingness mechanism is exactly M(C)AR - meaning that ignoring it is valid - we recover the equality $\MLEinf=\theta^*$ provided the model is identifiable, that is $\theta\neq\theta' \implies P_{\theta} \neq P_{\theta'}$.

\vspace{0.2cm}
However, this result has two key limitations: (i) it is unclear whether it extends beyond the one-dimensional case, where MCAR and MAR are distinct concepts; and (ii) the result fails entirely when the assumption $\PX=\Pthetastar$ no longer holds. In particular, no guarantees can be derived in scenarios involving data contamination, where the true data distribution $\PX$ deviates from the target distribution $\Pthetastar$ and instead represents a contaminated version of it - this issue arises even in the absence of missing values.

\subsection{Robustness of the MMD to M(N)AR and misspecification}\label{Sec_general_general}


We recall the notation $\pi_m(X)=\PMmX$ for the missingness mechanism and $\pi_m=\PMm=\E_{X\sim\PX}\left[\pi_m(X)\right]$ for the marginal probability.
\vspace{0.2cm}

\begin{thm}[Robustness of the MMD limit to M(N)AR and model misspecification]
\label{MMD_Robust}
For $\MMDinf$ defined as in \eqref{eq_MMDinfdef}, it holds that\footnote{With a slight abuse of notation, we take
\[
\E_{M \sim \PM}\left[ \frac{\Var_{X\sim \PX} \left[ \pi_M(X) \right]}{\pi_M^2} \right] = \sum_{m\ne\{1\}^d} \PMm \frac{\Var_{X\sim \PX} \left[ \pi_m(X) \right]}{\pi_m^2}.
\]},
    \begin{equation}
    \label{MMD_Robust_2} 
       \E_{M \sim \PM}\left[\MMDtwo{P_{\MMDinf}^{(M)}}{\PmargMX}\right] \leq 2\cdot \inf_{\theta\in\Theta}\E_{M \sim \PM}\left[\MMDtwo{P_{\theta}^{(M)}}{\PmargMX}\right] + 4 \cdot \E_{M \sim \PM}\left[ \frac{\Var_{X\sim \PX} \left[ \pi_M(X) \right]}{\pi_M^2} \right] .   
    \end{equation}
\end{thm}


As mentioned above, this theorem leverages the averaged (w.r.t.\ the missing pattern $M\sim\PM$) squared MMD distance between the observed marginals $P_{\MMDinf}^{(M)}$ and $\Prob_{X}^{(M)}$ as a meaningful divergence to assess some notion of distance between $\MMDinf$ and $\theta^*$. Inequality \eqref{MMD_Robust_2} is particularly remarkable as it guarantees robustness to both sources of misspecification and characterizes them separately: the first term in the bound quantifies robustness to complete-data misspecification (which vanishes when $\PX=\Pthetastar$), while the second term captures robustness to deviation-to-MCAR (which vanishes when the missingness mechanism is MCAR). Note that this inequality holds without any assumption on the model. It is also worth noting that the variance appearing in the bounds could be tightened further, replacing it with an absolute mean deviation over the observed variables $X^{(m)}$ only. However, we present the bound using the variance for the sake of clarity and simplicity. As a direct corollary it holds that:

\begin{cor}
   Assuming that $\PX=\Pthetastar$, we have:
    \begin{equation}     
    \label{MMD_Robust_1} 
       \E_{M \sim \PM} \left[ \MMDtwo{P_{\MMDinf}^{(M)}}{P_{\theta^*}^{(M)}}\right] \leq 4 \cdot \E_{M \sim \PM}\left[ \frac{\Var_{X\sim \PX} \left[ \pi_M(X) \right]}{\pi_M^2} \right]
    \end{equation}   
\end{cor}

\vspace{0.2cm}
Inequality \eqref{MMD_Robust_1} is noteworthy as it guarantees robustness to any M(N)AR mechanism, while when the missingness mechanism is MCAR - ensuring the validity of ignorability - the bound becomes zero, and we recover the consistency of $\MMDn$ towards $\MMDinf = \theta^*$ whenever $\Prob(M=\{0\}^d) > 0$ and the kernel is characteristic. This result generalizes the bound established in the one-dimensional case for the MLE, offering a quantification of the deviation from MCAR based on the variability of the missingness mechanism with respect to the observed data. 

\vspace{0.2cm}
We finally derive a non-asymptotic inequality on the MMD estimator $\MMDn$ which holds in expectation over the sample $\mathcal{S}=\{(X_i,M_i)\}_{1\leq i\leq n}$ in the one-dimensional case where the expected squared MMD distance can be written as a proper distance. For the case $d > 1$, we refer to Appendix \ref{app_finite_sample}.

\begin{thm}[Robustness of the MMD estimator to MNAR and misspecification when $d=1$]
\label{thm_finite_main}
    With the notations $\pi(X)=\mathbb{P}[M=0|X]$ and $\pi=\mathbb{P}[M=0]=\E_{X\sim\PX}[\pi(X)]$, we have:
    $$
        \E_{\mathcal{S}}\left[ \MMD{P_{\MMDn}}{\PX} \right]  \leq \inf_{\theta\in\Theta} \MMD{P_\theta}{\PX} + \frac{2\sqrt{\Var_{X\sim \PX}\left[\pi(X)\right]}}{\pi} + \frac{2\sqrt{2}}{\sqrt{n\cdot\pi}} .
    $$
\end{thm}

This bound offers a detailed breakdown of the three types of errors involved. The first term represents the misspecification error arising from using an incorrect model. The second term reflects the error due to incorrectly assuming that the missingness mechanism follows the M(C)AR assumption. Finally, the third term accounts for the estimation error, which stems from the finite size of the dataset and the presence of unobserved data points, effectively reducing the available sample size to $n\pi$. In Appendix \ref{app_finite_sample}, we show that a similar result can be obtained for $d \geq 1$ and general patterns.



\vspace{0.2cm}
A result similar to Theorem \ref{thm_finite_main} above can be derived for the estimator computed in Algorithm \ref{PSGAAlgorithm}. 

\begin{thm}[Robustness of Algorithm \ref{PSGAAlgorithm} to MNAR when $d=1$]
\label{thm_finite_algo}
Denote $\pi(X)=\mathbb{P}[M=0|X]$ and $\pi=\mathbb{P}[M=0]=\E_{X\sim\PX}[\pi(X)]$, and assume that $\PX=\Pthetastar$.
Let $\Theta\subset\mathbb{R}$ be closed, convex and bounded, with
\[
D=\sup_{\theta,\theta'\in\Theta}|\theta-\theta'|.
\]
Define
\[
\widehat{\theta}_n^{(T)} = \frac{1}{T}\sum_{t=1}^T\theta^{(t)},
\]
where $(\theta^{(t)})_{t\geq1}$ are the iterates produced by
Algorithm~\ref{PSGAAlgorithm}. 

Assume that $\operatorname{Crit}_n(\cdot)$ is convex and that, for every $\theta\in\Theta$,
\[
\mathbb{E}\left[\left|\widehat{\nabla_\theta\operatorname{Crit}_n}(\theta)\right|^2\right]
\leq C^2 \, .
\]
Then, with step size
\[
\eta=\frac{D}{C\sqrt{T}}, 
\]
we have 
\[
\mathbb{E}\left[\mathbb{D}\left(P_{\widehat{\theta}_n^{(T)}},\PX\right)\right]
\leq
\frac{ 2\sqrt{\Var_{X\sim P_X}[\pi(X)]} }{\pi} + \frac{2\sqrt{2}}{\sqrt{n\pi}} + \sqrt{\frac{D C}{\sqrt{T}}}
\]
where expectations are taken with respect to all sources of randomness involved.
\end{thm}

This result is very similar to the complete data theorem derived in \cite{BadrAlquierMMD}. Of course, the MMD criterion is not always convex; this is however not a limitation of our missing data analysis, but rather an inherent challenge even with complete data. Little is known about global convergence guarantees for such algorithms beyond convex objectives, but recent progress in this direction has been made by \cite{seulkee2026gradient}.

\subsection{Robustness of the MMD to M(N)AR and contamination} \label{Sec_Huber_general}

We also consider contamination models, which account for the presence of outliers in the dataset. We first focus on Huber's contamination model, where the complete-data are drawn from a distribution of interest $\Pthetastar$ with probability $1 - \dataeps$, and from a noise distribution $\mathbb{Q}_X$ with probability $\dataeps$. This can be expressed as $\PX = (1 - \dataeps) \Pthetastar + \dataeps \mathbb{Q}_X$, where $\dataeps$ denotes the proportion of outliers. The goal is to estimate the parameter of interest $\theta^*$, and we would like to replace any dependence on $\PX$ in our bounds solely by a dependence on $\dataeps$ and $\Pthetastar$, irrespective of the contamination $\QX$.  We thus define for all $m$, $\pi_m^*=\E_{X\sim P_{\theta^*}}\left[\pi_m(X)\right]$ and $P^*_M(A)=\sum_{m \in A}  \pi_m^*$, the marginal distribution of $M$ not under $\Pjoint$ but under the joint distribution of interest $\Pthetastar\times\PMmX$.
The following theorem quantifies the error as a function of the proportion of outliers and of the deviation-to-MCAR:

\begin{thm}[Robustness of the MMD limit to M(N)AR and Huber's contamination model]
\label{MMD_Robust_cont}
    We have if $\PX=(1-\dataeps)\Pthetastar+\dataeps\mathbb{Q}_X$:
    $$
    \E_{M\sim P^*_M} \left[\mathbb{D}^2\left(P_{\MMDinf}^{(M)},\PmargMthetastar\right)\right] \leq \frac{24\dataeps^2}{1-\dataeps} + 64 \, \mathbb{E}_{M\sim P^*_M}\left[ \frac{1}{\pi_M^{*2}}\right] \left(\frac{\dataeps}{1-\dataeps}\right)^2 + 16 \, \mathbb{E}_{M\sim P^*_M}\left[\frac{\Var_{X\sim \Pthetastar} \left[ \pi_M(X) \right]}{\pi_M^{* 2}}\right] .
    $$
\end{thm}
This theorem ensures robustness to both M(N)AR data and the presence of outliers, with the expected squared MMD distance bounded by the squared proportion of outliers $\dataeps^2$ weighted by $\E_M[\pi_M^{* -2}]/(1-\dataeps)^2$, except in cases where the deviation-to-MCAR missingness term is significant and dominates. This result is remarkable, as such robustness is unattainable for the MLE, where even a single outlier can cause the estimate to deteriorate arbitrarily. We emphasize that $\pi_m^*$ and thus $P_M^*$ is defined here as expectations w.r.t.\ the distribution of interest $\Pthetastar$, not under $\PX$. The second term in the bound $\E_M[\pi_M^{* -2}]\dataeps^2/(1-\dataeps)^2$ reflects this shift in the marginal of $M$, while the first term $\dataeps^2/(1-\dataeps)$ corresponds to the error in the specification of the marginal of $X$. We believe that this last term can be tightened to $\dataeps^2$ and that the factor $1/(1-\dataeps)$ is just an artifact from the proof, although we are only able to establish it in dimension one, as shown below. This is anyway only little improvement, since this model misspecification term is always bounded by the pattern shift error $\E_M[\pi_M^{* -2}]\dataeps^2/(1-\dataeps)^2$.

\vspace{0.2cm}
In the one-dimensional setting, we can derive a similar bound with tighter constants for the MMD estimator $\MMDn$ directly:

\begin{thm}[Robustness of the MMD estimator to MNAR and Huber's contamination model when $d=1$]
\label{MMD_Robust_cont_finite}
    With the notations $\pi(X)=\mathbb{P}[M=0|X]$ and $\pi^*=\E_{X\sim\Pthetastar}[\pi(X)]$, we have for $\PX=(1-\dataeps)\Pthetastar+\dataeps\mathbb{Q}_X$:
    \begin{equation}
    \label{Huber_data_MNAR}
        \E_{\mathcal{S}}\left[ \mathbb{D}\left(P_{\MMDn},\Pthetastar\right) \right] \leq 4\cdot\dataeps + \frac{8\cdot\dataeps}{\pi^*(1-\dataeps)} + \frac{2\sqrt{\Var_{X\sim \Pthetastar}\left[\pi(X)\right]}}{\pi^*} + \frac{2\sqrt{2}}{\sqrt{n\pi^*(1-\dataeps)}} .
    \end{equation}
\end{thm}
The rate achieved by the estimator is $\max\left(\dataeps/\pi^*(1-\dataeps),\Var[\pi(X)]^{1/2}/\pi^*,(n\pi^*(1-\dataeps))^{-1/2}\right)$ with respect to the MMD distance. We will provide explicit rates with respect to the Euclidean distance for a Gaussian model in the next section. Once again, the first term $\dataeps$ corresponds to the mismatch error between $\PX$ and $\Pthetastar$, and is always dominated by the second term, $\dataeps/\pi^*(1-\dataeps)$, which is the pattern shift error induced by the contamination. The third term measures the deviation-to-M(C)AR level, while the last one is the convergence rate under M(C)AR using the effective sample size $n\pi^*(1-\dataeps)$ from the M(C)AR component. 

\section{Examples}
\label{Sec_exm}

We investigate in this section explicit applications of our theory for specific examples of MNAR missingness mechanisms. We provide three of them: the first one is a Huber contamination setting for the missingness mechanism; the second one involves a truncation mechanism; while the third one allows for adversarial contamination of the missingness mechanism.

\subsection{Huber's contamination mechanism}\label{Sec_Huber_Huber}

\begin{exm}
\label{exm_huber}
    We now collect a dataset $\{X_i\}_i$ composed of i.i.d.\ copies of $X$ and only observe $\{X_i^{(M_i)}\}_i$. While we believe that the MDM is MCAR with constant (w.r.t.\ $X$) probability $\PMmX=\alpha_m$ a.s.\ (with $\sum_m\alpha_m=1$) and thus ignore the missingness mechanism, the true missingness mechanism is actually M(N)AR, with a proportion $\missingeps$ of outliers in the missingness mechanism from a contamination process $\mathbb{Q}[M=m|X]$:
    $$
    \mathbb{P}[M=m|X=x] = (1-\missingeps) \cdot \alpha_m + \missingeps \cdot \mathbb{Q}[M=m|X=x] \quad \text{for any } x. 
    $$
\end{exm}
In the following we define the expectation $\E_{M\sim(\alpha_m)}$ with respect to the MCAR uncontaminated process, that is $\E_{M\sim(\alpha_m)}[f(M)]=\sum_{m} f(m) \alpha_m$.

\begin{thm}
\label{cor_huber}
    Under the setting adopted in Example \ref{exm_huber}, assuming $\PX = (1 - \dataeps) \Pthetastar + \dataeps \mathbb{Q}_X$, we have:
    $$
    \E_{M\sim(\alpha_m)} \left[\mathbb{D}^2\left(P_{\theta^*}^{(M)},P_{\MMDinf}^{(M)}\right)\right] \leq \frac{24\dataeps^2}{1-\missingeps} + 8 \cdot \E_{M \sim (\alpha_m)}\left[\frac{1}{\alpha_M^2}\right] \cdot \left(\frac{\missingeps}{1-\missingeps}\right)^2 .
    $$
\end{thm}
Contrary to the general Theorem \ref{MMD_Robust_cont}, the complete-data contamination ratio $\dataeps$ is only present in the first term accounting for model misspecification error. This is due to a finer analysis of the proofs in this specific mechanism contamination setting. We believe once again that the $1/(1-\missingeps)$ factor in front of the $\dataeps^2$ error may be removed, as it is the case in the one-dimensional setting below. We also provide a non-asymptotic bound in the one-dimensional case:

\begin{thm}
\label{cor_huber_finite}
    Under the setting adopted in Example \ref{exm_huber}, assuming $\PX = (1 - \dataeps) \Pthetastar + \dataeps \mathbb{Q}_X$, we have:
    $$
       \E_{\mathcal{S}}\left[ \mathbb{D}\left(P_{\MMDn},\Pthetastar\right) \right] \leq 4\cdot\dataeps + \frac{2\cdot\missingeps}{\alpha(1-\missingeps)} + \frac{2\sqrt{2}}{\sqrt{n\alpha(1-\missingeps)}}.
    $$
\end{thm}
An application of this result in the univariate Gaussian model would lead to a rate of order $\max(\missingeps/\alpha(1-\missingeps),\{n\alpha(1-\missingeps)\}^{-1/2})$ (for $|\MMDn-\theta^*|$) when $\missingeps = \dataeps$, which aligns with the minimax optimal (high-probability) rate established in Table 1 of \cite{MNARcontamination}. However, the two contamination settings are different since we are interested here in the separate contamination of the complete data marginal distribution and of the missingness mechanism, while \cite{MNARcontamination} rather focus on the contamination of the joint distribution. Nevertheless, while our contamination model for $\dataeps=\missingeps$ does not directly correspond to the so-called \textit{arbitrary contamination} setting of \cite{MNARcontamination}, our framework coincides with their \textit{realizable contamination} scenario when $\dataeps=0$, meaning that $\PX = \Pthetastar$.
Interestingly, the authors highlighted a surprising result in this case - and which seems to be unique to the Gaussian model; they demonstrated that consistency towards $\theta^*$ can still be achieved as $n \to \infty$, even if $\missingeps > 0$, and provided two interesting estimators that satisfy consistency. The first is a very simple average of the extreme values, with a convergence that is logarithmic in the sample size, specifically ${\log(n\alpha(1-\missingeps))}^{-1/2}$. While our MMD estimator does not achieve consistency towards $\theta^*$, its convergence rate towards $\MMDinf$ is in $\{(n\alpha(1-\missingeps))\}^{-1/2}$, thus outperforming the average of extremes estimator in this regard for finite samples. The blue curves in Figure \ref{fig:illustration} illustrate this behavior for the two estimators in a Gaussian mean example: while the MMD quickly suffers an error of $\missingeps=0.1$ as $n$ increases, the average of extremes moves very slowly to the true value $\theta^*=0$. 

\vspace{0.2cm}
The second estimator presented in \cite{MNARcontamination} is based on the Kolmogorov distance. Exploiting the structure of the realizable contamination, this estimator is consistent at the minimax rate derived in \cite[Theorem 8]{MNARcontamination}. This Kolmogorov estimator is thus a strong competitor in this setting, and further study is necessary to reveal the (empirical) finite sample performance compared to our MMD estimator. However, it is also tailored to this specific univariate Gaussian mean estimation problem under Huber-style MNAR contamination. In contrast, our MMD procedure is not specifically tailored to the realizable contamination model, and adds this robustness to different settings as well. We now illustrate this robustness to adversarial contamination in the missingness mechanism.


\subsection{MNAR truncation mechanism}\label{Sec_wellspecified_Truncation}

\begin{exm}
\label{exm_trunc}
    A one-dimensional dataset $\{X_i\}_i$ composed of i.i.d.\ copies of $X$ is collected, but we only observe a subsample of it $\{X_i:M_i=0\}_i$, where $M_i$ is an observed missingness random variable equal to $0$ if $X_i$ is observed, and $1$ otherwise. The true missingness mechanism is actually MNAR, with $X$ being observed only when it belongs to some large subset $\mathcal{I}$ of $\mathbb{R}$:
    $$
    \mathbb{P}[M = 0 | X] = 1 \quad \textnormal{if and only if} \hspace{0.2cm} X\in\mathcal{I} .
    $$
    We denote by $\missingeps =\mathbb{P}[M=1]=\mathbb{P}[X\notin\mathcal{I}]\in[0,1]$ the missingness level (no missingness when $\missingeps=0$).
\end{exm}

\vspace{0.2cm}
\begin{cor}
\label{cor_trunc}
    Under the setting adopted in Example \ref{exm_trunc}, we have when $\PX=\Pthetastar$:
    $$
     \MMD{P_{\MMDinf}}{\Pthetastar} \leq 4 \cdot \missingeps .
    $$
\end{cor}
In the particular case where the model is Gaussian $\Ptheta=\mathcal{N}(\theta,1)$, we have an explicit formula for the MMD distance between $\Ptheta$ and $P_{\theta'}$ using a Gaussian kernel as in \eqref{Gausskern}: 
$$
\mathbb{D}^2(P_{\theta},P_{\theta'}) = 2 \left(\frac{\gamma^2}{4 + \gamma^2}\right)^{\frac{1}{2}} \left[ 1 - \exp\left( -\frac{(\theta-\theta')^2}{4 + \gamma^2} \right) \right] ,
$$
so that
$$
|\theta-\theta'| = \sqrt{-(4+\gamma^2) \log\left( 1 - \frac{1}{2} \left(\frac{4+\gamma^2}{\gamma^2}\right)^{\frac{1}{2}} \mathbb{D}^2(P_{\theta},P_{\theta'}) \right)} .
$$
Hence, if the model is well-specified $\Pthetastar=\PX$, Corollary \ref{cor_trunc} becomes for small values of $\missingeps$:
$$
|\MMDinf-\theta^*| \leq \sqrt{-(4+\gamma^2) \log\left( 1 - 8 \left(\frac{4+\gamma^2}{\gamma^2}\right)^{\frac{1}{2}} \missingeps^2 \right)} \underset{\missingeps\rightarrow0}{\sim} \sqrt{8(4+\gamma^2) \left(\frac{4+\gamma^2}{\gamma^2}\right)^{\frac{1}{2}}} \cdot \missingeps ,
$$
and for $\gamma=\sqrt{2}$ (which minimizes the above linear factor), we have:
$$
|\MMDinf-\theta^*| \leq \sqrt{-6 \log\left( 1 - 8 \sqrt{3} \missingeps^2 \right)} \underset{\missingeps\rightarrow0}{\sim} \sqrt{48 \sqrt{3}} \cdot \missingeps .
$$
Notice that we have almost the same guarantee for the MLE which is available in closed-form:
$$
\MLEn = \frac{1}{\sum_i (1-M_i)} \sum_{i:M_i=0} X_i ,
$$
and which converges to the conditional expectation $\E[X | M = 0]$ computed as:
$$
\MLEinf = - \frac{\phi(b)-\phi(a)}{\Phi(b)-\Phi(a)},
$$
where $\phi$ is the density of the standard Gaussian and $\mathcal{I}=(a,b)$. In the particular case where $\mathcal{I}=(a,+\infty)$ (left-truncation), we thus have
$$
\MLEinf  = \frac{\phi(a)}{1-\Phi(a)},
$$
and then as a function of the missingness level $\missingeps$, the error is:
$$
\big| \MLEinf  - \theta^* \big| = \frac{\phi\left(\Phi^{-1}(\missingeps)\right)}{1-\missingeps} \underset{\missingeps\rightarrow0}{\sim} \missingeps \cdot \sqrt{\log\left(\frac{1}{2\pi\missingeps^2}\right)} ,
$$
which is slightly worse than the MMD estimator by a logarithmic factor, but with a better constant. For a treatment of $\MMDn$ in Gaussian mean estimation, for general MNAR mechanisms and dimension $d$, we refer to Appendix \ref{app_finite_sample}.

\subsection{Adversarial contamination mechanism} \label{Sec_AdversarialCont}

\begin{exm}
\label{exm_adver}
    Let us now introduce a more sophisticated adversarial contamination setting for the missingness mechanism in the one-dimensional case: from the dataset $\{X_i\}_i$ composed of i.i.d.\ copies of $X\sim\PX$, only a subsample of it $\{X_i:M_i=0\}_i$ is observed according to a MCAR missingness mechanism $\mathbb{P}[M=0|X=x]=\alpha$ for any $x$. However, the selected observed examples $M_i$'s have been modified by an adversary, and we finally observe $\{X_i:\widetilde{M}_i=0\}_i$ where at most $|\{i:M_i=0\}|\cdot\missingeps/\alpha$ (with $0<\missingeps<\alpha$) missing variables $\widetilde{M}_i$ are different from the original $M_i$. 
\end{exm}

The choice $|\{i:M_i=0\}|\cdot\missingeps/\alpha$ with $\missingeps<\alpha$ is such that the adversary cannot remove all the examples in the finally observed dataset by switching all $M_i$'s equal to $0$ to the value $\widetilde{M}_i=1$, and such that there are on average (over the sample $\{M_i\}_i$) at most $n\missingeps$ contaminated values (recall that contamination is in the missing variable $M$). The following theorem provides a bound that is valid for any finite value of the sample size $n$: 

\begin{thm}
\label{cor_adver}
    Under the setting adopted in Example \ref{exm_adver}, assuming $\PX=(1-\dataeps)\Pthetastar+\dataeps\QX$, we have:
    $$    \E_{\mathcal{S}}\left[\mathbb{D}\left(\Pthetastar,P_{\MMDn}\right)\right] \leq 4\cdot\dataeps + \frac{6\cdot\missingeps}{\alpha-\missingeps} + \frac{2\sqrt{2}}{\sqrt{n\cdot\alpha}} .
    $$
\end{thm}

Note that the rate is different between Huber's and adversarial contamination frameworks, which are by nature very different, and two terms vary: (i) First, the contamination error term, which slightly deteriorates and becomes $\missingeps/(\alpha-\missingeps)$ instead of $\missingeps/\alpha(1-\missingeps)$, and accounts for the expected number of outliers relative to the size of the uncontaminated observed dataset. In both settings, there are (on average) $n\missingeps$ outliers, but while there are $n\alpha(1-\missingeps)$ uncontaminated observed examples in Huber's model, this number becomes $n(\alpha-\missingeps)$ in the adversarial setting. (ii) Second, the rate of convergence slightly improves and becomes $\{n\alpha\}^{-1/2}$ instead of $\{n\alpha(1-\missingeps)\}^{-1/2}$. 
This difference arises because, in Huber's model, only a fraction $1-\missingeps$ of the total number $n$ of data points has been generated following the MCAR mechanism of interest, making the effective MCAR sample size $n\alpha(1-\missingeps)$. The contaminated fraction $\missingeps$ follows a different generating process and does not contribute to the MCAR sample. This is not the case in the adversarial model, where all of the initially sampled data points come from the MCAR mechanism.  

\vspace{0.2cm}
Consequences of this subtle difference can be important in the study of minimax theory for population mean estimation problems, with completely different minimax rates when compared to Huber's setting. It is very unlikely that consistency towards $\theta^*$ is possible any longer in the realizable setting where $\PX=\Pthetastar$, and very likely that optimal estimators for Huber's model behave extremely badly in this adversarial model. For instance, the average of extremes procedure becomes a terrible estimator under adversarial contamination of the mechanism: it is sufficient for the adversary to remove the $\missingeps/\alpha\cdot|\{i:M_i=0\}|-1$ smallest points and observe the largest one for a symmetric distribution to mislead the average of extremes estimator, no matter the value of the contamination ratio. At the opposite, the MMD estimator still behaves correctly in this setting, with an error which is (almost) linear in $\missingeps$. The red curves in Figure \ref{fig:illustration} illustrate the behavior of the MLE, MMD and Average of Extremes Estimator: while the MMD is stable around $\missingeps=0.1$, the average of extremes moves farther away from the true value of the mean as $n$ increases. The minimax optimality in the adversarial contamination setting is an open question.

\begin{figure}
    \centering
    \includegraphics[width=1\linewidth]{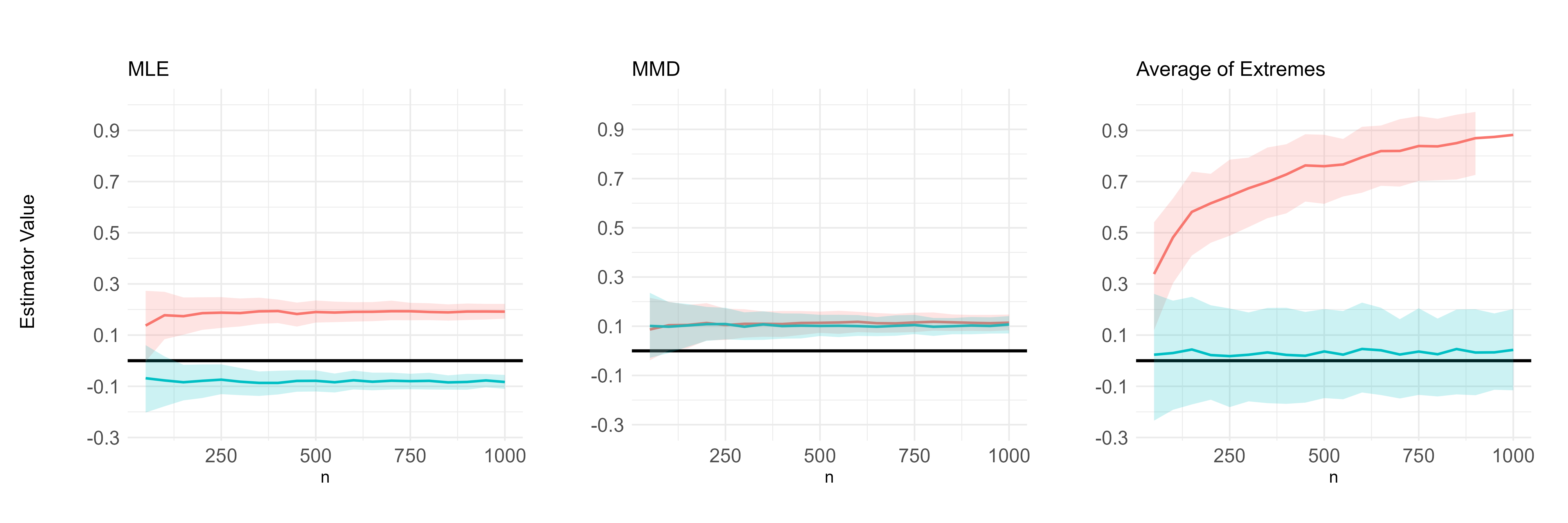}
    \caption{MLE, MMD (with Gaussian kernel) and Average of Extremes Estimator for a Gaussian mean estimation problem following Example \ref{exm_huber} (Huber contamination) in light blue and Example \ref{exm_adver} (Adversarial Contamination) in red. The shaded regions correspond to the empirical 50\% quantile range obtained by rerunning the experiment $1000$ times for each $n$. 
    In both settings, we simulate $X \sim N(0,1)$ without contamination. For the case of Huber contamination, the MCAR mechanism with $\alpha=0.5$ probability of missingness is contaminated by an MNAR mechanism for which $M=0$ iff $X > 0$. For the adversarial mechanism, each $M_i$ is first generated following an MCAR mechanism with $\alpha=0.5$ probability of missingness, and in a second step, the $M_i$'s corresponding to the $2\missingeps\cdot|\{i:M_i=0\}|-1$ smallest values are set to 1, while the $M_i$ corresponding to the largest value is set to $0$. In both cases, $\missingeps=0.1$ and $\dataeps=0$. 
    Implementation was performed in R with the help of the R package \texttt{regMMD} \citep{regMMD}.
    }
    \label{fig:illustration}
\end{figure}




\section{Simulation study} \label{Sec_Simulation}

In this section, we empirically study the problem of Gaussian mean estimation under missing values and contamination. The base data distribution is $N(\{0\}^d, I_d)$, the Gaussian distribution with identity covariance matrix. Following \cite{BadrAlquierMMD} we consider $n=500$, $d=10$ and contaminate the Gaussian distribution with a fraction $\dataeps=0.2$ of different distributions, in particular Gaussians with different means and delta measures, as shown in Table \ref{tab:results_fullcontamination}. As the baseline missing mechanism, we consider the following (blockwise) MCAR mechanism: With probability $\alpha_{m_1}$, $(X_1, X_2, X_3)$ are missing, with probability $\alpha_{m_2}$, $(X_3, X_4, X_5)$ are missing and with probability $\alpha_{m_3}$, $(X_6, X_7, X_8)$ are missing. One could think of a measurement process, where the measurements itself are independent from each other, but if one fails, the others fail automatically. Finally with probability $\alpha_{m_4}$ all variables are observed. We choose $\alpha_{m_j}=0.25$ for all $j\in[4]$, so that around $1/4$ of draws are fully observed. We then contaminate this MCAR mechanism with a fraction $\missingeps=0.2$ coming from a challenging self-censoring missingness mechanism. We adapt the self-censoring mechanism discussed in \citet{SelfCensoring}: we define $\beta=d^{-1/2}\{1\}^d=[1/\sqrt{d},\cdots,1/\sqrt{d}]^{\top}$ and 
\begin{align*}
\mathbb{Q}[M=m_1 \mid X=x]&=\mathbb{Q}[M=m_4 \mid X=x]=F(\beta^{\top}x)/2\\
\mathbb{Q}[M=m_2 \mid X=x]&=\mathbb{Q}[M=m_3 \mid X=x]=1/2-F(\beta^{\top}x)/2.
\end{align*}
Here $F$ is the cdf of $\beta^{\top}X$, where $X$ follows the contaminated distribution. For instance, if we contaminate with $N(\mu, I_d)$, $F$ will be the cdf of a mixture of $\mathcal{N}(0,\beta^{\top} \beta)=\mathcal{N}(0,1)$ and  $\mathcal{N}(\beta^{\top}\mu,1)$. This choice of parameters ensures that $\mathbb{Q}[M=m \mid X=x]$ is a valid distribution and that $\mathbb{Q}[M=m_j]=\alpha_{m_j}$ for $j \in \{1,\ldots, 4\}$. As in Section \ref{Sec_Huber_Huber}, we then have $\PMmx=(1-\missingeps) \alpha_m + \missingeps \QMmx$, with $\missingeps=0.2$. 

\vspace{0.2cm}
While mean estimation under contamination has attracted considerable research interest, less attention has been devoted to the setting with missing data. Notable recent exceptions include the approaches of \citet{Meanestimationmissing1, MNARcontamination}, although these methods appear primarily designed for theoretical exploration rather than practical implementation. As a baseline, we thus include the coordinate-wise median as a simple but robust comparator.
Under MCAR missingness, one can apply any robust method to the subset of fully observed data. We therefore incorporate the estimators ``evfiltering'' and ``QUE'' introduced in \citet{EVpruning0, EVpruning1} and \citet{QUE0}, respectively. Both methods rely on pruning observations using different criteria based on the assumed Gaussian distribution. We adopt the implementations from \citet{anderson2025robust}, who adapted these algorithms for high-dimensional, small-sample settings. These methods have been shown to perform well under various contamination scenarios, particularly for $n=500$, $d=500$.
Although we focus on $d=10$ here, both methods still perform strongly in the fully observed case. For our proposed MMD estimator, we use the Gaussian kernel with an implementation based on the \texttt{regMMD} R package \citep{regMMD}. The kernel bandwidth is selected via an adaptation of the median heuristic proposed by \citet{gretton2012optimal}.


\vspace{0.2cm}
Results for the root mean squared error (RMSE) are shown in Table \ref{tab:results_fullcontamination}. We note that in the first column, the Gaussian distribution is ``contaminated'' by the same distribution, corresponding to $\dataeps=0$. Thus, this setting corresponds to the realizable contamination model in \citet{MNARcontamination}, with contamination in the missing value mechanism only. Based on their results, we expect the sample mean to perform well in this setting. Indeed, the MLE performs extremely well here, though the MMD is not too far behind with the same performance as the median. However, as the degree of contamination is increased, the performance of the mean deteriorates rapidly, while the MMD estimate remains stable. The only exception to this is the case of $\delta_{\{1\}^d}$, which, as discussed in \cite{BadrAlquierMMD}, is a particularly bad contamination for MMD in this setting. Overall, the MMD is extremely competitive and in particular manages to beat the robust methods evfiltering and QUE. Although these methods are highly competitive, they only have access to around 1/4 of data points.

\vspace{0.2cm}
Compared to the similar simulation setting in \citet{BadrAlquierMMD} with complete data, we see that the RMSE is considerably higher throughout, indicating the increased difficulty of the problem. In particular, it appears that the effect of outliers on the mean worsens with missing data; while in Table 1 of \citet{BadrAlquierMMD} the RMSE of the mean is around 2 for the case of $N(\{10\}^d,I_d)$ and $\delta_{\{10\}^d}$, it is more than 3 times higher in Table \ref{tab:results_fullcontamination}. This is also the case without contamination of the MCAR mechanism, as shown in Appendix \ref{Sec_datacontamination}. Remarkably, the RMSE values of the MMD in Table~\ref{tab:results_fullcontamination} (with MNAR contamination) and in Table~\ref{tab:resultsMCAR} of Appendix~\ref{Sec_datacontamination} (with $\missingeps=0$) are nearly identical, suggesting that the additional MNAR contamination has little to no impact on the MMD estimator.

\begin{table}[ht]
\centering
\caption{RMSE under contamination of both the data distribution and missingness mechanism ($\missingeps=\dataeps=0.2$). Parentheses indicate estimated standard deviations.} 
\label{tab:results_fullcontamination}
\begin{tabular}{rrrrrr}
  \hline
 & $N(\{0\}^d,I_d)$  & $N(\{1\}^d,I_d)$  & $N(\{10\}^d,I_d)$ & $\delta_{\{1\}^d}$ & $\delta_{\{10\}^d}$ \\ 
  \hline
MMD & 0.2 & \textbf{0.56} & \textbf{0.25} & 0.83 & \textbf{0.26} \\ 
   & (0.06) & \textbf{(0.08)} & \textbf{(0.06)} & (0.09) & \textbf{(0.05)} \\ 
  MLE & \textbf{0.16} & 0.67 & 6.39 & \textbf{0.68} & 6.53 \\ 
   & \textbf{(0.03)} & (0.07) & (0.65) & \textbf{(0.07)} & (0.64) \\ 
  Median & 0.2 & 0.64 & 1.06 & 1.08 & 1.11 \\ 
   & (0.04) & (0.08) & (0.15) & (0.13) & (0.14) \\ 
  evfiltering & 0.3 & 0.86 & 0.36 & 0.89 & 0.37 \\ 
   & (0.06) & (0.15) & (0.11) & (0.13) & (0.09) \\ 
  QUE & 0.3 & 0.86 & 3.92 & 0.89 & 5.28 \\ 
   & (0.06) & (0.15) & (9.41) & (0.13) & (10.77) \\ 
   \hline
\end{tabular}
\end{table}



\section{Conclusion}
\label{Sec_conc}

In this paper, we adapted the concept of MMD estimation for missing values to obtain a parametric estimation procedure with provable robustness against misspecification of the data model as well as the missingness mechanism.

There are several directions for further improvements. Though we introduced a convenient SGD algorithm to use our methodology in practice, our considerations were mostly of a theoretical nature, and more empirical analysis of the performance of the estimator should be done. In addition, there is significant room for further development of the algorithm itself by leveraging techniques exploited in the complete data setting. For instance, quasi-Monte Carlo methods, as explored in \cite{niu2023discrepancy,AlquierCopulas}, could enhance efficiency. Another promising improvement involves replacing standard gradient methods with natural gradient methods, as done in \cite{briol2019statistical}, where the estimation problem is reformulated as a gradient flow using the statistical Riemannian geometry induced by the MMD metric on the parameter space. Such a stochastic natural gradient algorithm could lead to substantial computational gains compared to traditional stochastic gradient descent.

Finally as also mentioned in the introduction, our approach might also be combined with weighting approaches for M-Estimators under missing values, which might render the estimator consistent under a weaker condition on the missingness mechanism than MCAR. However, in this case the robustness properties need to be reevaluated and the advantages of this approach would need to be carefully weighted against the more complex estimation procedure.

\section*{Acknowledgements}

This work is part of the DIGPHAT project which was supported by a grant from the French government, managed by the National Research Agency (ANR), under the France 2030 program, with reference ANR-22-PESN-0017. BECA acknowledges funding from the ANR grant project BACKUP ANR-23-CE40-0018-01.

\bibliographystyle{imsart-nameyear} 
\bibliography{bibfile}       

\begin{appendix}
\section{Additional Empirical Results}\label{Sec_AdditonalEmpirical}








\subsection{Gaussian mean estimation under MCAR missingness and model contamination}\label{Sec_datacontamination}

We consider here the same setting as in the main text, but with $\missingeps=0$, i.e. the MCAR mechanism is not contaminated, such that $\Prob[M=m_{j} \mid X=x]= \alpha_{m_j}=0.25$ for $j \in \{1,\ldots, 4\}$. Again this leaves an expected fraction of 1/4 points fully observed. Results are shown in Table \ref{tab:resultsMCAR}. Even without MNAR contamination, it appears that the effect of outliers on the mean becomes worse with MCAR data; while in Table 1 of \citet{BadrAlquierMMD} the RMSE of the mean is around 2 for the case of $N(\{10\}^d,I_d)$ and $\delta_{\{10\}^d}$, it is more than 3 times higher in Table \ref{tab:resultsMCAR}. Here, the first column corresponds to no contamination ($\dataeps=\missingeps=0$). As with MNAR contamination, the MMD estimator is very competitive, only performing worse than the other robust methods in the notorious case of a $\delta_{\{1\}^d}$ contamination. Notably, the RMSE values of the MMD in Table \ref{tab:resultsMCAR} are very close to the ones in Table \ref{tab:results_fullcontamination}, indicating the MNAR contamination makes very little difference to the MMD.


\begin{table}[ht]
\centering
\caption{RMSE under contamination of the data distribution ($\dataeps=0.2$, $\missingeps=0$). Parentheses indicate estimated standard deviations.} 
\label{tab:resultsMCAR}
\begin{tabular}{rrrrrr}
  \hline
 & $N(\{0\}^d,I_d)$ & $N(\{1\}^d,I_d)$  & $N(\{10\}^d,I_d)$ & $\delta_{\{1\}^d}$ & $\delta_{\{10\}^d}$ \\ 
  \hline
  \hline
MMD & 0.21 & \textbf{0.57} & \textbf{0.25} & 0.84 & \textbf{0.24} \\ 
   & (0.05) & \textbf{(0.09)} & \textbf{(0.06)} & (0.12) & \textbf{(0.05)} \\ 
  MLE & \textbf{0.16} & 0.66 & 6.34 & \textbf{0.66} & 6.27 \\ 
   & \textbf{(0.04)} & (0.09) & (0.54) & \textbf{(0.08)} & (0.71) \\ 
  Median & 0.2 & 0.63 & 1.05 & 1.04 & 1.04 \\ 
   & (0.04) & (0.1) & (0.13) & (0.15) & (0.15) \\ 
  evfiltering & 0.28 & 0.7 & 0.33 & 0.67 & 0.33 \\ 
   & (0.06) & (0.14) & (0.08) & (0.13) & (0.06) \\ 
  QUE & 0.28 & 0.7 & 1.97 & 0.67 & 1.95 \\ 
   & (0.06) & (0.14) & (6.28) & (0.13) & (6.15) \\ 
   \hline
\end{tabular}
\end{table}

\subsection{Increased fraction of contamination}

Table \ref{tab:results_fullcontamination_03} presents the same results as in Section \ref{Sec_Simulation}, but with an increased contamination fraction $(\dataeps=\missingeps=0.3)$. Similarly, Table \ref{tab:resultsMCAR_03} shows the results under MCAR with an increased data contamination ($\dataeps=0.3$, $\missingeps=0$).

\begin{table}[ht]
\centering
\caption{RMSE under contamination of both the data distribution and missingness mechanism ($\missingeps=\dataeps=0.3$). Parentheses indicate estimated standard deviations.} 
\label{tab:results_fullcontamination_03}
\begin{tabular}{rrrrrr}
  \hline
 & $N(\{0\}^d,I_d)$  & $N(\{1\}^d,I_d)$  & $N(\{10\}^d,I_d)$ & $\delta_{\{1\}^d}$ & $\delta_{\{10\}^d}$ \\ 
  \hline
MMD & 0.21 & \textbf{0.86} & \textbf{0.3} & 1.28 & \textbf{0.29} \\ 
   & (0.05) & \textbf{(0.11)} & \textbf{(0.07)} & (0.11) & \textbf{(0.08)} \\ 
  MLE & \textbf{0.16} & 0.99 & 9.92 & \textbf{1} & 9.75 \\ 
   & \textbf{(0.03)} & (0.08) & (0.62) & \textbf{(0.08)} & (0.64) \\ 
  Median & 0.2 & 0.95 & 1.96 & 1.93 & 1.94 \\ 
   & (0.05) & (0.09) & (0.22) & (0.21) & (0.23) \\ 
  evfiltering & 0.32 & 1.24 & 0.49 & 1.31 & 0.55 \\ 
   & (0.07) & (0.16) & (0.17) & (0.17) & (0.17) \\ 
  QUE & 0.32 & 1.24 & 11.54 & 1.31 & 13.91 \\ 
   & (0.07) & (0.16) & (14.38) & (0.17) & (14.63) \\ 
   \hline
\end{tabular}
\end{table}

\begin{table}[ht]
\centering
\caption{RMSE under contamination of the data distribution ($\dataeps=0.3$, $\missingeps=0$). Parentheses indicate estimated standard deviations.} 
\label{tab:resultsMCAR_03}
\begin{tabular}{rrrrrr}
  \hline
 & $N(\{0\}^d,I_d)$  & $N(\{1\}^d,I_d)$  & $N(\{10\}^d,I_d)$ & $\delta_{\{1\}^d}$ & $\delta_{\{10\}^d}$ \\ 
  \hline
MMD & 0.21 & \textbf{0.88} & \textbf{0.3} & 1.3 & \textbf{0.29} \\ 
   & (0.04) & \textbf{(0.1)} & \textbf{(0.06)} & (0.12) & \textbf{(0.06)} \\ 
  MLE & \textbf{0.16} & 0.96 & 9.4 & \textbf{0.97} & 9.57 \\ 
   & \textbf{(0.03)} & (0.08) & (0.76) & \textbf{(0.08)} & (0.73) \\ 
  Median & 0.21 & 0.92 & 1.79 & 1.85 & 1.87 \\ 
   & (0.05) & (0.09) & (0.25) & (0.23) & (0.23) \\ 
  evfiltering & 0.28 & 0.97 & 0.38 & 1 & 0.37 \\ 
   & (0.07) & (0.12) & (0.11) & (0.14) & (0.09) \\ 
  QUE & 0.28 & 0.97 & 7.53 & 1 & 8.11 \\ 
   & (0.07) & (0.12) & (12.35) & (0.14) & (12.64) \\ 
   \hline
\end{tabular}
\end{table}

\section{Proofs and additional results}\label{app_proofs}

In this section, we first discuss the uniqueness of $\MMDinf$ under MCAR and well-specification in Section \ref{app_uniqueness}. We then provide additional finite sample results for $d \geq 1$ in Section \ref{app_finite_sample}. Finally, the remaining sections provide the proofs of the results in the main text.

\subsection{Uniqueness of $\MMDinf$ in the well-specified setting and MCAR}\label{app_uniqueness}

We define the set of minimizers
\[
\Theta_{\infty}=\arg\min_{\theta\in\Theta} \E_{M\sim\PM}\left[\MMDtwo{\PmargMtheta}{\PmargMXM} \right].
\]
We assume a well-specified setting ($\Prob_X = P_{\theta^*} \in \{\Ptheta\}_\theta$), MCAR, and that $k$ is a characteristic kernel. Then the uniqueness of $\MMDinf$ in Assumption \ref{cond_consistency1} holds if
\begin{align}\label{eq_uniquness2}
    \forall \theta, \theta' \in \Theta, \theta \neq \theta' \implies \exists m \in \{0,1\}^d \text{ with } \Prob(M=m) >0 \text{ such that } P_{\theta}^{(m)} \neq P_{\theta'}^{(m)}.
\end{align}


\begin{lemma}
    Assume MCAR and that $\PX=\Pthetastar \in \{\Ptheta\}_\theta$ (well-specified) and that the kernel $k$ is characteristic. Then $\MMDinf=\theta^*$ is unique if \eqref{eq_uniquness2} holds.
\end{lemma}

\begin{proof}
    Under MCAR and since $\PX=\Pthetastar$ by assumption,
        \begin{align*}
        \E_{M\sim\PM}\left[\MMDtwo{\PmargMthetastar}{\PmargMXM} \right]=\E_{M\sim\PM}\left[\MMDtwo{\PmargMthetastar}{\PmargMX} \right]=\E_{M\sim\PM}\left[\MMDtwo{P_{\theta^*}^{(M)}}{\PmargMthetastar} \right]=0,
    \end{align*}
    and thus $\Theta_{\infty}$ is nonempty, as it contains at least $\theta^*$. Moreover, if there were another $\theta \neq \theta^*$ minimizing the average MMD distance, \eqref{eq_uniquness2} implies that there exist marginals $P_{\theta}^{(m)}, P_{\theta^*}^{(m)}$ such that $P_{\theta}^{(m)} \neq P_{\theta^*}^{(m)}$ and thus 
    \begin{align*}
        0< \MMDtwo{P_{\theta}^{(m)}}{\PmargmX} \leq \frac{1}{\Prob(M=m)} \E_{M\sim\PM}\left[\MMDtwo{P_{\theta}^{(M)}}{\PmargMX} \right],
    \end{align*}
    under \eqref{eq_uniquness2}. This shows that $\MMDinf=\theta^*$ is the unique minimizer.
\end{proof}

In other words, for $\theta \neq \theta'$, we require that we can detect the distributional difference in at least one pattern with nonzero probability. This shows that, at least under MCAR, the required uniqueness is an interplay between the complexity of the model class $\{P_{\theta}\}_{\theta \in \Theta}$ and the support of the patterns. If $\Prob(M = \{0\}^d) > 0$, \eqref{eq_uniquness2} simplifies to
\begin{align}\label{eq_uniquness}
    \forall \theta, \theta' \in \Theta, \theta \neq \theta' \implies P_{\theta} \neq P_{\theta'},
\end{align}
just as in the case of complete data. On the other hand, if \eqref{eq_uniquness2} is not valid for $(\theta, \MMDinf)$ any such $\theta \neq \MMDinf$ inducing a \textit{compatible} distribution (a joint distribution with a specified set of given marginals \citep{berrett2023compatibility}) can minimize the loss. 

Beyond MCAR and the well-specified case, the showing the uniqueness of $\MMDinf$ is an open problem and will be studied more thoroughly in future work. However, we note that a condition such as \eqref{eq_uniquness2} is also crucial in the case of misspecification and MNAR missingness mechanism. Indeed, if there exists $(\theta, \MMDinf)$, $\theta \neq \MMDinf$, contradicting \eqref{eq_uniquness2}, then for all $m$ with $\Prob(M=m)>0$, $\| \Phi(\Pmargmtheta) - \Phi(P_{\MMDinf}^{(m)})  \|_{\H}^2=0$, and,
\begin{align*}
    \| \Phi(\Pmargmtheta) - \Phi(\PmargmXMm)   \|_{\H}^2
    &= \| \Phi(\Pmargmtheta) - \Phi(P_{\MMDinf}^{(m)}) + \Phi(P_{\MMDinf}^{(m)}) - \Phi(\PmargmXMm)   \|_{\H}^2\\
    &\leq \| \Phi(P_{\MMDinf}^{(m)}) - \Phi(\PmargmXMm)  \|_{\H}^2,
\end{align*}
where the inequality follows by multiplying out the squared norm, using Cauchy-Schwarz and then setting the respective terms to zero. Similarly, $\| \Phi(P_{\MMDinf}^{(m)}) - \Phi(\PmargmXMm)   \|_{\H}^2 \leq \| \Phi(\Pmargmtheta) - \Phi(\PmargmXMm)  \|_{\H}^2$, showing that
\begin{align*}
    \E_{M \sim \PM}\left[\MMDtwo{P_{\MMDinf}^{(M)}}{\PmargMXM}\right]=\E_{M \sim \PM}\left[\MMDtwo{P_{\theta}^{(M)}}{\PmargMXM}\right].
\end{align*}

\subsection{Finite sample result for $d > 1$}\label{app_finite_sample}

In the following we generalize the finite-sample results in Section \ref{Sec_robust} to $d > 1$. Define in the following,
\begin{align}\label{eq_Mdefinition}
    \mathcal{M}=\{m \in \{0,1\}^{d}: \Prob(M=m)>0\}
\end{align}
and let for each $m \in \mathcal{M}$, $n_m=|\{i: M_i=m\}|$ and
\begin{align*}
    P_{n_m}=\frac{1}{n_m} \sum_{i: M_i=m} \delta_{\{X_i\}}.
\end{align*}
We first note that we can rewrite the objective function defining $\MMDn$ as\footnote{We take the convention that $\frac{n_m}{n}\MMDtwo{P_{\theta}^{(m)}}{P_{n_m}^{(m)}}=0$, if $n_m=0$.}
\begin{align*}
\MMDn = \arg\min_{\theta\in\Theta} \left( \sum_{m \in \mathcal{M}} \frac{n_m}{n}\MMDtwo{P_{\theta}^{(m)}}{P_{n_m}^{(m)}} \right)^{1/2},
\end{align*}
where we normalize by $n$ instead of $|\{i: M_i \neq \{1\}^d\}|$, since we assume in the following that $\Prob(M=\{1\}^d)=0$. Inspired by this observation, we define the following empirical distance: For a sequence of distributions $(P_m)_{m}$, $(Q_m)_{m}$, we define
\[
\AMMDn{(P_m)_{m}}{(Q_m)_{m}}=\left(\sum_{m \in \mathcal{M}} \frac{n_m}{n}\MMDtwo{P_m^{(m)}}{Q_m^{(m)}}\right)^{1/2}.
\]
If $P_{m}=P$ for all $m$, we shorten the notation to
\[
\AMMDn{P}{(Q_m)_{m}}=\left(\sum_{m \in \mathcal{M}} \frac{n_m}{n}\MMDtwo{P^{(m)}}{Q_m^{(m)}}\right)^{1/2},
\]
and analogously if $Q_{m}=Q$ for all $m$. In particular it then holds that
\begin{align}\label{eq_Minimizer}
    \MMDn = \arg\min_{\theta\in\Theta} \AMMDn{P_{\theta}}{(P_{n_m})_{m}}.
\end{align}

Moreover, we can obtain a form of triangle inequality: For any sequence of probability measures $(P_m)_{m}$, $(Q_m)_{m}$, $(\tilde{Q}_m)_{m}$
\begin{align}\label{eq_triangle}
    \AMMDn{(P_m)_{m}}{(Q_m)_{m}}\leq \AMMDn{(P_m)_{m}}{(\tilde{Q}_m)_{m}} + \AMMDn{(\tilde{Q}_m)_{m}}{(Q_m)_{m}}.
\end{align}
This follows from the fact that $\AMMDn{\cdot}{\cdot}$ is the usual product metric, only with each element weighted by $n_m/n$. In addition, we will also define the population quantity,
\[
\AMMD{(P_m)_{m}}{(Q_m)_{m}}=\left(\sum_{m \in \mathcal{M}} \Prob(M=m)\MMDtwo{P_m^{(m)}}{Q_m^{(m)}}\right)^{1/2}.
\]
As before, we shorten the notation to $\AMMD{P}{Q}$, if $P_m=P$ and $Q_m=Q$ for all $m$. In the following, the goal will be to bound
\[
\E_{\mathcal{S}}\left[ \AMMD{P_{\MMDn}}{\PX} \right]=\E_{\mathcal{S}}\left[\left(\sum_{m \in \mathcal{M}} \Prob(M=m)\MMDtwo{P_{\MMDn}^{(m)}}{\PX^{(m)}}\right)^{1/2} \right].
\]
We note that if $\Prob(M=\{0\}^d) > 0$, $P=Q$ iff $\AMMD{P}{Q}=0$. In fact, if $\Prob(M=\{0\}^d) > 0$,
\begin{align*}
      \MMD{P_{\MMDn}}{\PX}  \leq \frac{1}{\sqrt{\Prob(M=\{0\}^d)}}  \AMMD{P_{\MMDn}}{\PX}.
\end{align*}

We now derive a bound for the expectation of $\AMMD{P_{\MMDn}}{\PX}$ over the sample $\mathcal{S}=\{(X_i,M_i)\}_{1\leq i\leq n}$, generalizing Theorem \ref{thm_finite_main} to $d > 1$. To simplify the derivations we assume in the following that $\Prob(M=\{1\}^d)=0$, i.e. we only consider patterns with at least one observed component.

\begin{thm}[First finite sample result for $d > 1$]
\label{thm_finite_main_d}
Assume that $\Prob(M=\{1\}^d)=0$. Then
    \begin{align*}
    \E_{\mathcal{S}}\left[ \AMMD{P_{\MMDn}}{\PX} \right]\leq& \inf_{\theta\in\Theta} \E_{M \sim \PM}\left[\MMDtwo{P_{\theta}^{(M)}}{\PmargMX}\right]^{1/2}\\
        &+ 2\E_{M \sim \PM}\left[ \frac{\Var_{X\sim \PX} \left[ \pi_M(X) \right]}{\pi_M^2} \right]^{1/2} + 4 \sqrt{\frac{|\mathcal{M}|}{n}}, 
    \end{align*}
    with $\mathcal{M}$ defined as in \eqref{eq_Mdefinition}.
\end{thm}

\begin{proof}[Proof of Theorem \ref{thm_finite_main_d}] 


We first show that 
    \begin{align}\label{eq_firstresult}
    \E_{\mathcal{S}}\left[ \AMMDn{P_{\MMDn}}{\PX} \right]  \leq & \inf_{\theta\in\Theta} \E_{M \sim \PM}\left[\MMDtwo{P_{\theta}^{(M)}}{\PmargMX}\right]^{1/2} \nonumber \\
        &+ 2\E_{M \sim \PM}\left[ \frac{\Var_{X\sim \PX} \left[ \pi_M(X) \right]}{\pi_M^2} \right]^{1/2} + 2 \sqrt{\frac{|\mathcal{M}|}{n}} . 
    \end{align}

Indeed, for any $\theta\in\Theta$:
\begin{align}\label{overall}
    \AMMDn{P_{\MMDn}}{\PX} & \leq  \AMMDn{P_{\MMDn}}{(P_{n_m})_{m}} +\AMMDn{\PX}{(P_{n_m})_{m}}\nonumber\\
    & \leq  \AMMDn{P_\theta}{(P_{n_m})_{m}} +\AMMDn{\PX}{(P_{n_m})_{m}}\nonumber\\
    & \leq \AMMDn{P_\theta}{\PX} +2\AMMDn{\PX}{(P_{n_m})_{m}} \nonumber \\
    & \leq \AMMDn{P_\theta}{\PX} +2\AMMDn{(\PXMm)_{m}}{(P_{n_m})_{m}} + 2 \AMMDn{\PX}{(\PXMm)_{m}},
\end{align}
where the first, third and last inequality follow by \eqref{eq_triangle}, the second inequality by \eqref{eq_Minimizer}. We now bound each term. We first note that
\begin{align}\label{firstpart}
    \E_{\mathcal{S}}\left[ \AMMDn{P_\theta}{\PX}\right]&\leq \E_{\mathcal{S}}\left[ \AMMDtwon{P_\theta}{\PX}\right]^{1/2}\nonumber\\
    &= \E_{\mathcal{S}}\left[ \sum_{m \in \mathcal{M}} \frac{n_m}{n}\MMDtwo{P_{\theta}^{(m)}}{\PmargmX}\right]^{1/2}\nonumber\\
    &= \left( \sum_{m \in \mathcal{M}} \Prob(M=m)\MMDtwo{P_{\theta}^{(m)}}{\PmargmX}\right)^{1/2}\nonumber\\
    &=\E_{M \sim \PM}\left[\MMDtwo{P_{\theta}^{(M)}}{\PmargMX}\right]^{1/2}.
    \end{align}
since $\E_{\mathcal{S}}\left[ n_m \right]= n \Prob(M=m)$. For the last term, we have that
\begin{align*}
    \AMMDn{\PX}{(\PXMm)_{m}}&=\left(\sum_{m \in \mathcal{M}} \frac{n_m}{n}  \MMDtwo{\PmargmXMm}{\PmargmX}\right)^{1/2}\\
    &\leq \left(\sum_{m \in \mathcal{M}} \frac{n_m}{n}  \frac{\Var_{X\sim \PX} \left[ \pi_m(X) \right]}{\pi_m^2}\right)^{1/2}
\end{align*}
since, using the same argument as in Theorem \ref{MMD_Robust}:
\begin{align*}
    \mathbb{D}\left(\PmargmX,\PmargmXMm\right)
    %
    & \leq \frac{\sqrt{\Var_{X\sim \PX} \left[ \pi_m(X) \right]}}{\pi_m} .
\end{align*}


Thus, it holds that
\begin{align}\label{secondpart}
    \E_{\mathcal{S}}\left[\AMMDn{\PX}{(\PXMm)_{m}}\right]&\leq \E_{\mathcal{S}}\left[\AMMDtwon{\PX}{(\PXMm)_{m}}\right]^{1/2}\nonumber\\
    &\leq \E_{\mathcal{S}}\left[\sum_{m \in \mathcal{M}} \frac{n_m}{n}  \frac{\Var_{X\sim \PX} \left[ \pi_m(X) \right]}{\pi_m^2}\right]^{1/2}\nonumber\\
    &=\left(\sum_{m \in \mathcal{M}} \Prob(M=m)  \frac{\Var_{X\sim \PX} \left[ \pi_m(X) \right]}{\pi_m^2}\right)^{1/2}\nonumber\\
        &=\E_{M \sim \PM}\left[ \frac{\Var_{X\sim \PX} \left[ \pi_M(X) \right]}{\pi_M^2} \right]^{1/2}
\end{align}

Finally, the second term can be controlled in expectation as follows. We first write:
\begin{align*}
    &\frac{n_m}{n}\MMDtwo{\PmargmXMm}{P_{n_m}^{(m)}} \\
    & = \frac{n_m}{n}\left\lVert \E_{X^{(m)}\sim \PmargmXMm}\left[\Phi(X^{(m)})\right] - \frac{1}{n_m} \sum_{i:M_i=m} \Phi(X_i^{(m)}) \right\rVert_{\mathcal{H}}^2 \\
    & = \frac{n_m}{n}\left\lVert \frac{1}{n_m} \sum_{i:M_i=m} \left\{ \Phi(X_i^{(m)}) - \E_{X^{(m)}\sim\PmargmXMm}\left[\Phi(X^{(m)})\right] \right\} \right\rVert_{\mathcal{H}}^2 \\
    & = \frac{1}{n n_m} \sum_{i:M_i=m} \left\lVert \Phi(X_i^{(m)}) - \E_{X^{(m)}\sim\PmargmXMm}\left[\Phi(X^{(m)})\right] \right\rVert_{\mathcal{H}}^2 \\
    & + \frac{2}{n n_m} \sum_{i<j:M_i=M_j=m} \left\langle \Phi(X_i^{(m)}) - \E_{X^{(m)}\sim\PmargmXMm}\left[\Phi(X^{(m)})\right],\Phi(X_j^{(m)}) - \E_{X^{(m)}\sim\PmargmXMm}\left[\Phi(X^{(m)})\right] \right\rangle_{\mathcal{H}} .
\end{align*}
The key point now lies in the fact that conditional on $\{M_i\}_{1\leq i \leq n}$, all the $X_i$'s such that $M_i=m$ are i.i.d.\ with common distribution $\PmargmXMm$. Then, in expectation over the sample $\mathcal{S}=\{(X_i,M_i)\}_{1\leq i \leq n}$, the first term in the last line becomes:
\begin{align*}
    \E_{\mathcal{S}} \bigg[ & \frac{1}{n_m n} \sum_{i:M_i=m} \left\lVert \Phi(X_i^{(m)}) - \E_{X^{(m)}\sim\PmargmXMm}\left[\Phi(X^{(m)})\right] \right\rVert_{\mathcal{H}}^2\bigg] \\
    & = \E_{\{M_i\}_{1\leq i \leq n}\sim\PM} \left[ \frac{1}{n_m n} \sum_{i:M_i=m} \E_{X_i^{(m)}\sim\PmargmXMm}\left[ \left\lVert \Phi(X_i^{(m)}) - \E_{X^{(m)}\sim\PmargmXMm}\left[\Phi(X^{(m)})\right] \right\rVert_{\mathcal{H}}^2\right] \right] \\
    & = \E_{\{M_i\}_{1\leq i \leq n}\sim\PM} \left[ \frac{1}{n} \cdot \E_{X^{(m)}\sim\PmargmXMm}\left[ \left\lVert \Phi(X^{(m)}) - \E_{X^{(m)}\sim\PmargmXMm}\left[\Phi(X^{(m)})\right] \right\rVert_{\mathcal{H}}^2\right] \right] \\
    & \leq \E_{\{M_i\}_{1\leq i \leq n}\sim\PM} \left[ \frac{1}{n} \cdot \E_{X^{(m)}\sim\PmargmXMm}\left[ \left\lVert \Phi(X^{(m)}) \right\rVert_{\mathcal{H}}^2\right] \right] \\
    & \leq \frac{1}{n},
\end{align*}
where the first inequality comes from the standard inequality $\textnormal{Variance}\leq\textnormal{Second order moment}$, the second one from the boundedness of the kernel.

\vspace{0.2cm}
Similarly, we have:

\begin{align*}
    \E_{\mathcal{S}} \bigg[ & \frac{2}{n_m n} \sum_{i<j:M_i=M_j=m} \left\langle \Phi(X_i^{(m)}) - \E_{X^{(m)}\sim\PmargmXMm}\left[\Phi(X^{(m)})\right],\Phi(X_j^{(m)}) - \E_{X^{(m)}\sim\PmargmXMm}\left[\Phi(X^{(m)})\right] \right\rangle_{\mathcal{H}} \bigg] \\
    & = \E_{\{M_i\}_{1\leq i \leq n}\sim\PM} \bigg[ \frac{2}{n_m n} \sum_{i<j:M_i=M_j=m} \left\langle 0, 0 \right\rangle_{\mathcal{H}} \bigg] \\
    & = 0,
\end{align*}
so that
\begin{align}\label{thirdpart}
    \E_{\mathcal{S}} [\AMMDn{(\PXMm)_{m}}{(P_{n_m})_{m}}] &\leq \E_{\mathcal{S}} [\AMMDtwon{(\PXMm)_{m}}{(P_{n_m})_{m}}]^{1/2}\nonumber\\
    &=\left(\sum_{m \in \mathcal{M}}  \E_{\mathcal{S}} \left[\frac{n_m}{n}\MMDtwo{\PmargmXMm}{P_{n_m}^{(m)}}\right]\right)^{1/2}\nonumber\\
    &\leq \left(\sum_{m \in \mathcal{M}} \frac{1}{n}\right)^{1/2}
\end{align}

Connecting \eqref{overall} with \eqref{firstpart}--\eqref{thirdpart} proves \eqref{eq_firstresult}.



For the final result, let $\widehat{\pi}_m=n_m/n$ and $\pi_m=\PMm$. We have,
\[
\begin{aligned}
\left| \AMMDn{P_{\MMDn}}{\PX} - \AMMD{P_{\MMDn}}{\PX}\right|
&=
\left|
\left(\sum_{m\in\mathcal{M}}\widehat{\pi}_m \MMDtwo{P_{\MMDn}^{(m)}}{\PmargmX}\right)^{1/2}
-
\left(\sum_{m\in\mathcal{M}}\pi_m \MMDtwo{P_{\MMDn}^{(m)}}{\PmargmX}\right)^{1/2}
\right| \\
&\leq
\left(
\sum_{m\in\mathcal{M}}
(\sqrt{\widehat{\pi}_m}-\sqrt{\pi_m})^2 \MMDtwo{P_{\MMDn}^{(m)}}{\PmargmX}
\right)^{1/2},
\end{aligned}
\]
where the last inequality can be seen by an application of the reverse triangle inequality.

Since $\MMDtwo{P_{\MMDn}^{(m)}}{\PmargmX}\leq 4$, we have
\begin{align}\label{eq_ADvsADn}
    \left| \AMMDn{P_{\MMDn}}{\PX} - \AMMD{P_{\MMDn}}{\PX}\right|
\leq
2\,\left(
\sum_{m\in\mathcal{M}}
(\sqrt{\widehat{\pi}_m}-\sqrt{\pi_m})^2
\right)^{1/2}.
\end{align}
Moreover, since $\pi_m>0$ for all $m\in\mathcal{M}$,
\[
(\sqrt{\widehat{\pi}_m}-\sqrt{\pi_m})^2
=
\frac{(\widehat{\pi}_m-\pi_m)^2}
{(\sqrt{\widehat{\pi}_m}+\sqrt{\pi_m})^2}
\leq
\frac{(\widehat{\pi}_m-\pi_m)^2}{\pi_m}.
\]
Consequently, by Jensen's inequality,
\begin{align}\label{eq_ADnmADexpectationbound}
\mathbb{E}_{\mathcal{S}}\left[
\left| \AMMDn{P_{\MMDn}}{\PX} - \AMMD{P_{\MMDn}}{\PX}\right|
\right]
& \leq 
2\mathbb{E}_{\mathcal{S}}\left[\left(
\sum_{m\in\mathcal{M}}
\frac{
\left[(\widehat{\pi}_m-\pi_m)^2\right]}
{\pi_m}
\right)^{1/2}\right] \nonumber \\
&\leq
2\left(
\sum_{m\in\mathcal{M}}
\frac{\mathbb{E}_{\mathcal{S}}
\left[(\widehat{\pi}_m-\pi_m)^2\right]}
{\pi_m}
\right)^{1/2}\nonumber \\
&=
2\left(
\frac{1}{n}
\sum_{m\in\mathcal{M}}(1-\pi_m)
\right)^{1/2}
\leq
2\sqrt{\frac{|\mathcal{M}|}{n}},
\end{align}
where we used
\[
\mathbb{E}_{\mathcal{S}}
\left[(\widehat{\pi}_m-\pi_m)^2\right]
=
\frac{\pi_m(1-\pi_m)}{n},
\]
which ends the proof.

\end{proof}

The bound in Theorem \ref{thm_finite_main_d} again offers a detailed breakdown of the three types of errors involved:
$$
\inf_{\theta\in\Theta} \E_{M \sim \PM}\left[\MMDtwo{P_{\theta}^{(M)}}{\PmargMX}\right]^{1/2},  \qquad \E_{M \sim \PM}\left[ \frac{\Var_{X\sim \PX} \left[ \pi_M(X) \right]}{\pi_M^2} \right]^{1/2}, \qquad \sqrt{\frac{|\mathcal{M}|}{n}} \, .
$$
The first term represents the misspecification error arising from using an incorrect model. The second term reflects the error due to incorrectly assuming that the missingness mechanism follows the M(C)AR assumption. Finally, the third term accounts for the estimation error, which stems from the finite size of the dataset and the presence of unobserved data points. Compared to the bound in Theorem \ref{thm_finite_main}, where the pattern $M=\{1\}^d$ was allowed, here the denominator is $n$, not $n\pi$. On the other hand, the $\sqrt{2}$ in the numerator of the third term now gets replaced by $\sqrt{|\mathcal{M}|}$. In the worst case, when every pattern has nonzero probability $|\mathcal{M}|=2^d-1$ (excluding the fully unobserved pattern). However, depending on the support of $M$, it might also be substantially smaller. Interestingly, the bound does not depend on the number of samples per patterns. That is, a pattern $m$ with only a small number of observations does not have an adverse effect on the estimation error, as long as the overall number of patterns is ``small'' relative to $n$. On the other hand, if $\pi_m$ is itself small, the second error will be relatively large.

Let in the following
\[
\mathcal{E}=\inf_{\theta\in\Theta} \E_{M \sim \PM}\left[\MMDtwo{P_{\theta}^{(M)}}{\PmargMX}\right]^{1/2} + 2\E_{M \sim \PM}\left[ \frac{\Var_{X\sim \PX} \left[ \pi_M(X) \right]}{\pi_M^2} \right]^{1/2}.
\]

We now extend the expectation bound into a probabilistic bound on $\AMMD{P_{\MMDn}}{\Prob_X}$ directly, following the approach in \cite{briol2019statistical,BadrAlquierMMD}. Adapting their arguments to our case takes more care, but combining the different bounds correctly and using McDiarmid's inequality \citep{McDiarmid1989} together with the new finite sample result, we obtain the following.

\begin{thm}[Second finite sample result for $d > 1$]
\label{thm_finite_main_d2}
   Assume that $\Prob(M=\{1\}^d)=0$ and let $\ubar{\pi}=\min_{m \in \mathcal{M}} \Prob(M=m)$. Then for any $\Delta \in (0,1)$, 
\begin{align*}
    \Prob\left(\AMMD{P_{\MMDn}}{\Prob_X} >\mathcal{E}+\sqrt{\frac{|\mathcal{M}|}{n}}\left( \sqrt{8\log(2/\Delta)} + 18 \right)+ 8 \sqrt{\frac{\log(2/\Delta)}{n \ubar{{\pi}}}}  \right) < \Delta.
\end{align*}
\end{thm}

\begin{proof}[Proof of Theorem \ref{thm_finite_main_d2}]

We first note that for any $\theta$,
\begin{align}\label{main_bound}
    &\AMMD{P_{\MMDn}}{\Prob_X} \nonumber \\
    &\leq | \AMMD{P_{\MMDn}}{\Prob_X} - \AMMDn{P_{\MMDn}}{\Prob_X}| + \AMMDn{P_{\MMDn}}{\Prob_X}\nonumber \\
    & \leq | \AMMD{P_{\MMDn}}{\Prob_X} - \AMMDn{P_{\MMDn}}{\Prob_X}|+ \AMMDn{P_\theta}{\PX} +2\AMMDn{(\PXMm)_{m}}{(P_{n_m})_{m}}\nonumber\\
    &+ 2 \AMMDn{\PX}{(\PXMm)_{m}}\nonumber\\
    & \leq Z_n + 2\AMMDn{(\PXMm)_{m}}{(P_{n_m})_{m}}+ \AMMD{P_\theta}{\PX} + 2 \AMMD{\PX}{(\PXMm)_{m}},
\end{align}
where the second inequality followed from \eqref{overall} and where $Z_n$ is defined as,
\begin{align*}
    Z_n&=| \AMMD{P_{\MMDn}}{\Prob_X} - \AMMDn{P_{\MMDn}}{\Prob_X}| + |\AMMDn{P_\theta}{\PX}-\AMMD{P_\theta}{\PX}| \\
    & + 2|\AMMD{\PX}{(\PXMm)_{m}} -\AMMDn{\PX}{(\PXMm)_{m}}|.
\end{align*}
The last two terms in \eqref{main_bound} can be bounded as before:
\begin{align*}
    \AMMD{P_\theta}{\PX} &=\E_{M \sim \PM}\left[\MMDtwo{P_{\theta}^{(M)}}{\PmargMX}\right]^{1/2}, \\
\AMMD{\PX}{(\PXMm)_{m}} &\leq \sum_{m \in \mathcal{M}} \Prob(M=m) \frac{\Var_{X\sim \PX} \left[ \pi_m(X) \right]}{\pi_m^2}.
\end{align*}
Since \eqref{main_bound} holds for any $\theta$, we can thus replace the sum of $\AMMD{P_\theta}{\PX}$ and $\AMMD{\PX}{(\PXMm)_{m}}$ by $\mathcal{E}$.
We can further bound $Z_n$ as 
\[
|Z_n| \leq
8\,
\left(\sum_{m\in\mathcal{M}}
(\sqrt{\widehat{\pi}_m}-\sqrt{\pi_m})^2 \right)^{1/2}\leq 8 \left(\sum_m \frac{(\widehat{\pi}_m-\pi_m)^2}{\pi_m}\right)^{1/2}=8f(M_1, \ldots, M_n),
\]
as shown above after Equation \eqref{eq_ADvsADn}. Applying McDiarmid's inequality to $f$, we first notice that $f(M_1, \ldots, M_n)=\| z\|$, where $\| \cdot \|$ is the usual product norm and $z$ is a vector with elements $\left( (\widehat{\pi}_m-\pi_m)/\sqrt{\pi_m}\right)_{m}$. Thus if $M_i=m$ gets changed to $M_i'=m'$, using the reverse triangle inequality,
\begin{align*}
    &| f(M_1, \ldots, M_i, \ldots, M_n) - f(M_1, \ldots, M_i', \ldots, M_n)|\\
&\leq \left( \left| \frac{(\widehat{\pi}_m-\pi_m)}{\sqrt{\pi_m}} - \frac{(\widehat{\pi}_m-1/n-\pi_m)}{\sqrt{\pi_m}} \right|^2 + \left| \frac{(\widehat{\pi}_{m'}-\pi_{m'})}{\sqrt{\pi_{m'}}} - \frac{(\widehat{\pi}_{m'}+1/n-\pi_{m'})}{\sqrt{\pi_{m'}}} \right|^2\right)^{1/2}\\
&=\frac{1}{n}\left(\frac{1}{\pi_m} + \frac{1}{\pi_{m'}}\right)^{1/2}\\
&\leq \frac{1}{n} \sqrt{\frac{2}{\ubar{\pi}}}.
\end{align*}
We thus have 
\begin{align*}
    \Prob\left( f(M_1, \ldots, M_i, \ldots, M_n) - \E_{\mathcal{S}} [f(M_1, \ldots, M_i, \ldots, M_n)] > t  \right) \leq \exp\left(- n\ubar{\pi} t^2 \right).
\end{align*}
Together with the bound derived on the expectation of $f$ in \eqref{eq_ADnmADexpectationbound}, it holds with probability at least $1-\Delta$ that,
\begin{align}\label{eq_Znprobabilitybound}
    |Z_n| &\leq 8 f(M_1, \ldots, M_n) \nonumber \\
    & \leq 8 \sqrt{\frac{\log(1/\Delta)}{n \ubar{{\pi}}}} + 16\sqrt{\frac{|\mathcal{M}|}{n}}.
\end{align}

Next, we apply the McDiarmid's inequality to $f(X_1, \ldots, X_n)=\AMMDn{(\PXMm)_{m}}{(P_{n_m})_{m}}$ conditioned on $M_1, \ldots, M_n$. Again, for the vector $z$ with elements $(\sqrt{n_m/n} \MMD{\PmargmXMm}{P_{n_m}} )_m$, we have $\AMMDn{\PmargmXMm}{(P_{n_m})_{m}}=\| z \|$, where $\| \cdot \|$ is the usual product norm. Thus, for $z'$ with elements $(\sqrt{n_m/n} \MMD{\PmargmXMm}{P_{n_m}'} )_m$, and using the reverse triangle inequality,
\begin{align*}
    &\left |\AMMDn{(\PXMm)_{m}}{(P_{n_m})_{m}} - \AMMDn{(\PXMm)_{m}}{(P_{n_m})_{m}'}\right|\\ 
    & \leq \sqrt{\frac{n_m}{n}} |\MMD{P_{n_m}}{\PmargmXMm} - \MMD{P_{n_m}'}{\PmargmXMm}  |\\
    &\leq \frac{\sqrt{n_m/n}}{n_m}\left\lVert   \Phi(X_i^{(m)})- \Phi(X_i'^{(m)})  \right\rVert_{\mathcal{H}}\\
& \leq \sqrt{\frac{1}{n n_m}}\left( \left\lVert   \Phi(X_i^{(m)})\right \rVert_{\mathcal{H}}+ \left\lVert\Phi(X_i'^{(m)})  \right\rVert_{\mathcal{H}}\right)\\
& \leq \sqrt{\frac{4}{n n_m}}.
\end{align*}
We thus have that
\begin{align*}
  &\Prob\left( \AMMDn{(\PXMm)_{m}}{(P_{n_m})_{m}} - \E_{\mathcal{S}} [\AMMDn{(\PXMm)_{m}}{(P_{n_m})_{m}} \mid M_1, \ldots, M_n] > t \mid M_1, \ldots, M_n \right)\\
  &\leq \exp\left(- \frac{2 t^2}{\frac{4}{n} |\mathcal{M}| } \right),
\end{align*}
and since the right-hand side does not depend on $M_1, \ldots, M_n$, also
\begin{align*}
  \Prob\left( \AMMDn{(\PXMm)_{m}}{(P_{n_m})_{m}} - \E_{\mathcal{S}} [\AMMDn{(\PXMm)_{m}}{(P_{n_m})_{m}} \mid M_1, \ldots, M_n] > t  \right) \leq \exp\left(- \frac{2 t^2}{\frac{4}{n} |\mathcal{M}| } \right),
\end{align*}
Moreover, as before, since
\begin{align*}
    &\frac{n_m}{n}\MMDtwo{\PmargmXMm}{P_{n_m}^{(m)}} \\
    & = \frac{1}{n n_m} \sum_{i:M_i=m} \left\lVert \Phi(X_i^{(m)}) - \E_{X^{(m)}\sim\PmargmXMm}\left[\Phi(X^{(m)})\right] \right\rVert_{\mathcal{H}}^2 \\
    & + \frac{2}{n n_m} \sum_{i<j:M_i=M_j=m} \left\langle \Phi(X_i^{(m)}) - \E_{X^{(m)}\sim\PmargmXMm}\left[\Phi(X^{(m)})\right],\Phi(X_j^{(m)}) - \E_{X^{(m)}\sim\PmargmXMm}\left[\Phi(X^{(m)})\right] \right\rangle_{\mathcal{H}},
\end{align*}
it holds that,
\begin{align*}
    \E_{\mathcal{S}} \bigg[ & \frac{1}{n_m n} \sum_{i:M_i=m} \left\lVert \Phi(X_i^{(m)}) - \E_{X^{(m)}\sim\PmargmXMm}\left[\Phi(X^{(m)})\right] \right\rVert_{\mathcal{H}}^2\mid M_1, \ldots, M_n\bigg] \\
    & =  \frac{1}{n_m n} \sum_{i:M_i=m} \E_{X_i^{(m)}\sim\PmargmXMm}\left[ \left\lVert \Phi(X_i^{(m)}) - \E_{X^{(m)}\sim\PmargmXMm}\left[\Phi(X^{(m)})\right] \right\rVert_{\mathcal{H}}^2\right]  \\
    & = \frac{1}{n} \cdot \E_{X^{(m)}\sim\PmargmXMm}\left[ \left\lVert \Phi(X^{(m)}) - \E_{X^{(m)}\sim\PmargmXMm}\left[\Phi(X^{(m)})\right] \right\rVert_{\mathcal{H}}^2\right] \\
    & \leq  \frac{1}{n} \cdot \E_{X^{(m)}\sim\PmargmXMm}\left[ \left\lVert \Phi(X^{(m)}) \right\rVert_{\mathcal{H}}^2\right] \\
    & \leq \frac{1}{n},
\end{align*}
and
\begin{align*}
    \E_{\mathcal{S}} \bigg[ & \frac{2}{n_m n} \sum_{i<j:M_i=M_j=m} \left\langle \Phi(X_i^{(m)}) - \E_{X^{(m)}\sim\PmargmXMm}\left[\Phi(X^{(m)})\right],\Phi(X_j^{(m)}) - \E_{X^{(m)}\sim\PmargmXMm}\left[\Phi(X^{(m)})\right] \right\rangle_{\mathcal{H}} \mid M_1, \ldots, M_n \bigg] \\
    & =  \frac{2}{n_m n} \sum_{i<j:M_i=M_j=m} \left\langle 0, 0 \right\rangle_{\mathcal{H}}  \\
    & = 0,
\end{align*}
so that
\begin{align}\label{thirdpart2}
    \E_{\mathcal{S}} [\AMMDn{(\PXMm)_{m}}{(P_{n_m})_{m}} \mid M_1, \ldots, M_n] &\leq \E_{\mathcal{S}} [\AMMDtwon{(\PXMm)_{m}}{(P_{n_m})_{m}}\mid M_1, \ldots, M_n]^{1/2}\nonumber\\
    &=\left(\sum_{m \in \mathcal{M}}  \E_{\mathcal{S}} \left[\frac{n_m}{n}\MMDtwo{\PmargmXMm}{P_{n_m}^{(m)}}\mid M_1, \ldots, M_n\right]\right)^{1/2}\nonumber\\
    &\leq \left(\sum_{m \in \mathcal{M}} \frac{1}{n}\right)^{1/2}.
\end{align}
Thus, with probability at least $1-\Delta$,
\begin{align}\label{eq_AMMDnprobabilitybound}
     \AMMDn{(\PXMm)_{m}}{(P_{n_m})_{m}} \leq \sqrt{\frac{|\mathcal{M}|}{n}} + \sqrt{\frac{2\log(1/\Delta)|\mathcal{M}|}{n}}.
\end{align}
Combining \eqref{main_bound}, \eqref{eq_Znprobabilitybound} and \eqref{eq_AMMDnprobabilitybound}, we obtain with probability at least $1-\Delta$,
\begin{align*}
        &\AMMD{P_{\MMDn}}{\Prob_X} \\
    & \leq Z_n + 2\AMMDn{(\PXMm)_{m}}{(P_{n_m})_{m}}+ \AMMD{P_\theta}{\PX} + 2 \AMMD{\PX}{(\PXMm)_{m}}\\
    & \leq  8 \sqrt{\frac{\log(2/\Delta)}{n \ubar{{\pi}}}}   + 2\sqrt{\frac{2\log(2/\Delta)|\mathcal{M}|}{n}} + 18\sqrt{\frac{|\mathcal{M}|}{n}}+ \mathcal{E}\\
        & \leq \sqrt{\frac{|\mathcal{M}|}{n}}\left( \sqrt{8\log(2/\Delta)} + 18 \right)+ 8 \sqrt{\frac{\log(2/\Delta)}{n \ubar{{\pi}}}}  + \mathcal{E},
\end{align*}
using a union bound in the second step.


\end{proof}

We see the same error decomposition as in the case of the expectation in Theorem \ref{thm_finite_main_d}, including the dependence on the number of patterns $|\mathcal{M}|$. Again, the bound does not depend on the number of samples per patterns, though now the smallest positive probability $\ubar{\pi}$ enters the bound.

Finally, we study a concrete example, with $P_{\theta}=N(\theta, I)$, $\PX=\Pthetastar=N(\theta^*, I)$ and $\Prob(M=\{0\}^d) > 0$. Combining Theorem \ref{thm_finite_main_d2} with the fact that
\[
\MMDtwo{P_{\MMDn}}{\PX}  \leq \frac{1}{\Prob(M=\{0\}^d)}  \AMMDtwo{P_{\MMDn}}{\PX},
\]
and, for $P_{\theta}=N(\theta, I)$, $P_{\theta'}=N(\theta', I)$,
\[
\|\theta-\theta'\|^2 \leq  -(4+\gamma^2) \log\left( 1 - \frac{1}{2} e^{\frac{2d}{\gamma^2}}\mathbb{D}^2(P_{\theta},P_{\theta'}) \right),
\]
as derived in \cite{BadrAlquierMMD} after Equation (6) of the arXiv version, we have that, with probability at least $1-\Delta$,
\begin{align*}
    &\|\theta_n-\theta^*\|^2 \leq  -(4+\gamma^2) \\
    & \log\left( 1 -  e^{\frac{2d}{\gamma^2}} \frac{1}{2\Prob(M=\{0\}^d)}\left( 2\E_{M \sim \PM}\left[ \frac{\Var_{X\sim \PX} \left[ \pi_M(X) \right]}{\pi_M^2} \right]^{1/2} + \sqrt{\frac{|\mathcal{M}|}{n}}\left( \sqrt{8\log(2/\Delta)} + 18 \right)+ 8 \sqrt{\frac{\log(2/\Delta)}{n \ubar{{\pi}}}} \right)^2 \right).
\end{align*}
If $|\mathcal{M}|$ is of (much) lower order than $d$, the estimator performs roughly the same as the estimator without missing values, as described in \cite[Proposition 4.1]{BadrAlquierMMD}. We refer to the discussion in the former paper, in particular, to the discussion surrounding Figure 1. If $|\mathcal{M}|$ is of the order $d$ or higher, then the rate can be substantially worse, although this appears to be unavoidable.

\subsection{Proofs from Section \ref{sec_asymp}}

In order to simplify the proofs, we first introduce the following notations:
$$
\widehat{L}_n(\theta)  = \frac{1}{n} \sum_{i=1}^n \ell(X_i,M_i;\theta) = \frac{1}{n} \sum_{i=1}^n \MMDtwo{P_\theta^{(M_i)}}{\delta_{\{X_i^{(M_i)}\}}} - \frac{1}{n} \sum_{i=1}^n \left\|\Phi(X_i^{(M_i)})\right\|_{\mathcal{H}}^2 
$$ 
and
$$
L(\theta) = \mathbb{E}_{(M,X)}\left[\ell(M,X;\theta)\right] = \mathbb{E}_{(X,M)}\left[ \MMDtwo{P_\theta^{(M)}}{\delta_{\{X^{(M)}\}}}\right] - \mathbb{E}_{(X,M)}\left[\left\|\Phi(X^{(M)})\right\|_{\mathcal{H}}^2\right] ,
$$
where 
$$
\ell(x,m;\theta) = \left\langle \Phi(P_\theta^{(m)}) , \Phi(P_\theta^{(m)}) \right\rangle_{\mathcal{H}} - 2 \left\langle \Phi(P_\theta^{(m)}) , \Phi(x^{(m)}) \right\rangle_{\mathcal{H}} . 
$$
Hence, we have
$$
\MMDn = \arg\min_{\theta\in\Theta} \widehat{L}_n(\theta)= \arg\min_{\theta\in\Theta} \frac{1}{n} \sum_{i=1}^n \MMDtwo{P_\theta^{(M_i)}}{\delta_{\{X_i^{(M_i)}\}} }
$$
and
$$
\MMDinf = \arg\min_{\theta\in\Theta} L(\theta) = \arg\min_{\theta\in\Theta} \E_{M}\left[\MMDm{2}{\PmargMtheta}{\PmargMXM} \right] .
$$

\subsubsection{Intermediate lemmas}

We first show an almost sure uniform law of large numbers. This is a remarkable result as it holds without any assumption on the model.

\begin{lemma}
\label{lemma_ULLN}
The empirical loss $\widehat{L_n}$ almost surely converges uniformly to $L$, i.e.\:
$$
\sup_{\theta\in\Theta} |\widehat{L_n}(\theta) - L(\theta)| \xrightarrow[\ n \to +\infty]{\Pjoint-\textnormal{a.s.}} 0 .
$$
\end{lemma}

\begin{proof}
The proof is straightforward. We start from:
\begin{align*}
    & \widehat{L_n}(\theta) - L(\theta) \\
    & = \frac{1}{n} \sum_{i=1}^n \left\langle \Phi(P_\theta^{(M_i)}) , \Phi(P_\theta^{(M_i)}) \right\rangle_{\mathcal{H}} - \mathbb{E}_{M} \left[ \left\langle \Phi(P_\theta^{(M)}) , \Phi(P_\theta^{(M)}) \right\rangle_{\mathcal{H}} \right] \\
    & \quad + 2 \cdot \mathbb{E}_{(X,M)} \left[ \left\langle \Phi(P_\theta^{(M)}) , \Phi(X^{(M)}) \right\rangle_{\mathcal{H}} \right] - \frac{2}{n} \sum_{i=1}^n \left\langle \Phi(P_\theta^{(M_i)}) , \Phi(X_i^{(M_i)}) \right\rangle_{\mathcal{H}} \\
    & = \sum_{m\in\{0,1\}^d} \left\{ \frac{1}{n} \sum_{i=1}^n \mathds{1}(M_i=m) - \mathbb{P}[M=m] \right\} \left\langle \Phi(P_\theta^{(m)}) , \Phi(P_\theta^{(m)}) \right\rangle_{\mathcal{H}} \\
    & \quad + 2 \sum_{m\in\{0,1\}^d} \left\{ \mathbb{P}[M=m] \left\langle \Phi(P_\theta^{(m)}) , \Phi(\PmargmXMm) \right\rangle_{\mathcal{H}} - \frac{1}{n} \sum_{i=1}^n \mathds{1}(M_i=m) \left\langle \Phi(P_\theta^{(m)}) , \Phi(X_i^{(m)}) \right\rangle_{\mathcal{H}} \right\} \\
   & = \sum_{m\in\{0,1\}^d} \left\{ \frac{1}{n} \sum_{i=1}^n \mathds{1}(M_i=m) - \mathbb{P}[M=m] \right\} \left\| \Phi(P_\theta^{(m)}) \right\|_{\mathcal{H}}^2 \\
    & \quad + 2 \sum_{m\in\{0,1\}^d} \left\langle \Phi(P_\theta^{(m)}) , \mathbb{P}[M=m] \Phi(\PmargmXMm) - \frac{1}{n} \sum_{i=1}^n \mathds{1}(M_i=m) \Phi(X_i^{(m)}) \right\rangle_{\mathcal{H}} .
\end{align*}
Then, using the triangle inequality, Cauchy-Schwarz's inequality and the boundedness of the kernel:
\begin{align*}
    & \sup_{\theta\in\Theta} \left| \widehat{L_n}(\theta) - L(\theta) \right| \\
    & \leq \sum_{m\in\{0,1\}^d} \left| \mathbb{P}[M=m] - \frac{1}{n} \sum_{i=1}^n \mathds{1}(M_i=m) \right| \sup_{\theta\in\Theta} \left\| \Phi(P_\theta^{(m)}) \right\|_{\mathcal{H}}^2 \\
    & \quad \quad \quad + 2 \sum_{m\in\{0,1\}^d} \sup_{\theta\in\Theta} \left\| \Phi(P_\theta^{(m)}) \right\|_{\mathcal{H}} \cdot \left\| \mathbb{P}[M=m] \Phi(\PmargmXMm) - \frac{1}{n} \sum_{i=1}^n \mathds{1}(M_i=m) \Phi(X_i^{(m)}) \right\|_{\mathcal{H}} \\
    & \leq \sum_{m} \underbrace{\left| \mathbb{P}[M=m] - \frac{1}{n} \sum_{i=1}^n \mathds{1}(M_i=m) \right|}_{\xrightarrow[\ n \to +\infty]{\Pjoint-\textnormal{a.s.}} 0} + 2 \sum_{m} \underbrace{\left\| \mathbb{P}[M=m] \Phi(\PmargmXMm) - \frac{1}{n} \sum_{i=1}^n \mathds{1}(M_i=m) \Phi(X_i^{(m)}) \right\|_{\mathcal{H}}}_{\xrightarrow[\ n \to +\infty]{\Pjoint-\textnormal{a.s.}} 0} ,
\end{align*}
where the two consistency properties in the last line are direct consequences of the strong law of large numbers. This concludes the proof.
\end{proof}

We also provide a uniform law of large numbers for second-order derivatives that will be useful to establish the asymptotic normality of the estimator.

\begin{lemma}\label{lemma_ULLN_2}
If Conditions \ref{cond_normality_4} and \ref{cond_normality_5} are fulfilled, then the Hessians of both the empirical loss $\widehat{L_n}$ and of $L$ exist on $K$, and each coefficient $\partial^2\widehat{L_n}/\partial\theta_k\partial\theta_j$ almost surely converges uniformly to $\partial^2L/\partial\theta_k\partial\theta_j$, i.e.\ for any pair $(j,k)$,
$$
\sup_{\theta\in K} \left| \frac{\partial^2\widehat{L_n}(\theta)}{\partial\theta_k\partial\theta_j} - \frac{\partial^2L(\theta)}{\partial\theta_k\partial\theta_j} \right| \xrightarrow[\ n \to +\infty]{\Pjoint-\textnormal{a.s.}} 0 .
$$
\end{lemma}

\begin{proof}
As previously, for any pair $(j,k)$:
\begin{align*}
    \sup_{\theta\in K} & \left| \frac{\partial^2\widehat{L_n}(\theta)}{\partial\theta_k\partial\theta_j} - \frac{\partial^2L(\theta)}{\partial\theta_k\partial\theta_j} \right| \\
    & = \sup_{\theta\in K} \bigg| \sum_{m\in\{0,1\}^d} \left( \mathbb{P}[M=m] - \frac{1}{n} \sum_{i=1}^n \mathds{1}(M_i=m) \right) 2 \left( \left\langle \frac{\partial^2\Phi(P_\theta^{(m)})}{\partial\theta_k\partial\theta_j} , \Phi(P_\theta^{(m)}) \right\rangle + \left\langle \frac{\partial\Phi(P_\theta^{(m)})}{\partial\theta_j} , \frac{\partial\Phi(P_\theta^{(m)})}{\partial\theta_k} \right\rangle \right) \\
    & \quad \quad \quad \quad  + 2 \sum_{m\in\{0,1\}^d} \left\langle \frac{\partial^2\Phi(P_\theta^{(m)})}{\partial\theta_k\partial\theta_j} , \mathbb{P}[M=m] \Phi(\PmargmXMm) - \frac{1}{n} \sum_{i=1}^n \mathds{1}(M_i=m) \Phi(X_i^{(m)}) \right\rangle_{\mathcal{H}} \bigg| \\
    & \leq 2 \sum_{m\in\{0,1\}^d} \underbrace{\left| \mathbb{P}[M=m] - \frac{1}{n} \sum_{i=1}^n \mathds{1}(M_i=m) \right|}_{\xrightarrow[\ n \to +\infty]{\Pjoint-\textnormal{a.s.}} 0}  \underbrace{\sup_{\theta\in K} \left\{ \left\| \frac{\partial^2\Phi(P_\theta^{(m)})}{\partial\theta_k\partial\theta_j} \right\| \cdot \left\| \Phi(P_\theta^{(m)}) \right\| + \left\| \frac{\partial\Phi(P_\theta^{(m)})}{\partial\theta_j} \right\| \cdot \left\| \frac{\partial\Phi(P_\theta^{(m)})}{\partial\theta_k} \right\| \right\}}_{<+\infty} \\
    & \quad \quad \quad \quad  + 2 \sum_{m\in\{0,1\}^d} \underbrace{\sup_{\theta\in K} \left\| \frac{\partial^2\Phi(P_\theta^{(m)})}{\partial\theta_k\partial\theta_j} \right\|_{\mathcal{H}}}_{<+\infty} \cdot \underbrace{\left\| \mathbb{P}[M=m] \Phi(\PmargmXMm) - \frac{1}{n} \sum_{i=1}^n \mathds{1}(M_i=m) \Phi(X_i^{(m)}) \right\|_{\mathcal{H}}}_{\xrightarrow[\ n \to +\infty]{\Pjoint-\textnormal{a.s.}} 0} 
\end{align*}
where we used the assumptions to ensure the existence of the involved quantities and the boundedness of the supremum in the last line. Note that the uniform boundedness of the first-order derivative follows from the uniform boundedness of the second-order derivative and the fact that $K$ is compact.
\end{proof}

\subsubsection{Consistency and asymptotic normality}

We can now prove the consistency of the MMD estimator which holds in one line:

\begin{proof}[Proof of Theorem \ref{thm_consistency}]
Invoking Theorem 2.1 in~\cite{NeweyMcFadden}, Assumption \ref{cond_consistency1} and \ref{cond_consistency2} along with Lemma \ref{lemma_ULLN} imply the strong consistency of the minimizer $\MMDn$ of $\widehat{L_n}$ towards the unique minimizer of $L$.
\end{proof}

The proof of asymptotic normality is a little longer.

\begin{proof}[Proof of Theorem \ref{thm_normality}]
According to Condition~\ref{cond_normality_4}, $\widehat{L_n}$ is twice differentiable on a neighborhood of $\MMDinf$.
Moreover, due to the consistency of $\MMDn$ (\ref{cond_consistency1} + \ref{cond_consistency2} satisfied), we can assume that $\MMDn$ belongs to such a neighborhood for $n$ large enough. As \ref{cond_normality_3} holds, the first-order condition is
\begin{equation}
\label{second-order-condition}
    0 =  \nabla_\theta \widehat{L_n}(\MMDn) = \nabla_\theta \widehat{L_n}(\MMDinf) + \nabla^2_{\theta,\theta} \widehat{L_n}(\bar\theta_n) (\MMDn-\MMDinf),
\end{equation}
where $\bar\theta_{n}$ is a random vector whose components lie between those of $\MMDinf$ and $\MMDn$.
Let us now study the asymptotic behavior of the Hessian matrix $H_n=\nabla^2_{\theta,\theta} \widehat{L_n}(\bar\theta_n)$ and of $\nabla_\theta \widehat{L_n}(\MMDinf)$.


Consistency of $H_n$: Notice first that as $\bar\theta_{n}$ lies between $\MMDn$ and $\MMDinf$ componentwise, then $\bar\theta_{n} \rightarrow  \MMDinf$ ($\Pjoint$-a.s.) as $n\rightarrow+\infty$ and we have existing second-order derivatives of $\widehat{L}_n$ at $\bar\theta_n\in K$ for $n$ large enough (which ensures the existence of $H_n$). For any pair $(j,k)$,
\begin{align*}
\bigg| \frac{\partial^2\widehat{L_n}(\bar\theta_n)}{\partial\theta_k\partial\theta_j} - \frac{\partial^2L(\MMDinf)}{\partial\theta_k\partial\theta_j} \bigg| & \leq \left| \frac{\partial^2\widehat{L_n}(\bar\theta_n)}{\partial\theta_k\partial\theta_j} - \frac{\partial^2L(\bar\theta_n)}{\partial\theta_k\partial\theta_j} \right| + \bigg| \frac{\partial^2L(\bar\theta_n)}{\partial\theta_k\partial\theta_j}  - \frac{\partial^2L(\MMDinf)}{\partial\theta_k\partial\theta_j} \bigg| \\
& \leq \, \sup_{\theta\in K} \left| \frac{\partial^2\widehat{L_n}(\theta)}{\partial\theta_k\partial\theta_j} - \frac{\partial^2L(\theta)}{\partial\theta_k\partial\theta_j} \right| + \bigg| \frac{\partial^2L(\bar\theta_n)}{\partial\theta_k\partial\theta_j}  - \frac{\partial^2L(\MMDinf)}{\partial\theta_k\partial\theta_j} \bigg| .
\end{align*}
According to Lemma \ref{lemma_ULLN_2} (\ref{cond_normality_4} and \ref{cond_normality_5} are satisfied), the first term in the bound is going to zero almost surely, as well as the second one due to the continuity of ${\partial^2L(\cdot;)}/{\partial\theta_k\partial\theta_j}$ at $\MMDinf$ \eqref{cond_normality_4}. Notice then that the definition of $H$ in Assumption \ref{cond_normality_6} matches $\nabla_{\theta,\theta}^2L(\MMDinf)$, 
and we obtain the almost sure convergence of the matrix $H_n \rightarrow H$ ($\Pjoint$-a.s.) as $n\rightarrow+\infty$.

Asymptotic normality of $\nabla_\theta \widehat{L_n}(\MMDinf)$: We start by noticing that for any $\theta$, $j$,
$$
\frac{\partial\widehat{L_n}(\theta)}{\partial\theta_j} - \frac{\partial L(\theta)}{\partial\theta_j} = \frac{1}{n} \sum_{i=1}^n \frac{\partial\ell(X_i,M_i;\theta)}{\partial\theta_j}  - \frac{\partial\mathbb{E}_{(X,M)}\left[\ell(X,M;\theta)\right]}{\partial\theta_j} .
$$
and when evaluated in particular at the minimizer $\MMDinf\in\text{int}(\Theta)$ \eqref{cond_normality_3} of $L$:
$$
\frac{\partial\widehat{L_n}(\MMDinf)}{\partial\theta_j} = \frac{1}{n} \sum_{i=1}^n \frac{\partial\ell(X_i,M_i;\MMDinf)}{\partial\theta_j} - \mathbb{E}_{(X,M)}\left[\frac{\partial\ell(X,M;\MMDinf)}{\partial\theta_j}\right] .
$$
According to \eqref{cond_normality_7}, the variance of the gradient exists and is equal to
\begin{align*}
    &\Var_{(X,M)\sim\Pjoint}\left[\nabla_\theta\ell(X,M;\MMDinf)\right] \\
    & = \Var_{(X,M)}\left[\nabla_{\theta} \MMDtwo{P_{\MMDinf}^{(M)}}{\delta_{\{X^{(M)}\}}}\right] \\
    & = \mathbb{E}_{M\sim\PM}\left[\Var_{X\sim\PXM}\left[\nabla_{\theta} \MMDtwo{P_{\MMDinf}^{(M)}}{\delta_{\{X^{(M)}\}}}\right]\right] + \Var_{M\sim\PM}\left[\mathbb{E}_{X\sim\PXM}\left[\nabla_{\theta} \MMDtwo{P_{\MMDinf}^{(M)}}{\delta_{\{X^{(M)}\}}}\right]\right] \\
    & = \mathbb{E}_{M\sim\PM}\left[\Var_{X\sim\PXM}\left[\nabla_{\theta} \MMDtwo{P_{\MMDinf}^{(M)}}{\delta_{\{X^{(M)}\}}}\right]\right] + \Var_{M\sim\PM}\left[\nabla_{\theta} \MMDtwo{P_{\MMDinf}^{(M)}}{\PmargMXM}\right] \\
    & = \Sigma_1 + \Sigma_2 ,
\end{align*}
which then leads using the Central Limit Theorem to
$$
\sqrt{n} \nabla_\theta \widehat{L}_n(\MMDinf) \xrightarrow[n\rightarrow+\infty]{\mathcal{L}}  \mathcal{N}\left( 0 , \Sigma_1 + \Sigma_2 \right) .
$$

Finally, as $H$ is nonsingular, the matrix $H_n$ is a.s. invertible for a sufficiently large $n$, and using Slutsky’s lemma in Equation \ref{second-order-condition}, we get
$$
\sqrt{n} \left(\MMDn-\MMDinf\right) =  - H_n^{-1}\nabla_\theta \widehat{L_n}(\MMDn) \xrightarrow[n\rightarrow+\infty]{\mathcal{L}} \mathcal{N}\left(0,H^{-1}\Sigma_1 H^{-1} + H^{-1}\Sigma_2 H^{-1}\right) .
$$

\end{proof}

\subsection{Proofs from Section \ref{Sec_robust}}

\begin{proof}[Proof of Theorem \ref{MLE_Robust}]
We recall that $d=1$, that the distribution of interest is the actual data distribution $\PX=\Pthetastar$, and that we denote $\pi(X)=\mathbb{P}[M=0|X]$. We then have,
\begin{align*}
    \textnormal{KL}\left(\mathbb{P}_{X|M=0}\|\PX\right) & =  \textnormal{KL}\left(\mathbb{P}_{X|M=0}\|\PX\right) \\
    & = \E_{X\sim \mathbb{P}_{X|M=0}}\left[ \log\left(\frac{\mathrm{d}\mathbb{P}_{X|M=0}}{\mathrm{d}\PX}(X)\right) \right] \\
    & = \E_{X\sim \PX}\left[  \log\left( \frac{\pi(X)}{\pi} \cdot \frac{\mathrm{d}\PX}{\mathrm{d}\PX}(X) \right) \frac{\pi(X)}{\pi} \right] \\
    & = \E_{X\sim \PX}\left[  \log\left( \frac{\pi(X)}{\pi} \right) \frac{\pi(X)}{\pi} \right] \\
    & = \E_{X\sim \PX}\left[  \log\left( 1 + \frac{\pi(X)-\pi}{\pi} \right) \frac{\pi(X)}{\pi} \right] \\
    & \leq \E_{X\sim \PX}\left[  \left( \frac{\pi(X)-\pi}{\pi} \right) \frac{\pi(X)}{\pi} \right] \\
    & = \E_{X\sim \PX}\left[  \left( \frac{\pi(X)}{\pi} \right)^2 \right] - 1 \\
    & = \E_{X\sim \PX}\left[  \frac{\pi(X)^2-\pi^2}{\pi^2} \right] \\
    & = \frac{\Var_{X\sim \PX}\left[\pi(X)\right]}{\pi^2} .
\end{align*}
Then using respectively the triangle inequality, Pinsker's inequality and the definition of $\MLEinf$:
\begin{align*}
    \textnormal{TV}\left(P_{\theta^*},P_{\MLEinf}\right) & \leq \textnormal{TV}\left(\mathbb{P}_{X|M=0},P_{\theta^*}\right) + \textnormal{TV}\left(\mathbb{P}_{X|M=0},P_{\MLEinf}\right) \\
    & \leq \sqrt{\frac{1}{2}\textnormal{KL}\left(\mathbb{P}_{X|M=0}\|P_{\theta^*}\right)} + \sqrt{\frac{1}{2}\textnormal{KL}\left(\mathbb{P}_{X|M=0}\|P_{\MLEinf}\right)} \\
    & \leq \sqrt{2\cdot\textnormal{KL}\left(\mathbb{P}_{X|M=0}\|\Pthetastar\right)} \\
    & = \sqrt{2\cdot\textnormal{KL}\left(\mathbb{P}_{X|M=0}\|\PX\right)} \\
    & = \sqrt{2\cdot\frac{\Var_{X\sim \PX}\left[\pi(X)\right]}{\pi^2}} ,
\end{align*}
which ends the proof.
\end{proof}

\begin{proof}[Proof of Theorem \ref{MMD_Robust}]
Denoting $\pi_m(X)=\PMmX$ the missingness mechanism, 
we start from the definition of the MMD distance: for any pattern $m$,
\begin{align*}
    \mathbb{D}\left(\PmargmX,\PmargmXMm\right)
    & = \left\lVert \E_{X^{(m)}\sim\PmargmX} \left[ k(X^{(m)},\cdot) \right] - \E_{X^{(m)}\sim\PmargmXMm} \left[ k(X^{(m)},\cdot) \right] \right\rVert_{\mathcal{H}} \\
    & = \left\lVert \E_{X\sim\PX} \left[ k(X^{(m)},\cdot) \right] - \E_{X\sim\PXMm} \left[ k(X^{(m)},\cdot) \right] \right\rVert_{\mathcal{H}} \\
    & = \left\lVert \E_{X\sim\PX} \left[ k(X^{(m)},\cdot) \right] - \E_{X\sim \PX} \left[ \frac{\pi_m(X)}{\pi_m} k(X^{(m)},\cdot) \right] \right\rVert_{\mathcal{H}} \\
    & = \left\lVert \E_{X\sim \PX} \left[ \frac{\pi_m(X)-\pi_m}{\pi_m} \cdot k(X^{(m)},\cdot) \right] \right\rVert_{\mathcal{H}} \\
    & \leq \E_{X\sim \PX} \left[ \frac{\left|\pi_m(X)-\pi_m\right|}{\pi_m} \cdot \left\lVert k(X^{(m)},\cdot) \right\rVert_{\mathcal{H}} \right]  \\
    & \leq \frac{\E_{X\sim \PX} \left[ \left|\pi_m(X)-\pi_m\right| \right]}{\pi_m} \\
    & \leq \frac{\sqrt{\Var_{X\sim \PX} \left[ \pi_m(X) \right]}}{\pi_m} .
\end{align*}
In the case where $\PX=\Pthetastar$, we can conclude using the triangle inequality and the definition of $\MMDinf$:
\begin{align*}
    \E_{M \sim \PM}\left[\mathbb{D}^2\left(P_{\MMDinf}^{(M)},P_{\theta^*}^{(M)}\right)\right] & \leq 2 \E_{M \sim \PM}\left[\mathbb{D}^2\left(P_{\MMDinf}^{(M)},\PmargMXM\right)\right] + 2 \E_{M \sim \PM}\left[\mathbb{D}^2\left(\PmargMXM,\PmargMthetastar\right)\right] \\
    & \leq 2 \E_{M \sim \PM}\left[\mathbb{D}^2\left(P_{\theta^*}^{(M)},\PmargMXM\right)\right] + 2 \E_{M \sim \PM}\left[\mathbb{D}^2\left(\PmargMXM,\PmargMthetastar\right)\right] \\
    & = 4 \E_{M \sim \PM}\left[\mathbb{D}^2\left(\PmargMX,\PmargMXM\right)\right] \\
    & \leq 4 \cdot \E_{M \sim \PM}\left[ \frac{\Var_{X\sim \PX} \left[ \pi_M(X) \right]}{\pi_M^2} \right] .
\end{align*}
In the misspecified setting, we have for any value of $\theta$:
\begin{align*}
    \E_{M \sim \PM}\left[\mathbb{D}^2\left(P_{\MMDinf}^{(M)},\PmargMX\right)\right] & \leq 2 \E_{M \sim \PM}\left[\mathbb{D}^2\left(P_{\MMDinf}^{(M)},\PmargMXM\right)\right] + 2 \E_{M \sim \PM}\left[\mathbb{D}^2\left(\PmargMXM,\PmargMX\right)\right] \\
    & \leq 2 \E_{M \sim \PM}\left[\mathbb{D}^2\left(P_{\theta}^{(M)},\PmargMXM\right)\right] + 2 \E_{M \sim \PM}\left[\mathbb{D}^2\left(\PmargMXM,\PmargMX\right)\right] \\
    & \leq 2 \E_{M \sim \PM}\left[\mathbb{D}^2\left(P_{\theta}^{(M)},\PmargMX\right)\right] + 4 \E_{M \sim \PM}\left[\mathbb{D}^2\left(\PmargMXM,\PmargMX\right)\right] \\
    & \leq 2 \cdot \E_{M \sim \PM}\left[\mathbb{D}^2\left(\PmargMtheta,\PmargMX\right)\right] \\
    & \quad \quad + 4 \cdot \E_{M \sim \PM}\left[ \frac{\Var_{X\sim \PX} \left[ \pi_M(X) \right]}{\pi_M^2} \right] .
\end{align*}
and taking the infimum in the r.h.s.\ over $\theta\in\Theta$ ends the proof. 
\end{proof}

\begin{proof}[Proof of Theorem \ref{thm_finite_main}] First observe that by definition:
$$
\MMDn = \argmin_{\theta\in\Theta} \mathbb{D}\left(P_n,P_{\theta}\right) \quad \textnormal{where} \quad P_n = \frac{1}{|\{i:M_i=0\}|} \sum_{i:M_i=0} \delta_{\{X_i\}} .
$$
We implicitly assume from now on that the quantity $|\{i:M_i=0\}|$ in the denominator is actually equal to $\max(1,|\{i:M_i=0\}|)$ so that $P_n$ is well-defined.
We then have for any $\theta\in\Theta$:
\begin{align*}
    \mathbb{D}\left(P_{\MMDn},\PX\right) & \leq \mathbb{D}\left(P_{\MMDn},P_n\right) + \mathbb{D}\left(\PX,P_n\right) \\
    & \leq \mathbb{D}\left(P_\theta,P_n\right) +  \mathbb{D}\left(\PX,P_n\right) \\
    & \leq \mathbb{D}\left(P_\theta,\PX\right) + 2\mathbb{D}\left(\PX,P_n\right) \\
    & \leq \mathbb{D}\left(P_\theta,\PX\right) + 2\mathbb{D}\left(\PX,\mathbb{P}_{X|M=0}\right) + 2\mathbb{D}\left(\mathbb{P}_{X|M=0},P_n\right) \\
    & \leq \mathbb{D}\left(P_\theta,\PX\right) + \frac{2\sqrt{\Var_{X\sim \PX}\left[\pi(X)\right]}}{\pi} + 2\mathbb{D}\left(\mathbb{P}_{X|M=0},P_n\right) ,
\end{align*}
where the bound on the second term directly follows from the proof of Theorem \ref{MMD_Robust}. The third term can be controlled in expectation as follows. We first write:
\begin{align*}
    \mathbb{D}^2\left(\mathbb{P}_{X|M=0},P_n\right) & = \left\lVert \E_{X\sim\mathbb{P}_{X|M=0}}\left[\Phi(X)\right] - \frac{1}{|\{i:M_i=0\}|} \sum_{i:M_i=0} \Phi(X_i) \right\rVert_{\mathcal{H}}^2 \\
    & = \left\lVert \frac{1}{|\{i:M_i=0\}|} \sum_{i:M_i=0} \left\{ \Phi(X_i) - \E_{X\sim\mathbb{P}_{X|M=0}}\left[\Phi(X)\right] \right\} \right\rVert_{\mathcal{H}}^2 \\
    & = \frac{1}{|\{i:M_i=0\}|^2} \sum_{i:M_i=0} \left\lVert \Phi(X_i) - \E_{X\sim\mathbb{P}_{X|M=0}}\left[\Phi(X)\right] \right\rVert_{\mathcal{H}}^2 \\
    & + \frac{2}{|\{i:M_i=0\}|^2} \sum_{i<j:M_i=M_j=0} \left\langle \Phi(X_i) - \E_{X\sim\mathbb{P}_{X|M=0}}\left[\Phi(X)\right],\Phi(X_j) - \E_{X\sim\mathbb{P}_{X|M=0}}\left[\Phi(X)\right] \right\rangle_{\mathcal{H}} .
\end{align*}
The key point now lies in the fact that conditional on $\{M_i\}_{1\leq i \leq n}$, all the $X_i$'s such that $M_i=0$ are i.i.d.\ with common distribution $\mathbb{P}_{X|M=0}$. Then, in expectation over the sample $\mathcal{S}=\{(X_i,M_i)\}_{1\leq i \leq n}$, the first term in the last line becomes:
\begin{align*}
    \E_{\mathcal{S}} \bigg[ & \frac{1}{|\{i:M_i=0\}|^2} \sum_{i:M_i=0} \left\lVert \Phi(X_i) - \E_{X\sim\mathbb{P}_{X|M=0}}\left[\Phi(X)\right] \right\rVert_{\mathcal{H}}^2\bigg] \\
    & = \E_{\{M_i\}_{1\leq i \leq n}\sim\PM} \left[ \frac{1}{|\{i:M_i=0\}|^2} \sum_{i:M_i=0} \E_{X_i\sim\mathbb{P}_{X|M=0}}\left[ \left\lVert \Phi(X_i) - \E_{X\sim\mathbb{P}_{X|M=0}}\left[\Phi(X)\right] \right\rVert_{\mathcal{H}}^2\right] \right] \\
    & = \E_{\{M_i\}_{1\leq i \leq n}\sim\PM} \left[ \frac{1}{|\{i:M_i=0\}|} \cdot \E_{X\sim\mathbb{P}_{X|M=0}}\left[ \left\lVert \Phi(X) - \E_{X\sim\mathbb{P}_{X|M=0}}\left[\Phi(X)\right] \right\rVert_{\mathcal{H}}^2\right] \right] \\
    & \leq \E_{\{M_i\}_{1\leq i \leq n}\sim\PM} \left[ \frac{1}{|\{i:M_i=0\}|} \cdot \E_{X\sim\mathbb{P}_{X|M=0}}\left[ \left\lVert \Phi(X) \right\rVert_{\mathcal{H}}^2\right] \right] \\
    & \leq \E_{\{M_i\}_{1\leq i \leq n}\sim\PM} \left[ \frac{1}{|\{i:M_i=0\}|} \right] \\
    & \leq \frac{2}{n\cdot\mathbb{P}[M=0]} ,
\end{align*}
where the first inequality comes from the standard inequality $\textnormal{Variance}\leq\textnormal{Second order moment}$, the second one from the boundedness of the kernel, and the last one as a standard bound on the negative moments of a binomial random variable \citep{ChaoStrawderman1972} ($\E[1/(1+\textnormal{Bin(n,p)})]\leq1/np$). Note that we also used the inequality $1/x\leq2/(1+x)$.

\vspace{0.2cm}
Similarly, we have:
\begin{align*}
    \E_{\mathcal{S}} \bigg[ & \frac{2}{|\{i:M_i=0\}|^2} \sum_{i<j:M_i=M_j=0} \left\langle \Phi(X_i) - \E_{X\sim\mathbb{P}_{X|M=0}}\left[\Phi(X)\right],\Phi(X_j) - \E_{X\sim\mathbb{P}_{X|M=0}}\left[\Phi(X)\right] \right\rangle_{\mathcal{H}} \bigg] \\
    & = \E_{\{M_i\}_{1\leq i \leq n}\sim\PM} \bigg[ \frac{2}{|\{i:M_i=0\}|^2} \sum_{i<j:M_i=M_j=0} \\
    & \quad \quad \left\langle \E_{X_i\sim\mathbb{P}_{X|M=0}}\left[\Phi(X_i) - \E_{X\sim\mathbb{P}_{X|M=0}}\left[\Phi(X)\right]\right], \E_{X_j\sim\mathbb{P}_{X|M=0}}\left[\Phi(X_j) - \E_{X\sim\mathbb{P}_{X|M=0}}\left[\Phi(X)\right]\right] \right\rangle_{\mathcal{H}} \bigg] \\
    & = \E_{\{M_i\}_{1\leq i \leq n}\sim\PM} \bigg[ \frac{2}{|\{i:M_i=0\}|^2} \sum_{i<j:M_i=M_j=0} \left\langle 0, 0 \right\rangle_{\mathcal{H}} \bigg] \\
    & = 0 .
\end{align*}

This finally provides
\begin{align*}
    \E_{\mathcal{S}}\left[ \mathbb{D}\left(P_{\MMDn},\PX\right) \right] & \leq \inf_{\theta \in \Theta }\mathbb{D}\left(\Ptheta,\PX\right) + \frac{2\sqrt{\Var_{X\sim \PX}\left[\pi(X)\right]}}{\pi} + 2\E_{\mathcal{S}}\left[\mathbb{D}\left(\mathbb{P}_{X|M=0},P_n\right)\right] \\
    & \leq \inf_{\theta \in \Theta } \mathbb{D}\left(\Ptheta,\PX\right) + \frac{2\sqrt{\Var_{X\sim \PX}\left[\pi(X)\right]}}{\pi} + 2\sqrt{\E_{\mathcal{S}}\left[\mathbb{D}^2\left(\mathbb{P}_{X|M=0},P_n\right)\right]} \\
    & \leq \inf_{\theta \in \Theta }\mathbb{D}\left(\Ptheta,\PX\right) + \frac{2\sqrt{\Var_{X\sim \PX}\left[\pi(X)\right]}}{\pi} + 2\sqrt{\frac{2}{n\cdot\mathbb{P}[M=0]}} .
\end{align*}
\end{proof}

\begin{proof}[Proof of Theorem \ref{thm_finite_algo}]

$\Theta$ is bounded with diameter $D$ so (2.17) in \cite{nemirovski2009robust} holds, and the bounded expected gradient assumption guarantees (2.5) in \cite{nemirovski2009robust}. This ensures (2.21), which can be read in our setting as
\[
\mathbb{E}\left[\operatorname{Crit}_n\left(\widehat{\theta}_n^{(T)}\right)-\operatorname{Crit}_n\left(\theta_n^{\mathrm{MMD}}\right)\right]\leq\frac{D C}{\sqrt{T}} .
\]
for $\eta=D/(C\sqrt{T})$.
Since $\operatorname{Crit}_n(\theta)$ is equal to $\mathbb{D}^2(P_\theta,P_n)$ up to a constant where 
$$
P_n = \frac{1}{|\{i:M_i=0\}|} \sum_{i:M_i=0} \delta_{\{X_i\}} \, , 
$$
Jensen's inequality and $\sqrt{a+b}\leq\sqrt{a}+\sqrt{b}$ imply
\[
\mathbb{E}\left[\mathbb{D}\left(P_{\widehat{\theta}_n^{(T)}},P_n\right)\right]
\leq
\sqrt{\mathbb{E}\left[\mathbb{D}^2\left(P_{\theta_n^{\mathrm{MMD}}},P_n\right)\right]} + \sqrt{\frac{D C}{\sqrt{T}}} 
\leq
\sqrt{\mathbb{E}\left[\mathbb{D}^2\left(P_{\theta^*},P_n\right)\right]} + \sqrt{\frac{D C}{\sqrt{T}}} \, ,
\]
where the second inequality follows from the definition of
$\MMDn$ as a minimum MMD estimator. Hence, by the triangle
inequality and Jensen's inequality,
\begin{align*}
\mathbb{E}\left[\mathbb{D}\left(P_{\widehat{\theta}_n^{(T)}},\PX\right)\right]
& \leq
\mathbb{E}\left[\mathbb{D}\left(P_{\widehat{\theta}_n^{(T)}},P_n\right)\right] + \mathbb{E}\left[\mathbb{D}\left(P_n,\PX\right)\right] \\
& \leq
\sqrt{\mathbb{E}\left[\mathbb{D}^2\left(P_{\theta^*},P_n\right)\right]}+\sqrt{\frac{DC}{\sqrt{T}}} + \mathbb{E}\left[\mathbb{D}\left(P_n,\PX\right)\right] \\
& =
\sqrt{\mathbb{E}\left[\mathbb{D}^2\left(\PX,P_n\right)\right]}+\sqrt{\frac{DC}{\sqrt{T}}} + \mathbb{E}\left[\mathbb{D}\left(P_n,\PX\right)\right] \\
& \leq
\sqrt{\frac{DC}{\sqrt{T}}} + 2\,\sqrt{\mathbb{E}\left[\mathbb{D}^2\left(P_n,\PX\right)\right]} \, .
\end{align*}
Finally, the argument used in the proof of Theorem~\ref{MMD_Robust} gives
\[
\sqrt{\mathbb{E}\left[\mathbb{D}^2(P_n,\PX)\right]}
\leq
\frac{ \sqrt{\Var_{X\sim P_X}[\pi(X)]} }{\pi} + \sqrt{\frac{2}{n\pi}} .
\]
Combining the last two displays yields
\[
\begin{aligned}
\mathbb{E}\left[\mathbb{D}\left(P_{\widehat{\theta}_n^{(T)}},\PX\right)\right]
&\leq
\frac{ 2\sqrt{\Var_{X\sim P_X}[\pi(X)]} }{\pi} + \frac{2\sqrt{2}}{\sqrt{n\pi}} + \sqrt{\frac{DC}{\sqrt{T}}} ,
\end{aligned}
\]
which proves the result.
\end{proof}

\begin{proof}[Proof of Theorem \ref{MMD_Robust_cont}]
Theorem \ref{MMD_Robust} (using the squared absolute mean deviation $\mathbb{M}$ instead of the variance $\Var$) leads to:
\begin{align*}
    &\E_{M \sim \PM}\left[\mathbb{D}^2\left(P_{\MMDinf}^{(M)},\PmargMthetastar\right)\right] \leq 2 \E_{M \sim \PM}\left[\mathbb{D}^2\left(P_{\MMDinf}^{(M)},\PmargMX\right)\right] + 2 \E_{M \sim \PM}\left[\mathbb{D}^2\left(\PmargMX,\PmargMthetastar\right)\right] \\
    & \leq 2 \left(2 \cdot \inf_{\theta\in\Theta}\E_{M \sim \PM}\left[\mathbb{D}^2\left(P_{\theta}^{(M)},\PmargMX\right)\right] + 4 \cdot \E_{M \sim \PM}\left[ \frac{\mathbb{M}_{X\sim \PX} \left[ \pi_M(X) \right]}{\mathbb{E}_{X\sim \PX} \left[ \pi_M(X) \right]^2} \right] \right) + 2 \E_{M \sim \PM}\left[\mathbb{D}^2\left(\PmargMX,\PmargMthetastar\right)\right] \\
    & = 4 \cdot \inf_{\theta\in\Theta}\E_{M \sim \PM}\left[\mathbb{D}^2\left(P_{\theta}^{(M)},\PmargMX\right)\right] + 8 \cdot \E_{M \sim \PM}\left[ \frac{\mathbb{M}_{X\sim \PX} \left[ \pi_M(X) \right]}{\mathbb{E}_{X\sim \PX} \left[ \pi_M(X) \right]^2} \right] + 2 \E_{M \sim \PM}\left[\mathbb{D}^2\left(\PmargMX,\PmargMthetastar\right)\right] \\
    & \leq 6 \cdot \E_{M \sim \PM}\left[\mathbb{D}^2\left(\PmargMthetastar,\PmargMX\right)\right] + 8 \cdot \E_{M \sim \PM}\left[ \frac{\mathbb{M}_{X\sim \PX} \left[ \pi_M(X) \right]}{\mathbb{E}_{X\sim \PX} \left[ \pi_M(X) \right]^2} \right] .
\end{align*}
The first term is bounded as follows:
\begin{align*}
    \E&_{M\sim \PM}\left[\mathbb{D}^2\left(\PmargMthetastar,\PmargMX\right)\right] = \E_{M\sim \PM}\left[\mathbb{D}^2\left(\PmargMthetastar,(1-\dataeps)\PmargMthetastar+\dataeps\mathbb{Q}_{X}^{(M)}\right)\right] \\
    & = \E_{M\sim \PM}\left[\bigg\lVert \E_{X\sim P_{\theta^*}} \left[ k(X^{(M)},\cdot) \right] - \bigg\{ (1-\dataeps) \E_{X\sim P_{\theta^*}} \left[ k(X^{(M)},\cdot) \right] + \dataeps \E_{X\sim \QX} \left[ k(X^{(M)},\cdot) \right] \bigg\} \bigg\rVert_{\mathcal{H}}^2\right] \\
    & = \E_{M\sim \PM}\left[\bigg\lVert \dataeps \bigg\{ \E_{X\sim P_{\theta^*}} \left[ k(X^{(M)},\cdot) \right] - \E_{X\sim \QX} \left[ k(X^{(M)},\cdot) \right] \bigg\} \bigg\rVert_{\mathcal{H}}^2\right]\\
    & \leq \dataeps^2 \cdot \E_{M\sim \PM}\left[\left(\big\lVert \E_{X\sim P_{\theta^*}} \left[ k(X^{(M)},\cdot) \right] \big\rVert_{\mathcal{H}} + \big\lVert \E_{X\sim \QX} \left[ k(X^{(M)},\cdot) \right] \big\rVert_{\mathcal{H}}\right)^2\right] \\
    & \leq 4 \dataeps^2 ,
\end{align*}
while the other variance term is bounded as in the previous proof above: with the notation $\pi_m^*=\E_{X\sim P_{\theta^*}}\left[\pi_m(X)\right]$, we have for any pattern $m$,
$$
\frac{\mathbb{M}_{X\sim \PX} \left[ \pi_M(X) \right]}{\mathbb{E}_{X\sim \PX} \left[ \pi_M(X) \right]^2} = \frac{\mathbb{M}_{X\sim(1-\dataeps) P_{\theta^*} + \dataeps \mathbb{Q}_X} \left[ \pi_M(X) \right]}{\mathbb{E}_{X\sim \PX} \left[ \pi_M(X) \right]^2}  \leq \frac{2(1-\dataeps)^2\mathbb{M}_{X\sim \Pthetastar} \left[ \pi_M(X) \right]+8\dataeps^2}{\mathbb{E}_{X\sim \PX} \left[ \pi_M(X) \right]^2} ,
$$
since for any measurable function $f\in(0,1)$, $\mathbb{E}_{X\sim(1-\dataeps) P_{\theta^*} + \dataeps \mathbb{Q}_X}[f(X)] = (1-\dataeps) \mathbb{E}_{X\sim P_{\theta^*} }[f(X)]  + \dataeps \mathbb{E}_{X\sim \mathbb{Q}_X}[f(X)]$ and 
\begin{align*}
\mathbb{M}^{1/2}_{X\sim(1-\dataeps) P_{\theta^*} + \dataeps \mathbb{Q}_X}\left[f(X)\right] & = \mathbb{E}_{X\sim(1-\dataeps) P_{\theta^*} + \dataeps \mathbb{Q}_X}\left[\big|f(X)-\{(1-\dataeps) \mathbb{E}_{X\sim P_{\theta^*} }[f(X)]  + \dataeps \mathbb{E}_{X\sim \mathbb{Q}_X}[f(X)]\}\big|\right] \\
& = \mathbb{E}_{X\sim(1-\dataeps) P_{\theta^*} + \dataeps \mathbb{Q}_X}\left[\big|\left\{ f(X)-\mathbb{E}_{X\sim P_{\theta^*} }[f(X)]  \right\} + \dataeps \left\{ \mathbb{E}_{X\sim P_{\theta^*}}[f(X)] - \mathbb{E}_{X\sim \mathbb{Q}_X}[f(X)] \right\} \big|\right] \\
& \leq \mathbb{E}_{X\sim(1-\dataeps) P_{\theta^*} + \dataeps \mathbb{Q}_X}\left[\big|f(X)- \mathbb{E}_{X\sim P_{\theta^*} }[ f(X)]\big|\right] + \dataeps \big| \mathbb{E}_{X\sim P_{\theta^*}}[f(X)] - \mathbb{E}_{X\sim \mathbb{Q}_X}[f(X)] \big| \\
& = (1-\dataeps) \mathbb{E}_{X\sim P_{\theta^*}}\left[\big|f(X)- \mathbb{E}_{X\sim P_{\theta^*} }[ f(X)]\big|\right] 
+ \dataeps \mathbb{E}_{X\sim\mathbb{Q}_X}\left[\big|f(X)- \mathbb{E}_{X\sim P_{\theta^*} }[f(X)]\big|\right] \\
& \quad \quad \quad \quad \quad \quad \quad \quad \quad \quad \quad \quad \quad \quad \quad \quad \quad \quad \quad + \dataeps \big| \mathbb{E}_{X\sim P_{\theta^*}}[f(X)] - \mathbb{E}_{X\sim \mathbb{Q}_X}[f(X)] \big| \\
& \leq (1-\dataeps) \cdot \mathbb{M}^{1/2}_{X\sim P_{\theta^*}}\left[f(X)\right] + 4\dataeps .
\end{align*}
Thus we get:
\begin{align*}
    \E&_{M \sim \PM}\left[\mathbb{D}^2\left(P_{\MMDinf}^{(M)},\PmargMthetastar\right)\right] \leq 6 \cdot \E_{M\sim\PM}\left[\mathbb{D}^2\left(\PmargMthetastar,\PmargMX\right)\right] + 8 \cdot \E_{M\sim\PM}\left[ \frac{\mathbb{M}_{X\sim \PX} \left[ \pi_M(X) \right]}{\mathbb{E}_{X\sim \PX} \left[ \pi_M(X) \right]^2} \right] \\
    & \leq 24\dataeps^2 + 16(1-\dataeps)^2 \cdot \mathbb{E}_{M \sim \PM}\left[\frac{\mathbb{M}_{X\sim \Pthetastar} \left[ \pi_M(X) \right]}{\mathbb{E}_{X\sim\PX} \left[ \pi_M(X) \right]^2}\right] + 64\dataeps^2 \cdot \mathbb{E}_{M \sim \PM}\left[\frac{1}{\mathbb{E}_{X\sim\PX} \left[ \pi_M(X) \right]^2}\right] \\
    & = 24\dataeps^2 + 16(1-\dataeps)^2 \cdot  \sum_m \E_{X\sim\PX}\left[\pi_m(X)\right] \frac{\mathbb{M}_{X\sim \Pthetastar} \left[ \pi_m(X) \right]}{\mathbb{E}_{X\sim\PX} \left[ \pi_m(X) \right]^2} + 64\dataeps^2 \sum_m \E_{X\sim\PX}\left[\pi_m(X)\right] \frac{1}{\mathbb{E}_{X\sim\PX} \left[ \pi_m(X) \right]^2} \\
    & = 24\dataeps^2 + 16(1-\dataeps)^2 \cdot  \sum_m \frac{\mathbb{M}_{X\sim \Pthetastar} \left[ \pi_m(X) \right]}{\mathbb{E}_{X\sim\PX} \left[ \pi_m(X) \right]} + 64\dataeps^2 \sum_m \frac{1}{\mathbb{E}_{X\sim\PX} \left[ \pi_m(X) \right]} \\
    & \leq 24\dataeps^2 + 16(1-\dataeps)^2 \cdot \sum_m \frac{\mathbb{M}_{X\sim \Pthetastar} \left[ \pi_m(X) \right]}{(1-\dataeps)\mathbb{E}_{X\sim\Pthetastar} \left[ \pi_m(X) \right]} + 64\dataeps^2 \sum_m \frac{1}{(1-\dataeps)\mathbb{E}_{X\sim\Pthetastar} \left[ \pi_m(X) \right]} \\
    & = 24\dataeps^2 + 16(1-\dataeps) \cdot \mathbb{E}_{M\sim P^*_M}\left[\frac{\mathbb{M}_{X\sim \Pthetastar} \left[ \pi_M(X) \right]}{\pi_M^{* 2}}\right] + \frac{64\dataeps^2}{1-\dataeps} \cdot \mathbb{E}_{M\sim P^*_M}\left[ \frac{1}{\pi_M^{* 2}}\right] 
\end{align*}
since $\mathbb{E}_{X\sim\PX} \left[ \pi_m(X) \right]=(1-\dataeps)\mathbb{E}_{X\sim \Pthetastar} \left[ \pi_m(X) \right]+\dataeps\mathbb{E}_{X\sim \QX} \left[ \pi_m(X) \right]\geq(1-\dataeps)\mathbb{E}_{X\sim \Pthetastar} \left[ \pi_m(X) \right]$,
and finally, using the identity:
\begin{align*}
\E_{M\sim\PM}\left[f(M)\right] & = \sum_m \E_{X\sim\PX}\left[\pi_m(X)\right] f(m) \\
& = \sum_m \E_{X\sim(1-\dataeps) P_{\theta^*} + \dataeps \mathbb{Q}_X } \left[\pi_m(X)\right] f(m) \\
& = (1-\dataeps) \sum_m \E_{X\sim P_{\theta^*}}\left[\pi_m(X)\right] f(m) + \dataeps \sum_m \E_{X\sim \QX}\left[\pi_m(X)\right] f(m) \\
& \geq (1-\dataeps) \sum_m \E_{X\sim P_{\theta^*}}\left[\pi_m(X)\right] f(m) \\
& = (1-\dataeps) \sum_m \pi_m f(m) \\
& = (1-\dataeps) \E_{M\sim P^*_M}\left[f(M)\right] ,
\end{align*}
we get
\begin{align*}
    \E_{M\sim P^*_M} \left[\mathbb{D}^2\left(P_{\MMDinf}^{(M)},\PmargMthetastar\right)\right] & \leq \frac{1}{1-\dataeps} \E_{M \sim \PM}\left[\mathbb{D}^2\left(P_{\MMDinf}^{(M)},\PmargMthetastar\right)\right] \\ 
    & \leq \frac{24\dataeps^2}{1-\dataeps} + 16 \cdot \mathbb{E}_{M\sim P^*}\left[\frac{\mathbb{M}_{X\sim \Pthetastar} \left[ \pi_M(X) \right]}{\pi_M^{* 2}}\right] + \frac{64\dataeps^2}{(1-\dataeps)^2} \cdot \mathbb{E}_{M\sim P^*}\left[ \frac{1}{\pi_M^{* 2}}\right] ,
\end{align*}
which ends the proof.
\end{proof}

\begin{proof}[Proof of Theorem \ref{MMD_Robust_cont_finite}] The proof follows the same line as for the proof of Theorem \ref{MMD_Robust_cont}. Theorem \ref{thm_finite_main} (using the squared absolute mean deviation $\mathbb{M}$ instead of the variance $\Var$) leads to:
\begin{align*}
    \mathbb{D}\left(\Pthetastar,P_{\MMDn}\right) & \leq \mathbb{D}\left(\Pthetastar,\PX\right) + \mathbb{D}\left(\PX,\mathbb{P}_{X|M=0}\right) +     \mathbb{D}\left(\mathbb{P}_{X|M=0},P_n\right) + \mathbb{D}\left(P_n,P_{\MMDn}\right) \\
    & \leq \mathbb{D}\left(\Pthetastar,\PX\right) + \mathbb{D}\left(\PX,\mathbb{P}_{X|M=0}\right) +     \mathbb{D}\left(\mathbb{P}_{X|M=0},P_n\right) + \mathbb{D}\left(P_n,\Pthetastar\right) \\
    & \leq 2\mathbb{D}\left(\Pthetastar,\PX\right) + 2\mathbb{D}\left(\PX,\mathbb{P}_{X|M=0}\right) +     2\mathbb{D}\left(\mathbb{P}_{X|M=0},P_n\right) .
\end{align*}
The first term is controlled as follows:
\begin{align*}
    \mathbb{D}\left(\Pthetastar,\PX\right) & = \mathbb{D}\left(\Pthetastar,(1-\dataeps)\Pthetastar+\dataeps\QX\right) \\
    & = \big\lVert \E_{X\sim P_{\theta^*}} \left[ k(X,\cdot) \right] - \E_{X\sim (1-\dataeps)P_{\theta^*} + \dataeps \QX} \left[ k(X,\cdot) \right] \big\rVert_{\mathcal{H}} \\
    & = \bigg\lVert \E_{X\sim P_{\theta^*}} \left[ k(X,\cdot) \right] - \bigg\{ (1-\dataeps) \E_{X\sim P_{\theta^*}} \left[ k(X,\cdot) \right] + \dataeps \E_{X\sim \QX} \left[ k(X,\cdot) \right] \bigg\} \bigg\rVert_{\mathcal{H}} \\
    & = \dataeps \cdot \big\lVert \E_{X\sim P_{\theta^*}} \left[ k(X,\cdot) \right] - \E_{X\sim \QX} \left[ k(X,\cdot) \right] \big\rVert_{\mathcal{H}}\\
    & \leq \dataeps \cdot \left(\big\lVert \E_{X\sim P_{\theta^*}} \left[ k(X,\cdot) \right] \big\rVert_{\mathcal{H}} + \big\lVert \E_{X\sim \QX} \left[ k(X,\cdot) \right] \big\rVert_{\mathcal{H}}\right) \\
    & \leq 2 \dataeps ,
\end{align*}
the second term as follows:
\begin{align*}
    &\mathbb{D}\left(\PX,\mathbb{P}_{X|M=0}\right) \\
    & \leq \frac{\mathbb{E}_{X\sim \PX}\left[\left|\pi(X)-\mathbb{E}_{X\sim \PX}[\pi(X)]\right|\right]}{\mathbb{E}_{X\sim \PX}\left[\pi(X)\right]} \\
    & = \frac{\mathbb{E}_{X\sim \PX}\left[\left|\pi(X)-\mathbb{E}_{X\sim \Pthetastar}[\pi(X)]+\dataeps\{\mathbb{E}_{X\sim \Pthetastar}[\pi(X)]-\mathbb{E}_{X\sim \QX}[\pi(X)]\}\right|\right]}{\mathbb{E}_{X\sim \PX}\left[\pi(X)\right]} \\
    & \leq \frac{\mathbb{E}_{X\sim \PX}\left[\left|\pi(X)-\mathbb{E}_{X\sim \Pthetastar}[\pi(X)]\right|\right]}{\mathbb{E}_{X\sim \PX}\left[\pi(X)\right]} + \dataeps \cdot \frac{\left|\mathbb{E}_{X\sim \Pthetastar}[\pi(X)]-\mathbb{E}_{X\sim \QX}[\pi(X)]\right|}{\mathbb{E}_{X\sim \PX}\left[\pi(X)\right]} \\
    & \leq \frac{\mathbb{E}_{X\sim \PX}\left[\left|\pi(X)-\mathbb{E}_{X\sim \Pthetastar}[\pi(X)]\right|\right]}{\mathbb{E}_{X\sim \PX}\left[\pi(X)\right]} + \frac{2\dataeps}{\mathbb{E}_{X\sim \PX}\left[\pi(X)\right]} \\
    & = \frac{(1-\dataeps)\mathbb{E}_{X\sim \Pthetastar}\left[\left|\pi(X)-\mathbb{E}_{X\sim \Pthetastar}[\pi(X)]\right|\right]}{\mathbb{E}_{X\sim \PX}\left[\pi(X)\right]} + \frac{\dataeps\mathbb{E}_{X\sim \QX}\left[\left|\pi(X)-\mathbb{E}_{X\sim \Pthetastar}[\pi(X)]\right|\right]}{\mathbb{E}_{X\sim \PX}\left[\pi(X)\right]} + \frac{2\dataeps}{\mathbb{E}_{X\sim \PX}\left[\pi(X)\right]} \\
    & \leq \frac{(1-\dataeps)\mathbb{E}_{X\sim \Pthetastar}\left[\left|\pi(X)-\mathbb{E}_{X\sim \Pthetastar}[\pi(X)]\right|\right]}{(1-\dataeps)\mathbb{E}_{X\sim \Pthetastar}\left[\pi(X)\right]} + \frac{\dataeps\mathbb{E}_{X\sim \QX}\left[\left|\pi(X)-\mathbb{E}_{X\sim \Pthetastar}[\pi(X)]\right|\right]}{(1-\dataeps)\mathbb{E}_{X\sim \Pthetastar}\left[\pi(X)\right]} + \frac{2\dataeps}{(1-\dataeps)\mathbb{E}_{X\sim \Pthetastar}\left[\pi(X)\right]} \\
    & = \frac{\mathbb{E}_{X\sim \Pthetastar}\left[\left|\pi(X)-\mathbb{E}_{X\sim \Pthetastar}[\pi(X)]\right|\right]}{\mathbb{E}_{X\sim \Pthetastar}\left[\pi(X)\right]} + \frac{\dataeps}{1-\dataeps} \cdot \frac{\mathbb{E}_{X\sim \QX}\left[\left|\pi(X)-\mathbb{E}_{X\sim \Pthetastar}[\pi(X)]\right|\right]}{\mathbb{E}_{X\sim \Pthetastar}\left[\pi(X)\right]} + \frac{2\dataeps}{(1-\dataeps)\mathbb{E}_{X\sim \Pthetastar}\left[\pi(X)\right]} \\
    & \leq \frac{\mathbb{E}_{X\sim \Pthetastar}\left[\left|\pi(X)-\mathbb{E}_{X\sim \Pthetastar}[\pi(X)]\right|\right]}{\mathbb{E}_{X\sim \Pthetastar}\left[\pi(X)\right]} + \frac{2\dataeps}{1-\dataeps} \cdot \frac{1}{\mathbb{E}_{X\sim \Pthetastar}\left[\pi(X)\right]} + \frac{2\dataeps}{1-\dataeps} \cdot \frac{1}{\mathbb{E}_{X\sim \Pthetastar}\left[\pi(X)\right]} \\
    & = \frac{\mathbb{E}_{X\sim \Pthetastar}\left[\left|\pi(X)-\mathbb{E}_{X\sim \Pthetastar}[\pi(X)]\right|\right]}{\mathbb{E}_{X\sim \Pthetastar}\left[\pi(X)\right]} + \frac{4\dataeps}{1-\dataeps} \cdot \frac{1}{\mathbb{E}_{X\sim \Pthetastar}\left[\pi(X)\right]} ,
\end{align*}
and the third one in expectation as follows:
\begin{align*}
    \E_{\mathcal{S}}
    \left[\mathbb{D}\left(\mathbb{P}_{X|M=0},P_n\right)\right] & \leq \sqrt{\E_{\mathcal{S}}\left[\mathbb{D}^2\left(\mathbb{P}_{X|M=0},P_n\right)\right]} \\
    & \leq \sqrt{\E_{\{M_i\}_{1\leq i \leq n}\sim\PM} \left[ \frac{1}{|\{i:M_i=0\}|} \right]} \\
    & \leq \sqrt{\frac{2}{n\cdot\mathbb{E}_{X\sim \PX}\left[\pi(X)\right]}} \\
    & \leq \sqrt{\frac{2}{n\cdot(1-\dataeps)\mathbb{E}_{X\sim \Pthetastar}\left[\pi(X)\right]}} , 
\end{align*}
so that in the end:
\begin{align*}
    &\E_{\mathcal{S}}\left[\mathbb{D}\left(\Pthetastar,P_{\MMDn}\right)\right] \\
    &  \leq 2\mathbb{D}\left(\Pthetastar,\PX\right) + 2\mathbb{D}\left(\PX,\mathbb{P}_{X|M=0}\right) +     2\mathbb{D}\left(\mathbb{P}_{X|M=0},P_n\right) \\
    & \leq 4\dataeps + \frac{2\cdot\mathbb{E}_{X\sim \Pthetastar}\left[\left|\pi(X)-\mathbb{E}_{X\sim \Pthetastar}[\pi(X)]\right|\right]}{\mathbb{E}_{X\sim \Pthetastar}\left[\pi(X)\right]} + \frac{8\dataeps}{1-\dataeps} \cdot \frac{1}{\mathbb{E}_{X\sim \Pthetastar}\left[\pi(X)\right]} + 2\sqrt{\frac{2}{n\cdot(1-\dataeps)\mathbb{E}_{X\sim \Pthetastar}\left[\pi(X)\right]}} \\
    & \leq 4\dataeps + \frac{2\sqrt{\Var_{X\sim \Pthetastar}\left[\pi(X)\right]}}{\pi^*} + \frac{8\dataeps}{1-\dataeps} \cdot \frac{1}{\pi^*} + 2\sqrt{\frac{2}{(1-\dataeps)n\pi^*}} .
\end{align*}
\end{proof}

\subsection{Proofs from Section \ref{Sec_exm}}

\begin{proof}[Proof of Corollary \ref{cor_trunc}]
Under the setting adopted in Example \ref{exm_trunc}, a finer version of Theorem \ref{MMD_Robust} (applying the triangle inequality to the MMD directly rather than to the squared MMD, and using the absolute mean deviation instead of the variance) simply leads to
$$
\mathbb{D}\left(P_{\MMDinf},P_{\theta^*}\right) \leq 2 \cdot \frac{\E_{X\sim \PX} \left[ \left|\mathbb{1}(X\in\mathcal{S})-\mathbb{P}[X\in\mathcal{S}]\right| \right]}{\mathbb{P}[X\in\mathcal{S}]} = 2 \cdot \frac{ 2 \mathbb{P}[X\in\mathcal{S}] (1-\mathbb{P}[X\in\mathcal{S}])}{\mathbb{P}[X\in\mathcal{S}]} = 4 \missingeps .
$$
\end{proof}

\begin{proof}[Proof of Theorem \ref{cor_huber}]
With the notation $\widetilde{\pi}_m(X)=\mathbb{Q}[M=m|X]$, a finer application of Theorem \ref{MMD_Robust} gives 
\begin{align*}
&\E_{M \sim \PM}\left[\mathbb{D}^2\left(\mathbb{P}_{X}^{(M)},P_{\MMDinf}^{(M)}\right)\right]\\
& \leq 2\cdot \inf_{\theta\in\Theta}\E_{M \sim \PM}\left[\mathbb{D}^2\left(P_{\theta}^{(M)},\PmargMX\right)\right] + 4 \cdot \E_{M \sim \PM}\left[ \frac{\Var_{X\sim \PX} \left[ \pi_M(X) \right]}{\E_{X\sim \PX} \left[\pi_M(X)\right]^2} \right] \\
& = 2\cdot \inf_{\theta\in\Theta}\E_{M \sim \PM}\left[\mathbb{D}^2\left(P_{\theta}^{(M)},\PmargMX\right)\right] + 4 \cdot \E_{M \sim \PM}\left[ \frac{\missingeps^2\cdot\Var_{X\sim \PX} \left[ \widetilde{\pi}_M(X) \right]}{\E_{X\sim \PX} \left[\pi_M(X)\right]^2} \right] \\
& \leq 2\cdot \inf_{\theta\in\Theta}\E_{M \sim \PM}\left[\mathbb{D}^2\left(P_{\theta}^{(M)},\PmargMX\right)\right] + 4 \missingeps^2 \cdot \E_{M \sim \PM}\left[\frac{1}{\E_{X\sim \PX} \left[\pi_M(X)\right]^2}\right] \\
& = 2\cdot \inf_{\theta\in\Theta}\E_{M \sim \PM}\left[\mathbb{D}^2\left(P_{\theta}^{(M)},\PmargMX\right)\right] + 4 \missingeps^2 \cdot \sum_m \E_{X\sim \PX} \left[\pi_m(X)\right] \frac{1}{\E_{X\sim \PX} \left[\pi_m(X)\right]^2} \\
& = 2\cdot \inf_{\theta\in\Theta}\E_{M \sim \PM}\left[\mathbb{D}^2\left(P_{\theta}^{(M)},\PmargMX\right)\right] + 4 \missingeps^2 \cdot \sum_m \frac{1}{\E_{X\sim \PX} \left[\pi_m(X)\right]} \\
& = 2\cdot \inf_{\theta\in\Theta}\E_{M \sim \PM}\left[\mathbb{D}^2\left(P_{\theta}^{(M)},\PmargMX\right)\right] + 4 \missingeps^2 \cdot \sum_m \frac{1}{(1-\missingeps)\alpha_m+\missingeps\cdot\E_{X\sim \PX} \left[\widetilde{\pi}_m(X)\right]} \\
& \leq 2\cdot \inf_{\theta\in\Theta}\E_{M \sim \PM}\left[\mathbb{D}^2\left(P_{\theta}^{(M)},\PmargMX\right)\right] + 4 \missingeps^2 \cdot \sum_m \frac{1}{(1-\missingeps)\alpha_m} \\
& = 2\cdot \inf_{\theta\in\Theta}\E_{M \sim \PM}\left[\mathbb{D}^2\left(P_{\theta}^{(M)},\PmargMX\right)\right] + \frac{4\missingeps^2}{1-\missingeps} \cdot \E_{M \sim (\alpha_m)}\left[\frac{1}{\alpha_M^2}\right] .
\end{align*}
Furthermore, we have for any pattern $m$:
\begin{align*}
    \mathbb{D}\left(P_{\theta^*}^{(m)},\mathbb{P}_{X}^{(m)}\right) & = \mathbb{D}\left(P_{\theta^*}^{(m)},(1-\dataeps)P_{\theta^*}^{(m)}+\dataeps\mathbb{Q}_X^{(m)}\right) \\
    & = \left\lVert \E_{X\sim\Pthetastar} \left[ k(X^{(m)},\cdot) \right] - \E_{X\sim(1-\dataeps)\Pthetastar+\dataeps\QX} \left[ k(X^{(m)},\cdot) \right] \right\rVert_{\mathcal{H}} \\
    & = \left\lVert \E_{X\sim\Pthetastar} \left[ k(X^{(m)},\cdot) \right] - \left\{(1-\dataeps)\E_{X\sim\Pthetastar} \left[ k(X^{(m)},\cdot) \right] + \dataeps \E_{X\sim\QX} \left[ k(X^{(m)},\cdot) \right] \right\} \right\rVert_{\mathcal{H}} \\
    & = \dataeps \cdot \left\lVert \E_{X\sim\Pthetastar} \left[ k(X^{(m)},\cdot) \right] -  \E_{X\sim\QX} \left[ k(X^{(m)},\cdot) \right] \right\rVert_{\mathcal{H}} \\
    & \leq \dataeps \cdot \left\{ \left\lVert \E_{X\sim\Pthetastar} \left[ k(X^{(m)},\cdot) \right] \right\rVert_{\mathcal{H}} + \left\lVert \E_{X\sim\QX} \left[ k(X^{(m)},\cdot) \right] \right\rVert_{\mathcal{H}} \right\} \\
    & \leq 2 \dataeps ,
\end{align*}
and thus
\begin{align*}
\E_{M \sim \PM}\left[\mathbb{D}^2\left(P_{\theta^*}^{(M)},P_{\MMDinf}^{(M)}\right)\right] & \leq 2 \cdot \E_{M \sim \PM}\left[\mathbb{D}^2\left(P_{\theta^*}^{(M)},\mathbb{P}_{X}^{(M)}\right)\right] + 2 \cdot \E_{M \sim \PM}\left[\mathbb{D}^2\left(\mathbb{P}_{X}^{(M)},P_{\MMDinf}^{(M)}\right)\right] \\
& \leq 8\dataeps^2 + 4\cdot \inf_{\theta\in\Theta}\E_{M \sim \PM}\left[\mathbb{D}^2\left(P_{\theta}^{(M)},\PmargMX\right)\right] + \frac{8\missingeps^2}{1-\missingeps} \cdot \E_{M \sim (\alpha_m)}\left[\frac{1}{\alpha_M^2}\right] \\
& \leq 8\dataeps^2 + 4\cdot\E_{M \sim \PM}\left[\mathbb{D}^2\left(P_{\theta^*}^{(M)},\PmargMX\right)\right] + \frac{8\missingeps^2}{1-\missingeps} \cdot \E_{M \sim (\alpha_m)}\left[\frac{1}{\alpha_M^2}\right] \\
& \leq 24\dataeps^2 + \frac{8\missingeps^2}{1-\missingeps} \cdot \E_{M \sim (\alpha_m)}\left[\frac{1}{\alpha_M^2}\right] .
\end{align*}
Finally, using the identity
\begin{align*}
\E_{M\sim\PM}\left[f(M)\right] & = \sum_m \E_{X\sim\PX}\left[\pi_m(X)\right] f(m) \\
& = \sum_m \E_{X\sim\PX}\left[(1-\missingeps)\alpha_m+\missingeps\widetilde{\pi}(X)\right] f(m) \\
& = (1-\missingeps) \sum_m \alpha_m f(m) + \missingeps \sum_m \E_{X\sim\PX}\left[\widetilde{\pi}(X)\right] f(m) \\
& \geq (1-\missingeps) \sum_m \alpha_m f(m) \\
& = (1-\missingeps) \sum_m \alpha_m f(m) \\
& = (1-\missingeps) \E_{M\sim(\alpha_m)}\left[f(M)\right] ,
\end{align*}
we have
\begin{align*}
\E_{M\sim(\alpha_m)}\left[\mathbb{D}^2\left(P_{\theta^*}^{(M)},P_{\MMDinf}^{(M)}\right)\right] & \leq \frac{1}{1-\missingeps} \cdot \E_{M \sim \PM}\left[\mathbb{D}^2\left(P_{\theta^*}^{(M)},P_{\MMDinf}^{(M)}\right)\right] \\
& \leq \frac{24\dataeps^2}{1-\missingeps} + \frac{8\missingeps^2}{(1-\missingeps)^2} \cdot \E_{M \sim (\alpha_m)}\left[\frac{1}{\alpha_M^2}\right] .
\end{align*}
\end{proof}

\begin{proof}[Proof of Theorem \ref{cor_huber_finite}] This is a straightforward application of a refined version of Theorem \ref{MMD_Robust_cont_finite}. Remind that Theorem \ref{MMD_Robust_cont_finite} is:
$$
\E_{\mathcal{S}}\left[ \mathbb{D}\left(P_{\MMDn},\Pthetastar\right) \right] \leq 4\cdot\dataeps + \frac{8\cdot\dataeps}{\E_{X\sim\Pthetastar}[\pi(X)](1-\dataeps)} + \frac{2\sqrt{\Var_{X\sim \Pthetastar}\left[\pi(X)\right]}}{\E_{X\sim\Pthetastar}[\pi(X)]} + \frac{2\sqrt{2}}{\sqrt{n\E_{X\sim\Pthetastar}[\pi(X)](1-\dataeps)}} ,
$$
where the quantities $\E_{X\sim\Pthetastar}[\pi(X)](1-\dataeps)$ in the denominator come as rough lower bounds on $\E_{X\sim\PX}[\pi(X)]$. Hence, using those quantities directly as established in the proof of Theorem \ref{MMD_Robust_cont_finite}, we have:
$$
\E_{\mathcal{S}}\left[ \mathbb{D}\left(P_{\MMDn},\Pthetastar\right) \right] \leq 4\cdot\dataeps + \frac{2\sqrt{\Var_{X\sim \PX}\left[\pi(X)\right]}}{\E_{X\sim\PX}[\pi(X)]} + \frac{2\sqrt{2}}{\sqrt{n\E_{X\sim\PX}[\pi(X)]}} .
$$
The arguments $\E_{X\sim\PX}[\pi(X)]\geq\alpha(1-\missingeps)$, $\E_{X\sim\Pthetastar}[\pi(X)]\geq\alpha(1-\missingeps)$ and $\sqrt{\Var_{X\sim\PX}[\pi(X)]}\leq\missingeps$ where $\pi(X) = (1-\missingeps) \cdot \alpha + \missingeps \cdot \mathbb{Q}[M=m|X]$ then lead to:
$$
\E_{\mathcal{S}}\left[ \mathbb{D}\left(P_{\MMDn},\Pthetastar\right) \right] \leq 4\cdot\dataeps + \frac{2\missingeps}{\alpha(1-\missingeps)} + \frac{2\sqrt{2}}{\sqrt{n\alpha(1-\missingeps)}} .
$$
\end{proof}


\begin{proof}[Proof of Theorem \ref{cor_adver}] First observe that by definition:
$$
\MMDn = \argmin_{\theta\in\Theta} \mathbb{D}\left(\widetilde{P}_n,P_{\theta}\right) \quad \textnormal{where} \quad \widetilde{P}_n = \frac{1}{|\{i:\widetilde{M}_i=0\}|} \sum_{i:\widetilde{M}_i=0} \delta_{\left\{X_i\right\}} .
$$
With the properties $\big| |\{i:\widetilde{M}_i=0\}|-|\{i:M_i=0\}| \big| \leq \missingeps/\alpha\cdot|\{i:M_i=0\}|$, $|\{i:\widetilde{M}_i=0\}|\geq(1-\missingeps/\alpha)\cdot|\{i:M_i=0\}|$, $|\{i:M_i=0,\widetilde{M}_i=1\}|\leq\missingeps/\alpha\cdot|\{i:M_i=0\}|$ and $|\{i:M_i=1,\widetilde{M}_i=0\}|\leq\missingeps/\alpha\cdot|\{i:M_i=0\}|$, we have:
\begin{align*}
    \mathbb{D}\left(P_n,\widetilde{P}_n\right) & = \left\lVert \frac{1}{|\{i:M_i=0\}|} \sum_{i:M_i=0} \Phi(X_i) - \frac{1}{|\{i:\widetilde{M}_i=0\}|} \sum_{i:\widetilde{M}_i=0} \Phi(X_i) \right\rVert_{\mathcal{H}} \\
    & = \left\lVert \sum_{i=1}^n \left( \frac{\mathbbm{1}(M_i=0)}{|\{i:M_i=0\}|} - \frac{\mathbbm{1}(\widetilde{M}_i=0)}{|\{i:\widetilde{M}_i=0\}|} \right) \Phi(X_i) \right\rVert_{\mathcal{H}} \\
    & \leq \sum_{i=1}^n \left| \frac{\mathbbm{1}(M_i=0)}{|\{i:M_i=0\}|} - \frac{\mathbbm{1}(\widetilde{M}_i=0)}{|\{i:\widetilde{M}_i=0\}|} \right| \left\lVert \Phi(X_i) \right\rVert_{\mathcal{H}} \\
    & \leq \sum_{i=1}^n \left| \frac{\mathbbm{1}(M_i=0)}{|\{i:M_i=0\}|} - \frac{\mathbbm{1}(\widetilde{M}_i=0)}{|\{i:\widetilde{M}_i=0\}|} \right| \\
    & = \sum_{i=1}^n \frac{ \left| \mathbbm{1}(M_i=0)|\{i:\widetilde{M}_i=0\}|-\mathbbm{1}(\widetilde{M}_i=0)|\{i:M_i=0\}| \right| }{|\{i:M_i=0\}|\cdot|\{i:\widetilde{M}_i=0\}|} \\
    & = \sum_{i:M_i=0,\widetilde{M}_i=1} \frac{|\{i:\widetilde{M}_i=0\}|}{|\{i:M_i=0\}|\cdot|\{i:\widetilde{M}_i=0\}|} + \sum_{i:M_i=1,\widetilde{M}_i=0} \frac{|\{i:M_i=0\}|}{|\{i:M_i=0\}|\cdot|\{i:\widetilde{M}_i=0\}|} \\
    & \quad \quad \quad \quad \quad \quad \quad \quad \quad \quad \quad \quad \quad \quad \quad \quad \quad \quad + \sum_{i:M_i=\widetilde{M}_i=0} \frac{ \left| |\{i:\widetilde{M}_i=0\}|-|\{i:M_i=0\}| \right| }{|\{i:M_i=0\}|\cdot|\{i:\widetilde{M}_i=0\}|} \\
    & = \frac{|\{i:M_i=0,\widetilde{M}_i=1\}|}{|\{i:M_i=0\}|} + \frac{|\{i:M_i=1,\widetilde{M}_i=0\}|}{|\{i:\widetilde{M}_i=0\}|} \\
    & \quad \quad \quad \quad \quad \quad \quad \quad \quad \quad \quad \quad + \frac{ |\{i:M_i=\widetilde{M}_i=0\}| \cdot \left| |\{i:\widetilde{M}_i=0\}|-|\{i:M_i=0\}| \right| }{|\{i:M_i=0\}|\cdot|\{i:\widetilde{M}_i=0\}|} \\
    & \leq \frac{|\{i:M_i=0,\widetilde{M}_i=1\}|}{|\{i:M_i=0\}|} + \frac{|\{i:M_i=1,\widetilde{M}_i=0\}|}{|\{i:\widetilde{M}_i=0\}|} + \frac{ |\{i:M_i=\widetilde{M}_i=0\}| \cdot \missingeps/\alpha\cdot|\{i:M_i=0\}| }{|\{i:M_i=0\}|\cdot|\{i:\widetilde{M}_i=0\}|} \\
    & = \frac{|\{i:M_i=0,\widetilde{M}_i=1\}|}{|\{i:M_i=0\}|} + \frac{|\{i:M_i=1,\widetilde{M}_i=0\}|}{|\{i:\widetilde{M}_i=0\}|} + \frac{ |\{i:M_i=\widetilde{M}_i=0\}| \cdot \missingeps/\alpha }{\{i:\widetilde{M}_i=0\}|} \\
    & \leq \frac{|\{i:M_i=0,\widetilde{M}_i=1\}|}{|\{i:M_i=0\}|} + \frac{|\{i:M_i=1,\widetilde{M}_i=0\}|}{(1-\missingeps/\alpha)\cdot|\{i:M_i=0\}|} + \frac{ |\{i:M_i=\widetilde{M}_i=0\}| \cdot \missingeps/\alpha }{(1-\missingeps/\alpha)\cdot|\{i:M_i=0\}|} \\
    & = \frac{(1-\missingeps/\alpha)\cdot|\{i:M_i=0,\widetilde{M}_i=1\}| + |\{i:M_i=1,\widetilde{M}_i=0\}| + |\{i:M_i=\widetilde{M}_i=0\}| \cdot \missingeps/\alpha }{(1-\missingeps/\alpha)\cdot|\{i:M_i=0\}|}  \\
    & \leq \frac{(1-\missingeps/\alpha)\cdot\missingeps/\alpha\cdot|\{i:M_i=0\}| + \missingeps/\alpha\cdot|\{i:M_i=0\}| + |\{i:M_i=0\}| \cdot \missingeps/\alpha }{(1-\missingeps/\alpha)\cdot|\{i:M_i=0\}|}  \\
    & = \frac{(1-\missingeps/\alpha)\cdot\missingeps/\alpha + \missingeps/\alpha + \missingeps/\alpha }{1-\missingeps/\alpha}  \\
    & = \frac{(3-\missingeps/\alpha)\cdot\missingeps/\alpha}{1-\missingeps/\alpha}  \\
    & \leq \frac{3\cdot\missingeps}{\alpha-\missingeps} .
\end{align*}
Once again, if $\PX=(1-\dataeps)\Pthetastar+\dataeps\QX$:
\begin{align*}
    \mathbb{D}\left(\Pthetastar,P_{\MMDn}\right) & \leq \mathbb{D}\left(\Pthetastar,\PX\right) + \mathbb{D}\left(\PX,\widetilde{P}_n\right) +     \mathbb{D}\left(\widetilde{P}_n,P_{\MMDn}\right)\\
    & \leq \mathbb{D}\left(\Pthetastar,\PX\right) + \mathbb{D}\left(\PX,\widetilde{P}_n\right) + \mathbb{D}\left(\widetilde{P}_n,\Pthetastar\right)\\
    & \leq 2\mathbb{D}\left(\Pthetastar,\PX\right) + 2\mathbb{D}\left(\PX,\widetilde{P}_n\right) \\
    & \leq 2\mathbb{D}\left(\Pthetastar,\PX\right) + 2\mathbb{D}\left(\PX,P_n\right) + 2\mathbb{D}\left(P_n,\widetilde{P}_n\right) \\
    & \leq 4\dataeps + 2\mathbb{D}\left(\PX,P_n\right) + 2\mathbb{D}\left(P_n,\widetilde{P}_n\right) .
\end{align*}
Since the base uncontaminated missingness mechanism is M(C)AR, $P_n = \frac{1}{|\{i:M_i=0\}|} \sum_{i:M_i=0} \delta_{\{X_i\}}$ is an empirical estimate of $\PX$ for which:
$$
\mathbb{E}_{\mathcal{S}}\left[\mathbb{D}^2\left(\PX,P_n\right)\right] \leq \E_{\{M_i\}_{1\leq i \leq n}\sim\PM} \left[ \frac{1}{|\{i:M_i=0\}|} \right] \leq \frac{2}{n\cdot\alpha} ,
$$
and so:
$$
\mathbb{E}_{\mathcal{S}} \left[\mathbb{D}\left(\Pthetastar,P_{\MMDn}\right) \right] \leq 4\dataeps + 2\mathbb{E}_{\mathcal{S}}\left[\mathbb{D}\left(\PX,P_n\right)\right] + 2\mathbb{E}_{\mathcal{S}}\left[\mathbb{D}\left(P_n,\widetilde{P}_n\right)\right] \leq 4\dataeps + 2\sqrt{\frac{2}{n\cdot\alpha}} + 2\cdot \frac{3\cdot\missingeps}{\alpha-\missingeps} .
$$
\end{proof}
\end{appendix}

\end{document}